%% file: main.tex
\documentclass[a4paper,11pt]{article}
\usepackage{packages}

\title{Mean-Field Dynamics of the Bose--Hubbard Model in High Dimension}

\author{
Shahnaz Farhat\thanks{School of Science, Constructor University Bremen, Campus Ring 1, 28759 Bremen, Germany; and Department of Mathematics, University of Tübingen, Auf der Morgenstelle 10, 72076 Tübingen, Germany. Email: \texttt{shahnaz.farhat@uni-tuebingen.de}}, 
Denis Périce\thanks{School of Science, Constructor University Bremen, Campus Ring 1, 28759 Bremen, Germany. Email: \texttt{dperice@constructor.university}, corresponding author.}, 
Sören Petrat\thanks{School of Science, Constructor University Bremen, Campus Ring 1, 28759 Bremen, Germany. Email: \texttt{spetrat@constructor.university}}
}

\begin{document}

    \maketitle

    \begin{abstract}
        This paper is a first step toward rigorously establishing the validity of mean-field approximations for the dynamics of quantum systems in a large dimensional limit. Our result is a derivation of the large $d$ mean-field limit of the Bose--Hubbard model on a square lattice. The Bose--Hubbard model effectively describes bosons on a lattice with on-site interactions and nearest-neighbour hopping, serving as a foundational framework for understanding strong particle interactions and the superfluid to Mott insulator phase transition. We prove a trace norm estimate between the one-lattice-site reduced density of the Schrödinger dynamics and the mean-field dynamics in the limit of large dimension. Here, the mean-field approximation is in the hopping amplitude and not in the interaction, leading to a very rich and non-trivial mean-field equation. This mean-field equation does not only describe the condensate, as is the case when the mean-field description comes from a large particle number limit averaging out the interaction, but it allows for a phase transition since it contains the full non-trivial interaction. Our work is a rigorous justification of a simple case of the highly successful dynamical mean-field theory (DMFT) for bosons, which somewhat surprisingly yields many qualitatively correct results in three dimensions.
    \end{abstract}
    
    \input{paper}

\section*{Declarations}
\addcontentsline{toc}{section}{Declarations}

\noindent\textbf{Fundings and acknowledgments.} S.~Farhat and S.~Petrat acknowledge funding by the
Deutsche Forschungsgemeinschaft (DFG, German Research Foundation) - project number
505496137. D.~Périce and S.~Petrat acknowledge funding by the
Deutsche Forschungsgemeinschaft (DFG, German Research Foundation) - project number
512258249.

The authors would like to thank the Institut Henri Poincaré (UAR 839 CNRS-Sorbonne Université) and the LabEx CARMIN (ANR-10-LABX-59-01) for their support, particularly for hosting us during the thematic program on "Quantum many-body systems out-of-equilibrium", where the idea that initiated this project was born. Special thanks go to one of its organizers, Thierry Giamarchi, for insightful discussions on quantum lattice systems and DMFT.

\vspace{3mm}
\noindent\textbf{Data availability statement.} No datasets were generated or analysed during the current study.

\vspace{3mm}
\noindent\textbf{Conflicts of interests declaration.} The authors have no competing interests to declare that are relevant to the content of this article.

\addcontentsline{toc}{section}{References}
\bibliography{bibliography}

\end{document}

%% file: paper.tex
\section{Introduction}

    One of the big aspirations of mathematical physics is to advance our rigorous understanding of phase transitions. Within this research area, much recent attention has been paid to the phenomenon of Bose--Einstein--Condensation (BEC), a phase of matter of cold Bose gases that was predicted in 1924 by Bose \cite{Bose_original} and Einstein \cite{Einstein_original_1,Einstein_original_2}. Since then, BEC has been studied extensively by theoretical physicists, and at least since the 1980s also by mathematical physicists with more rigorous methods. After the first experimental realizations in the labs of Cornell/Wieman \cite{anderson1995} and Ketterle \cite{davis1995} in 1995, the study of BEC has received a new wave of attention throughout experimental, theoretical, and mathematical physics. The motivation of our work comes from different perspectives.
    \begin{enumerate}
        \item\label{BEC_no_phase_transition} Lots of recent effort has been put into understanding BEC at zero temperature, e.g., in the Gross--Pitaevskii limit (see, e.g., \cite{schlein_2022_review} for a review) and the thermodynamic limit (see, e.g, \cite{yau_yin,Fournais_Solovej_1,Fournais_Solovej_2} for a derivation of the Lee--Huang--Yang formula). This yields insight into the behavior of the condensate and its excitations, e.g., a rigorous proof of Bogoliubov theory in the Gross--Pitaevskii regime \cite{boccato2018}, but it is very far from understanding the (thermodynamic) phase transition to BEC.

        \item\label{mean-field_no_phase_transition} A very useful and simple method for studying many-body systems is the mean-field approximation. For bosonic cold atoms, the interaction potential is scaled down with the inverse of the particle number $N$ \cite{erdos2001} (or density \cite{derezinski2014}), thus considering weak interaction. Then the inverse particle number can be regarded as a small parameter, and the interaction can be effectively replaced by its mean-field. In this mean-field limit, many physical effects can be rigorously established regarding the dynamics and the low-energy properties (see, e.g., \cite{QF,spectrum} for a list of references). In particular, one can prove the validity of Bogoliubov theory (see, e.g., \cite{lewin2015,lewin2015_2}), and perturbative expansions beyond Bogoliubov \cite{QF,spectrum}. However, these types of mean-field models do not describe phase transitions.
        \item\label{TD_vs_quantum_phase_transition} For Bose gases in the continuum, one would ultimately like to prove a thermodynamic phase transition. However, Bose gases on the lattice offer a different possibility for a phase transition, namely a quantum phase transition between a BEC and a localized state usually called a Mott insulator; see, e.g., \cite{Zwerger_2003} for a review. There are only a few mathematically rigorous works on this topic, for example \cite{ALSSY}, where the quantum transition to BEC is proved at half filling, and \cite{MottTransition}, where the authors establish $d$-dependent bounds on the critical line delimiting the Mott insulator phase.
    \end{enumerate}
    Our work addresses these points in the following way. \ref{BEC_no_phase_transition}: We study a limit that may describe a phase transition. \ref{mean-field_no_phase_transition}: Our limit is a mean-field limit, not for large particle number but for large dimension. We hope and indeed show that some of the methods of large $N$ mean-field limits are still relevant for this case. Since in our model the averaging is done over the hopping terms, and the interaction is treated non-perturbatively, our mean-field model is strongly interacting. \ref{TD_vs_quantum_phase_transition}: Our microscopic model is the Bose--Hubbard model, which is a lattice model that has been successfully used to describe the BEC-Mott transition.

    Our main result is convergence of the reduced one-lattice-site density matrix of the many-body Schrödinger dynamics to the mean-field dynamics, with an error bound that goes to zero as the dimension $d$ goes to infinity. Other parameters such as the density and the coupling remain fixed. Thus, we rigorously justify the validity of the mean-field approximation for a quantum system in large dimension. We choose the Bose--Hubbard model to illustrate this statement both for its remarkable usefulness in physics and for the technical simplifications it offers as a lattice model. The Bose--Hubbard model is a popular model used to describe bosons on a lattice with on-site interactions, allowing hopping between nearest-neighbor lattice sites. It is well known for capturing strong interactions between particles \cite{Bloch2007ManyBodyPW} and providing one of the simplest descriptions of the Mott transition to date; see \cite{PhysRevB.40.546} and later \cite{PhysRevLett.81.3108}, see also \cite{Greiner2002QuantumPT,GREINER200311}. 
    
    A common technique to study models such as Bose--Hubbard is Dynamical Mean-field Theory (DMFT). This theory is well known for its description of the Mott insulator/superfluid phase transition \cite{PhysRevB.80.245110,PhysRevB.40.546}. It is usually formulated via Green's function, using a self-consistency condition to reduce the problem on the full lattice to an effective single-site problem. DMFT is typically justified in the physics literature by stating that mean-field theories become exact in the limit of infinite dimensions \cite{DMFT-22}; see also \cite{PhysRevLett.62.324} for fermions. A remarkable fact is that DMFT tends to provide qualitatively good results in three dimensions already \cite{DMFT-96}. In the literature, the effective equation we are deriving here is often called the mean-field model of Fisher et al. \cite{boson_DMFT}, referring to \cite{PhysRevB.40.546}. A more involved mean-field type equation is obtained in \cite{boson_DMFT} by scaling different parts of the hopping term in different ways. Compared to this scaling and the fermionic case, our equation can hence be considered as a simple case of DMFT. 
    
    Let us also mention the result \cite{PhysRevB.40.546}, where the authors consider the Bose--Hubbard model on a complete graph (the hopping term is of equal strength between all vertices). In comparison, our model has only nearest-neighbor hopping. Rigorous justifications for the effective thermodynamic behavior of the Bose--Hubbard model on a complete graph were obtained in \cite{Bru_Dorlas}. As mentioned above, in the mathematical literature, mean-field limits are typically considered as many-particle limits for the Bose--Hubbard model \cite{MFBH} or, more generally, for continuous models where the Hartree equation is obtained as effective dynamics (see, e.g., \cite{benedikter_lec,Golse_2016} for reviews). This approach requires dividing the interaction term by the number of particles to ensure that the kinetic energy and the interaction energy are of the same order.

    Our goal is to provide a rigorous justification, in the \(d \to \infty\) limit, that DMFT is a good approximation of the Schrödinger equation in the context of the Bose--Hubbard model. Here the dimension $d$ is the mean-field parameter that plays the usual role of the particle number $N$ as we average over the $d$ directions instead of the $N$ particles. It is interesting to note that, in the large \(d\) limit, the roles of the kinetic energy and the interaction between particles are reversed compared to the mean-field limit \(N \to \infty\). The terms we aim to average in our regime are the hopping terms between nearest-neighbor sites which come from the kinetic energy. Since we only consider on-site interactions, the interaction between particles acts as a one-site operator and therefore does not contribute to correlations between two different lattice sites. For our setting, the basic idea behind the mean-field approximation is that the coordination number of the lattice (the number of nearest neighbors) increases with the dimension. This means that we have a mean-field picture locally around every site, which allows us to control the correlations between sites. Note that our main estimates are for the reduced one-lattice-site density matrix, and not for the one-particle reduced density matrix that is usually used to describe convergence in the large $N$ limit.

\section{The Model and Main Results}

\subsection{Model}

    We consider the $d$-dimensional square lattice with periodic boundary conditions $\Lambda \coloneq \prth{\mathbb{Z}\slash L\mathbb{Z}}^d$ of volume $|\Lambda| \coloneq L^d$, with $L \in \mathbb{N}, L \ge 2$. We write $x \sim y$ if $x, y \in \Lambda$ are nearest neighbors. The one-site Hilbert space is $\ell^2(\C)$ and its canonical Hilbert basis is denoted by $(\ket{n})_{n\in\mathbb{N}}$. We define the standard creation and annihilation operators $a^*$, $a$ satisfying the CCR $[a,a^*]=1$, $[a,a] = 0 = [a^*,a^*]$; explicitly,
    \begin{align*}
        &a\ket{0} \coloneq 0, a \ket{n} \coloneq \sqrt{n} \ket{n-1}  \forall n \in \mathbb{N}^*, \\
        &a^* \ket{n} \coloneq \sqrt{n+1}\ket{n+1} \forall n \in \mathbb{N}.
    \end{align*}
    The number operator is given as $\mathcal{N} \coloneq a^* a$. To simplify our notation in some later proofs, we introduce an order on $\Lambda$ such that $\forall x \in \Lambda$
    \begin{align*}
        \#\sett{y \in \Lambda | y > x \text{ and } x \sim y} = \#\sett{y \in \Lambda | x > y \text{ and } x \sim y} = d.
    \end{align*}
    For example, the lexicographic order does the job. The Fock space is
    \begin{equation}\label{eq:Fock space}
        \mathcal{F} \coloneq \ell^2(\C)^{\otimes \abs{\Lambda}} \cong \mathcal{F}_+\prth{L^2(\Lambda, \C)} \coloneq \bigoplus_{k\in\mathbb{N}} L^2(\Lambda, \C)^{\otimes_+ k},
    \end{equation}
    where $\otimes_+$ denotes the symmetric tensor product. We emphasize that the Fock space is not symmetric with respect to lattice sites, as in general two lattice sites cannot be exchanged without changes to the lattice structure. It is, however, symmetric with respect to particle exchanges, as one can see from the isomorphism in \eqref{eq:Fock space}.
    
    Given a one-site operator $A$ and $x\in \Lambda$, we define
    \begin{align*}
        A_x \coloneq \mathds{1}^{\otimes \abs{\sett{y\in \Lambda, y<x}}} \otimes  A \otimes \mathds{1}^{\otimes \abs{\sett{y\in \Lambda, x<y}}}.
    \end{align*}
    In the following, we define $\nn{x}{y}$ to mean that $x, y \in \Lambda, x\sim y$ and $x < y$. The kinetic energy is given by the negative second quantized discrete Laplacian
    \begin{align*}
        - \dGamma(\Delta_d) 
            \coloneq&  \sum_{\nn{x}{y}} (a^*_x - a^*_y)(a_x - a_y) = - \sum_{\nn{x}{y}} \prth{a^*_x a_y + a^*_y a_x} + 2d \sum_{x\in \Lambda} \mathcal{N}_x.
    \end{align*}
    Furthermore, we denote by $\mathcal{N}_\mathcal{F}$ the number operator on Fock space, i.e., 
    \begin{align*}
        \mathcal{N}_\mathcal{F} \coloneq d\Gamma(1) := \sum_{x\in\Lambda} \mathcal{N}_x.
    \end{align*}
    Given hopping amplitude $J \in \R$, chemical potential $\mu \in \R$, and coupling constant $U\in \R$, we define the Bose--Hubbard Hamiltonian
    \begin{align}
        H_d \coloneq& - \frac{J}{2d} \dGamma(\Delta_d) - \mu \mathcal{N}_\mathcal{F} + \frac{U}{2} \sum_{x\in\Lambda} \mathcal{N}_x (\mathcal{N}_x - 1) \nonumber\\
            =& - \frac{J}{2d}\sum_{\nn{x}{y}} \prth{a^*_x a_y + a^*_y a_x} + (J - \mu) \sum_{x\in\Lambda} \mathcal{N}_x + \frac{U}{2} \sum_{x\in\Lambda} \mathcal{N}_x (\mathcal{N}_x - 1). \label{eq:HBH}
    \end{align}
    Here, we have scaled down the hopping term with the inverse of the dimension $d$.
    The time-dependent Schrödinger equation for $\Psi_d \in \mathcal{F}$ is
    \begin{equation}\label{schrodinger}
    i \partial_t \Psi_d(t) = H_d \Psi_d(t).
    \end{equation}

    The mean-field approximation amounts to a lattice site product state ansatz $\Psi_d \approx \prod_{x\in\Lambda} \varphi_x$ where $\varphi \in \ell^2(\C)$ is a one-lattice-site wave function. Such a state $\prod_{x\in\Lambda} \varphi_x$ is sometimes called Gutzwiller product state \cite{Gutzwiller,Rokhsar_Kotliar_1991}. Our main results state that if $\Psi_d(0) \approx \prod_{x\in\Lambda} \varphi_x(0)$, then also $\Psi_d(t) \approx \prod_{x\in\Lambda} \varphi_x(t)$ for all times $t >0$, where $\approx$ is meant in an appropriate reduced sense. We can guess the right mean-field equation by rewriting the Hamiltonian \eqref{eq:HBH} for any $\alpha \in \C$ as
    \begin{align}\label{eq:HBH_rewritten} 
    \begin{split}
    H_d &= -J \sum_{x\in\Lambda} \prth{\alpha a^*_x + \overline{\alpha} a_x - |\alpha|^2} + (J - \mu) \sum_{x\in\Lambda} \mathcal{N}_x + \frac{U}{2} \sum_{x\in\Lambda} \mathcal{N}_x (\mathcal{N}_x - 1) \\
    &\quad - \frac{J}{2d} \sum_{\nn{x}{y}} \Big((a_x^*-\overline{\alpha}) (a_y-\alpha) + (a^*_y-\overline{\alpha})(a_x-\alpha)\Big),   
    \end{split}
    \end{align}
    neglecting all terms that are quadratic in $a-\alpha$ (i.e., the last line in \eqref{eq:HBH_rewritten}), and choosing $\alpha$ as the averaged annihilation operator. Then the corresponding mean-field equation is
    \begin{equation}\label{eq:MFS}
        i \partial_t \varphi(t) = h^{\varphi(t)} \varphi(t)
    \end{equation}
    with nonlinear mean-field operator
    \begin{equation}
    h^{\varphi} \coloneq - J \big( \alpha_\varphi \ad + \overline{\alpha_\varphi} a - \abs{\alpha_\varphi}^2 \big) + (J - \mu) \mathcal{N} + \frac{U}{2} \mathcal{N}(\mathcal{N}-1), \label{eq:h mf}
    \end{equation}
    and with order parameter $\alpha_\varphi \coloneqq \scp{\varphi}{a \varphi}$. The role of $\alpha_\varphi$ as an order parameter in the Mott-insulator/superfluid phase transition is developed in Remark \ref{rmk:phase transition} below. For a discussion of the corresponding mean-field energy functional and convergence of the ground state energies we refer to \cite{GS_2026}.
    
    The well-posedness of \eqref{eq:MFS} is discussed in Section~\ref{șec_well_posedness}; see in particular Theorem~\ref{gwp}. We would like to emphasize the richness of this mean-field dynamics (for $U\neq 0$) as it may exhibit a phase transition. For example, it is possible to find an initial state in the Mott insulator phase, i.e., with $\varphi(0) \in \ell^2(\C)$ such that $\alpha_{\varphi}(0) = 0$, but $\alpha_{\varphi}'(0) > 0$, meaning that the state ends up in the superfluid phase ($\alpha_{\varphi}(t) \neq 0$) for a neighborhood of $t>0$.
    
    The approximation $\Psi_d \approx \prod_{x\in\Lambda} \varphi_x$ is not expected to hold in $\mathcal{F}$, but in the sense of reduced lattice site density matrices. Given $\Psi_d$, the corresponding density matrix is $\gamma_d \coloneq \ket{\Psi_d}\bra{\Psi_d} \in \mathcal{L}^1\prth{\mathcal{F}}$ and satisfies the von Neumann equation
    \begin{align}
        i\partial_t \gamma_d(t) = \sbra{H_d, \gamma_d(t)}. \label{eq:gammad dynamics}
    \end{align}
    For $\gamma_d \in \mathcal{L}^1\prth{\mathcal{F}}$ a positive trace one operator, we define its first reduced one-lattice-site density matrix as
    \begin{align*}
        \gOne \coloneq \frac{1}{\abs{\Lambda}} \sum_{x \in \Lambda} \PTr{\Lambda\backslash{\sett{x}}}{\gamma_d},
    \end{align*}
    where the partial trace is defined as follows. For all $\Gamma \subseteq \Lambda$, and for all $\prth{A^{(x)}}_{x\in \Lambda} \subseteq \mathcal{B}(\ell^2(\C))$, we set
    \begin{align*}
        \PTr{\Gamma}{\bigotimes_{x\in \Lambda} A^{(x)}} \coloneq \Tr{\bigotimes_{x\in \Gamma} A^{(x)}} \bigotimes_{x\in \Lambda \backslash \Gamma} A^{(x)}.
    \end{align*}
    The operator $\gOne:\ell^2(\C)\to \ell^2(\C)$ should not be confused with the reduced one-particle density matrix $\gamma^{(1)}_{\mathrm{particle}}:L^2(\Lambda)\to L^2(\Lambda)$ defined via its integral kernel $\gamma^{(1)}_{\mathrm{particle}}(x,y) = \langle \Psi_d, a^*_y a_x \Psi_d \rangle$ with $x, y \in \Lambda$. Given $\varphi\in\ell^2(\C)$, let us also introduce the corresponding orthogonal projections
    \begin{equation}\label{def_p_q}
    p = p_\varphi := \ketbra{\varphi}, ~~\text{and}~~ q = q_\varphi := 1 - p_\varphi.
    \end{equation}
    In our main results (Theorems~\ref{th:moment method} and \ref{thm:excitation method} below), we prove convergence of $\gOne(t)$ to $p_{\varphi}(t)$.

\subsection{Main Results}

    Our main results are estimates for the trace-norm difference of $\gOne(t)$ and $p_{\varphi}(t)$, which we denote by $\big\|\gOne(t) - p_{\varphi}(t)\big\|_{\mathcal{L}^1}$. We prove two similar estimates. Theorem~\ref{th:moment method} proves an error bound that holds for any value of the parameters $J,\mu,U$ of the Bose--Hubbard model \eqref{eq:HBH}. The convergence rate is slightly worse than $\frac{1}{\sqrt{d}}$. However, we need to assume stronger conditions on the initial data, and the bound contains a double exponential growth in time. On the other hand, Theorem~\ref{thm:excitation method} holds only for repulsive interaction, i.e., $U>0$. However, it holds for a larger class of initial data, and the error bound only grows exponentially in time. Our first main result is the following.

    \begin{theorem} \label{th:moment method}
        Let $\gamma_d$ be the solution to \eqref{eq:gammad dynamics} with initial data $\gamma_d(0) \in \mathcal{L}^1\prth{\mathcal{F}}$ such that $\Tr{\gamma_d(0)}=1$, and let $\varphi$ be the solution to \eqref{eq:MFS} with initial data $\varphi(0) \in \ell^2(\C)$ such that $\norm{\varphi}_{\ell^2} = 1$. Let $p_\varphi$ be defined as in \eqref{def_p_q}.
        We assume that there exist
        $a, c > 0$ such that
        \begin{align} 
            \forall n \in \mathbb{N}, \Tr{p_{\varphi}(0) \mathds{1}_{\mathcal{N} = n}} \le c e^{-\frac{n}{a}} ~\text{ and }~
            \Tr{\gOne(0) \mathds{1}_{\mathcal{N} = n}} \le c e^{-\frac{n}{a}}. \label{eq:decay hypo}
        \end{align}
        Then for all $t \in \mathbb{R}_+$ we have
        \begin{align}\label{moments_thm_main_result}
            \norm{\gOne(t) - p_{\varphi}(t)}_{\mathcal{L}^1}
                &\le\sqrt{2} \prth{\norm{\gOne(0) - p_\varphi(0)}_{\mathcal{L}^1} + \frac{C_2 e^{C_1 t} + \Tr{p_{\varphi}(0) \mathcal{N}}^{\frac{1}{2}}}{d\prth{C_4 + 2\prth{\sqrt{2(a+e)}e^{\frac{C_1}{2}t} +1} \sqrt{\ln(d+1)}}}}^{\frac{1}{2}} \nonumber\\
                    &\quad e^{JC_3\prth{C_4 + 2\prth{\sqrt{2(a+e)}e^{\frac{C_1}{2}t} +1} \sqrt{\ln(d+1)}}t},
        \end{align}
        with the following constants independent of $d, L$ and $t$:
        \begin{align*}
            &C_1 \coloneq 2eJ\max\prth{\Tr{p_\varphi(0)\mathcal{N}},1}, \\
            &C_2 \coloneq 4\prth{c (1+a) + e^{-1}}\prth{2+4(a+e)}, \\
            &C_3 \coloneq \prth{\Tr{p_\varphi(0) \mathcal{N}} + 1}^{\frac{1}{2}}, \\
            &C_4 \coloneq 4\Tr{p_{\varphi}(0) \mathcal{N}}^{\frac{1}{2}} + 2.
        \end{align*}
    \end{theorem}
    Note that the $d$-dependent terms on the right-hand side of \eqref{moments_thm_main_result} are small when $d\to\infty$, since
    \begin{align*}
        \prth{d \sqrt{\ln(d+1)}}^{-\frac{1}{2}} e^{C \sqrt{\ln(d+1)}t} 
            = e^{C \sqrt{\ln(d+1)}t - \frac{1}{2}\ln(d) - \frac{1}{4}\ln(\ln(d+1))} \limit\displaylimits_{d\to\infty} 0
    \end{align*}
    for any $C,t >0$. Our second main result is as follows.

    \begin{theorem}\label{thm:excitation method}
        Let $\gamma_d$ be the solution to \eqref{eq:gammad dynamics} with initial data $\gamma_d(0) \in \mathcal{L}^1\prth{\mathcal{F}}$, and let $\varphi$ be the solution to \eqref{eq:MFS} with initial data $\varphi(0) \in \ell^2(\C)$ such that $\norm{\varphi}_{\ell^2} = 1$. Let $p_\varphi, q_\varphi$ be defined as in \eqref{def_p_q}. We assume that  there is  $C>0$ such that
        \[{\rm Tr}(p_{\varphi}(0)\mathcal{N}^4)  \leq C,\]
        and that $U>0$.
        Then there exists $C(J,\mu,U)>0$ such that for all $t\in \mathbb{R}_+$ we have
        \begin{equation}\label{thm_2_main_est}
        \begin{aligned}
        \norm{\gOne(t) - p_{\varphi}(t)}_{\mathcal{L}^1} \leq C(J,\mu,U) e^{C(J,\mu,U) (1+t^7)} \left( \Tr{\gOne(0) \left( q_\varphi(0) \mathcal{N} ^2 q_\varphi(0) +q_\varphi(0)\right)} +\frac{1}{d} \right)^{1/2}.
        \end{aligned}
        \end{equation}
    \end{theorem}

    Note that for an initial Gutzwiller product state $\Psi_d(0) = \prod_{x \in \Lambda} \varphi_x(0)$, we have
    \begin{equation}
        \Tr{\gOne(0) \left( q_{\varphi}(0) \mathcal{N} ^2 q_{\varphi}(0) + q_{\varphi}(0)\right)} = 0.
    \end{equation}
    More generally, assuming $\Tr{\gOne(0) q_{\varphi}(0) \mathcal{N} ^2 q_{\varphi}(0)} \leq d^{-1}$ and $\Tr{\gOne(0) q_{\varphi}(0)} \leq d^{-1}$, the estimate \eqref{thm_2_main_est} becomes 
    \begin{equation*}
    \norm{\gOne(t) - p_{\varphi}(t)}_{\mathcal{L}^1} \leq C(J,\mu,U) e^{C(J,\mu,U) (1+t^7)} \frac{1}{\sqrt{d}}.
    \end{equation*}
    For example, for the state $\prod_{x \in \Lambda \setminus S} \varphi_x(0) \prod_{x \in S} \varphi^\perp_x(0)$ with $\varphi^\perp(0) \perp \varphi(0)$ and $|S|=\frac{|\Lambda|}{d}$, we find $\Tr{\gOne(0) q_{\varphi}(0)} = d^{-1}$, and the other conditions of the theorems can be satisfied by an appropriate choice of $\varphi(0)$ and $\varphi^\perp(0)$. From the perspective of the law of large numbers, this is the expected optimal convergence rate. However, this convergence rate obviously does not explain why the large-dimensional mean-field approximation is so successful even for $d=3$. We refer to \cite{PhysRevB.79.224515}, in particular \cite[FIG. 20]{PhysRevB.79.224515}, for a numerical analysis of the phase diagram in different dimensions with a comparison to the mean-field limit. Note also that the bound \eqref{thm_2_main_est} can be written in more detail (see \eqref{detailed}) as  
    \begin{equation}
    \begin{aligned}
    & \norm{\gOne(t) - p_{\varphi}(t)}_{\mathcal{L}^1}  
    \\& \leq \left(  \frac{1}{d}  \frac{1}{U} + \widetilde{C}(J,\mu,U) \left(1+\frac{1}{U^2} \right) e^{\widetilde{C}(J,\mu,U) \sum_{j=1}^{7} t^j} \left( \Tr{\gOne(0) \left( q_{\varphi}(0) \mathcal{N} ^2 q_{\varphi}(0) +q_{\varphi}(0)\right)} +\frac{1}{d} \right) \right)^{1/2},
    \end{aligned}
    \end{equation}
    where $\widetilde{C}(J,\mu,U)$ depends polynomially on the parameters $J$, $\mu$, $U$ of the Bose--Hubbard model. The divergence for small $U$ comes from our use of an energy estimate, as outlined below (around Equation~\eqref{exc_energy}).

\begin{remark}
In the above theorems, we could replace the local interaction term
$ 
\frac{U}{2}\mathcal{N}_x(\mathcal{N}_x-1)
$
by more general on-site interactions $P(\mathcal{N}_x)$ by suitably adapting the proofs, provided that $P$ is a polynomial of degree greater than $1$ (with positive leading coefficient for Theorem~\ref{thm:excitation method}).
\end{remark}

    \begin{remark}
    The trace norm convergence of Theorems~\ref{th:moment method} and \ref{thm:excitation method} in particular implies convergence of the order parameter $\alpha$, meaning
    \begin{equation}\label{alpha_convergence}
    \alpha_{\mathrm{micro}}(t) := \frac{1}{|\Lambda|} \sum_{x \in \Lambda} \langle \Psi_d(t), a_x \Psi_d(t) \rangle \to \alpha_{\varphi(t)} := \langle \varphi(t),a \varphi(t) \rangle ~~\text{as}~ d\to\infty.
    \end{equation}
    Note that for initial data $\Psi_d(0)$ with a fixed particle number, the left-hand side of \eqref{alpha_convergence} is zero (since $H_d$ is particle number conserving), but in general our initial data live on Fock space where the particle number is not fixed.
    
    To prove \eqref{alpha_convergence}, let us consider the operator $\mathcal{O} \coloneq (\mathcal{N}+1)^kB $ with $B$ a bounded operator on $\ell^2(\C)$. Inserting a cutoff $M \in \N$ we find
    \begin{align*}
        &\abs{\Tr{\gOne\mathcal{O}} - \Tr{p_\varphi \mathcal{O}}} \\
            &\quad\le \ \abs{\Tr{\prth{\gOne - p_\varphi}\prth{\mathcal{N}+1}^k\mathds{1}_{\mathcal{N} < M} B}}
                + \abs{\Tr{\prth{\gOne - p_\varphi}\prth{\mathcal{N}+1}^k\mathds{1}_{\mathcal{N} \ge M} B}} \\
            &\quad\le \ \norm{B}_{\mathcal{L}^\infty}M^k\norm{\gOne - p_\varphi}_{\mathcal{L}^1} + \Tr{\gOne (\mathcal{N}+1)^{2k}\mathds{1}_{\mathcal{N} \ge M}}^{\frac{1}{2}}\Tr{\gOne \abs{B}^2}^{\frac{1}{2}} \\
                &\qquad+ \Tr{p_\varphi (\mathcal{N}+1)^{2k}\mathds{1}_{\mathcal{N} \ge M}}^{\frac{1}{2}}\Tr{p_\varphi \abs{B}^2}^{\frac{1}{2}} \\
            &\quad\le \ \norm{B}_{\mathcal{L}^\infty}\prth{M^k\norm{\gOne - p_\varphi}_{\mathcal{L}^1} + \Tr{\gOne (\mathcal{N}+1)^{2k}\mathds{1}_{\mathcal{N} \ge M}}^{\frac{1}{2}} + \Tr{p_\varphi (\mathcal{N}+1)^{2k}\mathds{1}_{\mathcal{N} \ge M}}^{\frac{1}{2}}}
    \end{align*}
    Assuming that
    \begin{align*}
        \Tr{(\mathcal{N} +1)^{2k}\gOne(0)} + \Tr{(\mathcal{N} +1)^{2k} p_\varphi(0)} < \infty
    \end{align*}
    we are able to propagate these moments (see Proposition~\ref{prop:moments bound}). Thus we can estimate the remainder terms, i.e., 
    \begin{align*}
        \Tr{\gOne(\mathcal{N} +1)^{2k}\mathds{1}_{\mathcal{N}\ge M}} + \Tr{p_\varphi(\mathcal{N} +1)^{2k}\mathds{1}_{\mathcal{N}\ge M}} \limit\displaylimits_{M\to\infty} 0.
    \end{align*}
    and any choice of $M \gg 1$ such that $M^k \norm{\gOne - p_\varphi}_{\mathcal{L}^1} \ll 1$ as $d\to\infty$ is sufficient to prove that
    \begin{align*}
        \abs{\Tr{\gOne\mathcal{O}} - \Tr{p_\varphi \mathcal{O}}} \limit\displaylimits_{d\to\infty}0
    \end{align*}
    In particular this applies to the order parameter since $\norm{(\mathcal{N}+1)^{-\frac{1}{2}}a}_{\mathcal{L}^\infty} \le 1$.
    In line with the above arguments, we can similarly establish that  the reduced one-particle density matrix $\gamma^{(1)}_{\mathrm{particle}}$, defined via its integral kernel $\gamma^{(1)}_{\mathrm{particle}}(x,y) = \langle \Psi_d, a^*_y a_x \Psi_d \rangle$ with $x, y \in \Lambda$,
converges in Hilbert-Schmidt norm. Specifically,  we obtain
 \begin{align*}
\frac{1}{\vert \Lambda \vert^2} \sum_{x,y \in \Lambda} \Big\vert \langle \Psi_d(t), a_y^* a_x \Psi_d(t) \rangle - \vert \alpha_{\varphi(t)} \vert^2 \Big\vert^2\rightarrow 0 ~~\text{as}~ d\to\infty.
\end{align*}
To prove this, we  can insert first the identities $p_x+q_x=1$ and $p_y+q_y=1$  inside $ \langle \Psi_d(t), a_y^* a_x \Psi_d(t) \rangle$ and then  all the  resulting terms can be bounded by ${\rm Tr}(\gOne (q+q \mathcal{N}^2 q))$ which  converges to zero as $d \rightarrow \infty$. The convergence of the latter quantity follows from Propositions~\ref{equiv} and \ref{granwalllemmaf}.
Moreover, we can also establish  that
\begin{align*}
     \frac{1}{d |\Lambda|} \sum_{<x,y>} \langle \Psi_d(t), a_x^* a_y \Psi_d(t) \rangle \to \vert \alpha_{\varphi(t)} \vert^2 ~~\text{as}~ d\to\infty,
    \end{align*}
which implies convergence of the kinetic energy.
    \end{remark}

    \begin{remark}
        Our theorems imply factorization of higher order reduced lattice-site density matrices. For example, let us denote by
        \begin{align*}
            \eta_d^{(2)} \coloneq \frac{1}{\abs{\Lambda}^2} \sum_{x, y \in \Lambda} \PTr{\Lambda \backslash \sett{x, y}}{\gamma_d}
        \end{align*}
        the general reduced two-lattice-site density matrix. With a standard argument where we insert $(p_\varphi+q_\varphi)\otimes(p_\varphi+q_\varphi)$ on both sides of $\eta_d^{(2)}$ we can obtain with a similar argument as in Lemma \ref{lem_q_dens_mat} that
        \begin{align*}
            \norm{\eta_d^{(2)} - p_\varphi^{\otimes 2}} _{\mathcal{L}^1} \lesssim \sqrt{\Tr{\gamma_d^{(1)}q_\varphi}} \limit\displaylimits_{d\to\infty} 0.
        \end{align*}
    \end{remark}

    \begin{remark}
        We think that our work does not fall into the setting for well-know many-body mean-field limit derivations because of two reasons:
        \begin{itemize}
            \item hopping occurs only between nearest neighbours,
            \item the Fock space does not have lattice-site symmetry.
        \end{itemize}
            For instance, \cite{Liard2017} is written in a general setting that would include $\ell^2(\mathbb{C})$ as particle Hilbert space, $a^* \otimes a$ as interaction and $(J - \mu) \mathcal{N} + \frac{U}{2}\mathcal{N}(\mathcal{N}-1)$ as external potential, but then the model would be
            \begin{equation}
                - \frac{J}{\abs{\Lambda}}\sum_{x, y\in \Lambda} \prth{a^*_x a_y + a^*_y a_x} + (J - \mu) \sum_{x\in\Lambda} \mathcal{N}_x + \frac{U}{2} \sum_{x\in\Lambda} \mathcal{N}_x (\mathcal{N}_x - 1), \label{eq:fully sym model}
            \end{equation}
            i.e., a model with constant hopping amplitude between any two sites, with the Hamiltonian acting on symmetric products of the one-lattice-site Hilbert space. In this case, indeed, the large-d limit would be an equivalent setting to the standard many-body mean-field. We believe that \eqref{eq:fully sym model} can also be treated with our method of proof, but this is not as physically relevant because it would break the local characteristic of the Bose--Hubbard model. Note that this local aspect makes our limit a bit similar to the large volume mean-field limit, see, e.g., \cite{derezinski2014,deckert:2012,petrat_2020}, where the interaction is averaged over a small proportion of the total system volume as well.
    \end{remark}

    Both theorems are proven using a Gronwall estimate for $\text{Tr}(\gOne q_\varphi)$. This quantity heuristically counts the average number of lattice sites that do not follow the product state ansatz. It is inspired by the corresponding quantity for the weak coupling limit introduced by Pickl \cite{pickl2011}. The main technical challenge is then caused by the unboundedness of the creation and annihilation operators in the hopping term, which is a bit analogous to the technical problems that arise when considering the weak-coupling limit with singular interactions; see, e.g., \cite{knowles2010}. More concretely, we need to bound $\text{Tr}(\gOne q_\varphi (\mathcal{N}+1)q_\varphi)$ in terms of $\text{Tr}(\gOne q_\varphi)$ or terms that go to zero as $d\to\infty$. We do this in two different ways, leading to the two main theorems.
    
    For Theorem~\ref{th:moment method}, we use a moment method which localizes in the particle number $\mathcal{N}$. This idea has been used before in the context of bosonic dynamics \cite{NAM2017662} and spectrum \cite{lewin2015_2} \cite{LEWIN2014570} under the name \textit{localization in Fock space}. For this, we first separate
    \begin{equation*}
    \Tr{\gOne q_\varphi (\mathcal{N}+1)q_\varphi} = \Tr{\gOne q_\varphi (\mathcal{N}+1) \mathds{1}_{\mathcal{N} < M} q_\varphi} + \Tr{\gOne q_\varphi (\mathcal{N}+1) \mathds{1}_{\mathcal{N} \geq M} q_\varphi}.
    \end{equation*}
    Then the first term can simply be bounded by $M$, whereas we use moment estimates to bound the second term by $e^{MC(t)}d^{-1}$. Then it turns out to be possible to optimize in $M$ to close the Gronwall argument. The introduction of this cutoff parameter is inspired by \cite{Gyrokineticlimit}, where a Landau level cutoff is introduced, and the remainder term is controlled through moments of the kinetic energy operator in a similar manner as what we do here with the number of particles operator.
    
    For Theorem~\ref{thm:excitation method}, we proceed using an energy estimate inspired by \cite{lewin2015} (which deals with proving Bogoliubov theory for the dynamics of the weakly interacting Bose gas). The idea is to write the Hamiltonian \eqref{eq:HBH} as in \eqref{eq:HBH_rewritten}, i.e., as
    \begin{equation*}
    H_d = \sum_{x\in \Lambda} h^{\varphi(t)}_x + \widetilde{H}(t).
    \end{equation*}
    Here, $\widetilde{H}(t)$ describes the excitations around our product state ansatz. A similar splitting was used in \cite{lewin2015}, where, after a unitary transformation, $\widetilde{H}(t)$ converges to a Bogoliubov Hamiltonian. The energy of the excitations should now be defined as
    \begin{equation}\label{exc_energy}
    E^{\mathrm{exc}}(t) := \Big\langle \Psi_d, \left( \widetilde{H}(t) + \sum_{x\in \Lambda} q_x h^{\varphi(t)}_x q_x\right) \Psi_d \Big\rangle,
    \end{equation}
    where the first term corresponds to the kinetic energy, and the second term to the mean-field energy of the excitations (including the interaction energy). The energy $E^{\mathrm{exc}}(t)$ is not conserved, as excitations from the product state can be created and annihilated. However, it can be bounded using a Gronwall argument. The crucial point here is that the interaction term enters only as
    \begin{align*}
        \sum_{x\in\Lambda}q_x \mathcal{N}_x (\mathcal{N}_x - 1) q_x,
    \end{align*}
    and hence the Gronwall argument does not produce higher powers than $q_x \mathcal{N}_x^2 q_x$. This ultimately allows us to control terms that involve $q_x \mathcal{N}_x^2 q_x$ or $q_x \mathcal{N}_x q_x$, in particular $\text{Tr}(\gOne q_\varphi (\mathcal{N}+1)q_\varphi)$.

    The remainder of the paper is organized as follows. We first prove global well-posedness of the mean-field equation \eqref{eq:MFS} in Section~\ref{șec_well_posedness}. Then, we discuss some preliminary estimates in Section~\ref{sec_preliminaries}. In particular, we prove properties of a two-lattice-site reduced density matrix, preliminary bounds for the mean-field and Bose--Hubbard energies, and conservation laws and propagation estimates for the mean-field equation and the Bose--Hubbard dynamics. Furthermore, we compute $\partial_t \text{Tr}(\gOne(t) q_{\varphi}(t))$, and prove bounds on all the terms in this time derivative except for $\text{Tr}(\gOne q_\varphi (\mathcal{N}+1)q_\varphi)$. Then, Theorem~\ref{th:moment method} is proven in Section~\ref{sec_proof_thm_1}, and Theorem~\ref{thm:excitation method} is proven in Section~\ref{sec_proof_thm_2}, each using a different method to control $\text{Tr}(\gOne q_\varphi (\mathcal{N}+1)q_\varphi)$.

    We expect that these methods can be applied to a larger class of lattice models. For example one may consider a nearest neighbour interaction term and average out this interaction as it is done in this paper for the hopping term. If $V$ denotes the nearest neighbours' coupling, this would lead to the approximation
    \begin{align*}
        \dfrac{V}{2d} \sum_{\braket{x, y}} \prth{\mathcal{N}_x \mathcal{N}_y + \mathcal{N}_y \mathcal{N}_x} \eqvl\displaylimits_{d\to\infty} V \braket{\varphi, \mathcal{N} \varphi} \sum_{x\in\Lambda} \mathcal{N}_x.
    \end{align*}
    A natural generalization of our results would be to consider a general graph $\mathcal{G}_z$ with a large (and constant) coordination number $z \to \infty$. Since our bounds are uniform in $L$, the following restriction on the number of vertices would be sufficient:
    \begin{align*}
        \abs{\mathcal{G}_z} \gg z \, .
    \end{align*}
    The hopping term should then be replaced by
    \begin{align*}
        -\frac{J}{z} \sum_{\braket{x,y}} \prth{a_x^* a_y + a_y^* a_x} \, ,
    \end{align*}
    and the resulting estimates hold after replacing $2d$ by $z$. An example of such a graph is the three-dimensional cubic lattice with hopping inside a large ball of radius $r$, so that
    \begin{align*}
        z \eqvl\displaylimits_{r\to\infty} \frac{4}{3}\pi r^3 ,
    \end{align*}
    with the lattice size satisfying $L \gg r$. Alternatively, on the $d$-dimensional square lattice with hopping to neighbours of order $n\in\mathbb{N}$, one has
    \begin{align*}
        z \eqvl\displaylimits_{d\to\infty} \frac{(2d)^n}{n!}
    \end{align*}
    which can be seen as an improvement of our convergence rate.
    
    We end this section with some physical phenomenology about the phase transition.

    \begin{remark}[Overview of the phase transition] \label{rmk:phase transition}
        This paper focuses on the dynamics, but the phase transition in the mean-field model provided by \eqref{eq:h mf} is better understood from an energy perspective. Given a minimizer of the energy functional
        \begin{align}\label{eq:mf energy}
            \varphi \mapsto \braket{\varphi, h^\varphi \varphi} = -J \abs{\alpha_\varphi^2} + \braket{\varphi, ((J - \mu) \mathcal{N} + \frac{U}{2} \mathcal{N}(\mathcal{N}-1)) \varphi}
        \end{align}
        over $\varphi \in \ell^2(\mathbb{C})$ such that $\norm{\varphi}_{\ell^2} = 1$ and $\mathcal{N} \varphi \in \ell^2(\mathbb{C})$, we claim that $\alpha_\varphi = 0$ indicates a Mott insulator state, whereas $\alpha_\varphi \neq 0$ indicates a superfluid state.
    
        Indeed, if $\alpha_\varphi = 0$, then the ground state of \eqref{eq:h mf} is $\ket{n_0}$ with
        \begin{align*}
            n_0 \coloneq \argmin\displaylimits_{n\in\N} \prth{(J - \mu)n + \frac{U}{2} n(n-1)}
            = \round{\frac{1}{2} + \frac{\mu - J}{U}},
        \end{align*}
        where $\round{}$ is the nearest integer function. Here, the number of particles at every site is constant and fixed, which is representative of an insulator state. Notice, nevertheless, that a degeneracy occurs for half-integer values of $\frac{\mu - J}{U}$.
    
        A very different behavior occurs when $J \gg \mu, U$. We assume for simplicity that $\mu = U = 0$. Then the energy is minimal (and equal to $0$) if and only if the following Cauchy--Schwarz inequality is sharp:
        \begin{align*}
            \abs{\alpha_\varphi}^2 \le \braket{\varphi, \mathcal{N} \varphi}.
        \end{align*}
        This situation is obtained for the family of coherent states
        \begin{align*}
            \varphi_\alpha \coloneq e^{-\frac{\abs{\alpha}^2}{2} + \alpha a^*} \ket{0}, \alpha \in \mathbb{C},
        \end{align*}
        satisfying $\braket{\varphi_\alpha, a\varphi_\alpha} = \alpha$ and $\braket{\varphi_\alpha, \mathcal{N} \varphi_\alpha} = \abs{\alpha}^2$. The macroscopic state is then
        \begin{align*}
            \varphi_\alpha^{\otimes \abs{\Lambda}} 
                =  e^{-\frac{\abs{\alpha}^2 \Lambda}{2} + \alpha \sum_{x\in \Lambda}a_x ^*} \ket{\Omega}
                = e^{-\frac{\abs{\alpha}^2 \Lambda}{2} + \alpha \sqrt{\abs{\Lambda}} a^*(u)} \ket{\Omega}
        \end{align*}
        with $\Omega = \ket{0}^{\otimes \abs{\Lambda}} \in \mathcal{F}$ being the vacuum state, $u:x\in\Lambda \mapsto \abs{\Lambda}^{-\frac{1}{2}} \in L^2(\Lambda)$ being a $1$-particle state with constant density over the lattice, and
        \begin{align*}
            a^*(u) \coloneq \sum_{x\in \Lambda} u(x) a_x^*.
        \end{align*}
        Hence, $\varphi_\alpha^{\otimes \abs{\Lambda}}$ is a Bose--Einstein condensate of $\abs{\alpha}^2 \abs{\Lambda}$ particles. We refer to \cite{PhysRevB.40.546} for more details about the phase-transition, and to \cite{GS_2026} for a derivation of the mean-field energy functional \eqref{eq:mf energy} starting from the Bose--Hubbard model.
    \end{remark}

\section{Global Well-posedness of the Mean-Field Dynamics}\label{șec_well_posedness}

The goal of this section is to prove well-posedness of the mean-field dynamics. The mean-field dynamics can be written as 
\begin{equation}\label{pde}
\begin{cases}
i\partial_t \varphi =h^{\varphi} \varphi=A \varphi +F(\varphi),
\\ \varphi(0)=\varphi_0,
\end{cases}
\end{equation}
where the linear operator $A$ and the nonlinear operator $F$ are defined as
\begin{align}
A &:= (J-\mu)\mathcal{N} +\frac{U}{2} \mathcal{N}(\mathcal{N} -1),\label{linearpart}
\\ F(\varphi)&:=-J\left( \alpha_{\varphi} a^* +\overline{\alpha_{\varphi}} a-\vert \alpha_{\varphi}\vert^2 \right)\varphi .\label{nonlinearpart}
\end{align}
The Duhamel version  of the mean-field dynamics \eqref{pde} is 
\begin{equation} \label{DuhamelMF}
  \varphi(t) =   \tilde  \varphi_0(t)-i\int_0^t e^{-i(t-s)A} F(\varphi(s)) \, ds,
\end{equation}
with $\tilde \varphi_0(t):=e^{-itA} \varphi_0$.
When examining the semi-linear equation \eqref{pde} or \eqref{DuhamelMF} above, we cannot directly apply fixed-point arguments to study global well-posedness because both \( A \) and \( F \) are unbounded operators, and the nonlinear operator \( F \) is not Lipschitz continuous. Therefore, a different approach is required. Our strategy is to approximate the nonlinear term so that it becomes Lipschitz continuous. This allows us to establish the existence of a unique solution to the approximated problem by standard methods. Then, we show that the obtained solution converges to the solution of the untruncated mean-field equation. To this end, we have the following result. Denote by $W^{1,1}(\mathbb{R}^+, \ell^2(\mathbb{C}))$  the Sobolev space endowed with the norm 
\[\Vert \varphi\Vert_{ W^{1,1}(\mathbb{R}^+, \ell^2(\mathbb{C}))} := \Vert  \varphi\Vert_{ L^{1}(\mathbb{R}^+, \ell^2(\mathbb{C}))}+\Vert \partial_t \varphi \Vert_{ L^{1}(\mathbb{R}^+, \ell^2(\mathbb{C}))}= \int_0^\infty \Vert \varphi(t) \Vert_2 \ dt+\int_0^\infty \Vert \partial_t\varphi(t) \Vert_2 \ dt .\]
\begin{theorem}[Global existence of the solution to the mean-field dynamics]\label{gwp}
    Assume $\varphi_0\in \mathcal{D}(\mathcal{N}^2)$. There exists a solution $\varphi(\cdot) \in L^1(\mathbb{R}^+, \mathcal{D}(\mathcal{N}^2) ) \cap W^{1,1}(\mathbb{R}^+, \ell^2(\mathbb{C}) )$ to the mean-field dynamic \eqref{pde} for almost all $t\in \mathbb{R}^+$.
\end{theorem}
The proof of Theorem \ref{gwp} is postponed to the end of this section. In Paragraph \ref{approximatedynamic}, we first introduce the approximate model of \eqref{pde}. Next, in Paragraph \ref{conservationlawsapproximate}, we establish several conservation laws for the approximate dynamics. In the following Paragraph \ref{gwpapproximate}, we prove the global well-posedness of the approximate dynamics. Finally in Paragraph \ref{convergenceapproximate}, we show that the solution of the Duhamel formula for the mean-field dynamics \ref{DuhamelMF} can be obtained as the limit of these approximate solutions and  we conclude by completing the proof of Theorem \ref{gwp}.
\subsection{Approximating the Mean-field Dynamics}\label{approximatedynamic}
First, we truncate the mean-field dynamics by introducing a cutoff on the number operator $ \mathds{1}_{\mathcal{N}\leq M} $. We then prove uniform estimates in $M$, which yield well-posedness at the end of the section. Similar cutoff  (or localization) techniques  have appeared before and were used to control large excitation sectors in rigorous derivations and analysis of many-body quantum dynamics, see e.g. \cite{Bog,Nam_2019}. Let $M>0$ and consider the truncated creation and  annihilation operators 
\begin{align}
a_M:=a \mathds{1}_{\mathcal{N}\leq M}, \qquad a^*_M:=  \mathds{1}_{\mathcal{N}\leq M} a^*
\end{align}
and
\begin{align}
\alpha_M :=\langle \varphi_M, a_M \varphi_M \rangle, 
\end{align}
where $\varphi_M$ is the solution to the  approximated problem
\begin{equation}\label{pdeM}
\begin{cases}
i\partial_t \varphi_M =h_M^{\alpha_M} \varphi_M=A \varphi_M +F_M(\varphi_M),
\\ \varphi_M(0)=\varphi_0\in \mathcal{D}(\mathcal{N}^2),
\end{cases}
\end{equation}
where we have introduced the approximated nonlinear operator $F_M$ as 
\begin{align}
F_M(\varphi_M)&:=-J\left( \alpha_M a^*_M +\overline{\alpha_M} a_M-\vert \alpha_M\vert^2 \right)\varphi_M .\label{appnonlinearpart}
\end{align} 
The solution to \eqref{pdeM} solves the weak form of the preceding nonlinear equation \eqref{pdeM} which is usually known as Duhamel formula
\begin{equation}\label{Duhamel}
\begin{cases}
\varphi_M(t)= \tilde{\varphi}_0(t)-i\int_0^t e^{-i(t-s)A} F_M(\varphi_M(s)) \, ds,
\\ \tilde \varphi_0(t):=e^{-itA} \varphi_0, \qquad \varphi_0 \in \mathcal{D}(\mathcal{N}^{2}).
\end{cases} 
\end{equation}
\begin{remark}
Note the following:
\begin{enumerate}  
\item For the weak formulation \eqref{Duhamel} of the approximated nonlinear equation \eqref{pdeM}, the existence of a unique local solution can be established using fixed-point arguments for a broader class of initial data, specifically \( \varphi_0 \in \ell^2(\C) \), which leads to the existence of a unique local solution \( \varphi_M \in C([0,T], \ell^2(\C)) \). This follows from the fact that \( F_M \) is a nonlinear  operator satisfying \( F_M(\varphi_M) \in C([0,T], \ell^2(\C)) \). However, to extend the solution to global times, we rely on the conservation laws, which require the initial data \( \varphi_0 \in \mathcal{D}(\mathcal{N}^2) \), as $A$ remains an unbounded operator.  
\vskip 2mm  
\item Since \( F_M(\varphi_M) \in C([0,T], \ell^2(\C)) \), to ensure the equivalence between \eqref{pdeM} and \eqref{Duhamel}, it is enough to restrict our analysis to initial data \( \varphi_0 \in \mathcal{D}(\mathcal{N}^2) \). For further details, see \cite[Lemma 4.1.1,  Proposition 4.1.6 and Corollary 4.1.8]{cazenave1998introduction}.  
\item One could in addition truncate the unbounded linear term \( A \). On the one hand, this approach ensures equivalence between \eqref{pdeM} and \eqref{Duhamel} for all initial data \( \varphi_0 \in \ell^2(\C) \). On the other hand, it allows us to obtain a unique global strong solution \( \varphi_M \in C(\mathbb{R}, \ell^2(\C)) \). However, ensuring convergence to the solution of the mean-field dynamics becomes more complicated.
\end{enumerate}  
\end{remark}

\subsection{Properties of the Approximate Solution}\label{conservationlawsapproximate}
In this subsection, we state some conservation laws for the approximated problem.
\begin{lem}\label{conservationlaws}
Assume that $\varphi_M$ is a solution to \eqref{pdeM} with $\Vert \varphi_0\Vert_{\ell^2} =1$. Then, the following holds:
\begin{itemize}
\item [(i)] $\Vert \varphi_M(t)\Vert_{\ell^2}  =\Vert \varphi_0\Vert_{\ell^2}=1.$
\item [(ii)] $\langle \varphi_M(t), \mathcal{N} \varphi_M(t) \rangle =\langle \varphi_0, \mathcal{N} \varphi_0 \rangle .$
\item [(iii)]$\langle \varphi_M(t),h_M^{\alpha_M} \varphi_M(t) \rangle =\langle \varphi_0, h_M^{\alpha_M} \varphi_0 \rangle .$
\item [(iv)] $\vert \alpha_M\vert \leq \Vert \mathcal{N}^{1/2}\varphi_0 \Vert_{\ell^2}. $
\item [(v)] Assume that $  \varphi_0 \in \mathcal{D}(\mathcal{N}^{k}) $. Then, there exists a constant $C > 0$ such that
\vskip 1mm
$\langle \varphi_M(t),\mathcal{N}^k \varphi_M(t) \rangle \leq \sum_{j=0}^{2k-2}  \left(  C J k \Vert \mathcal{N}^{1/2}\varphi_0 \Vert_{\ell^2} \right)^j   \left \langle \varphi_0,(\mathcal{N}+j)^{k-\frac{j}{2} }\varphi_0 \right \rangle \frac{t^j}{j!}$.
\end{itemize}
\end{lem}

\begin{proof}
Statement (i) is true by definition of the truncation. For (ii), note that 
\begin{equation}
[h_M^{\alpha_M}, \mathcal{N}]= -J(-\alpha_M a_M^* +\overline{\alpha_M} a_M).
\end{equation}
This gives
\begin{align*}
\frac{d}{dt}\langle \varphi_M ,\mathcal{N} \varphi_M\rangle=i \langle \varphi_M , [h_M^{\alpha_M}, \mathcal{N}] \varphi_M \rangle 
=-iJ  \langle \varphi_M , (-\alpha_M a_M^* +\overline{\alpha_M} a_M) \varphi_M \rangle 
=0. 
\end{align*}
For (iii), we have 
\begin{align*}
\frac{d}{dt}\langle \varphi_M ,h_M^{\alpha_M} \varphi_M\rangle&=i \langle \varphi_M , [h_M^{\alpha_M}, h_M^{\alpha_M}] \varphi_M \rangle + \langle \varphi_M , \partial_t h_M^{\alpha_M} \varphi_M \rangle 
\\&=-J  \left \langle \varphi_M , \left(\partial_t \alpha_M a_M^* +\partial_t \overline{\alpha_M} a_M -\partial_t \alpha_M \overline{\alpha_M} - \partial_t\overline{\alpha_M} {\alpha_M} \right) \varphi_M  \right \rangle 
=0. 
\end{align*}
Statement (iv) follows directly from Cauchy--Schwarz and (i)-(ii). For (v), note that
\begin{equation}
\begin{aligned}
\frac{d}{dt}\langle \varphi_M ,\mathcal{N}^k \varphi_M \rangle  &= -iJ \alpha_M  \langle \varphi_M ,[a^*_M, \mathcal{N}^k] \varphi_M \rangle-iJ \overline{\alpha_M } \langle \varphi_M ,[a_M, \mathcal{N}^k] \varphi_M \rangle
\\ &  = -iJ \alpha_M  \langle \varphi_M ,\mathds{1}_{\mathcal{N}\leq M} [a^*, \mathcal{N}^k] \varphi_M \rangle-iJ \overline{\alpha_M } \langle \varphi_M ,[a, \mathcal{N}^k] \mathds{1}_{\mathcal{N}\leq M} \varphi_M \rangle
\\ &  = 2J \Im   \left(  \alpha_M  \langle \varphi_M ,\mathds{1}_{\mathcal{N}\leq M} [a^*, \mathcal{N}^k] \varphi_M \rangle\right) 
\\ & \leq 2 k\vert J \vert  \vert \alpha_M \vert \vert \langle  \varphi_M ,\mathds{1}_{\mathcal{N}\leq M} a^* \mathcal{N}^{k-1}  \varphi_M \rangle \vert 
\\ & \leq 2k \vert J\vert \Vert \mathcal{N}^{1/2}\varphi_0 \Vert_{\ell^2} \langle \varphi_M , (\mathcal{N}+1) ^{k-\frac{1}{2}} \varphi_M\rangle . 
\end{aligned}
\end{equation}
Iterating this $(2k-2)$ times leads to 
\begin{equation}\label{mombounds}
\langle \varphi_M ,\mathcal{N}^k \varphi_M \rangle  \leq  \sum_{j=0}^{2k-2}  \left(  C  \vert J\vert k \Vert \mathcal{N}^{1/2}\varphi_0 \Vert_{\ell^2} \right)^j   \left \langle \varphi_0,(\mathcal{N}+j)^{k-\frac{j}{2} }\varphi_0 \right \rangle \frac{t^j}{j!}.
\end{equation}
\end{proof}

\subsection{Global Well-posedness of the Approximated Problem}\label{gwpapproximate}
For the approximated problem, proving global well-posedness is straightforward due to the use of standard techniques such as fixed-point arguments, particularly because the nonlinearity in this case is Lipschitz. We have the following results.
\begin{lem}
Let $\varphi_0\in \mathcal{D}(\mathcal{N}^2)$. Then, for any  fixed $M>0$, we have the following statements:  
\begin{itemize}
\item [(i)] There exists a unique global strong solution  $\varphi_M(\cdot ) \in \cC(\RR,\mathcal{D}(\mathcal{N}^{2}))$ of the Duhamel formula \eqref{Duhamel}.
\item [(ii)] There exists a unique global strong solution   $\varphi_M(\cdot ) \in \cC(\RR,\mathcal{D}(\mathcal{N}^{2}))\cap \cC^1(\RR,\ell^2(\C))$ of the approximated  problem  \eqref{pdeM}.
\end{itemize}
\end{lem}

\begin{proof}
Let \( X = \mathcal{C}([0,T], \mathcal{D}(\mathcal{N}^{2})) \) denote the space of continuous functions from \( [0,T] \) to \( \mathcal{D}(\mathcal{N}^{2}) \), equipped with the norm
\[ \vert \vert \vert \varphi \vert \vert \vert : = \sup_{t\in [0,T]} \Vert \varphi(t) \Vert_{ \mathcal{D}(\mathcal{N}^{2})},\qquad     \Vert \varphi(t) \Vert^2_{ \mathcal{D}(\mathcal{N}^{2})}= \Vert \varphi(t) \Vert_{\ell^2}^2+\Vert \mathcal{N}^{2}  \varphi(t) \Vert_{\ell^2}^2.\]
 Note that $\left(\mathcal{D}(\mathcal{N}^{2}),  \Vert \cdot \Vert_{ \mathcal{D}(\mathcal{N}^{2})} \right) $ is a Banach space.
For a fixed \( M > 0 \), we define the map \( \Gamma_M: X \rightarrow X \) by
\[\Gamma_M(\varphi)(t):=\tilde \varphi_0(t) -i \int_0^t e^{-i(t-s)A} F_M(\varphi(s)) ds,\]
 with $\tilde \varphi_0(t):=e^{-it A}\varphi_0 $ and where $A$ and  $F_M$  are  defined in \eqref{linearpart}and \eqref{appnonlinearpart}.
We can check that, for any  $T>0$ and  for all $\varphi_1,\varphi_2\in X$, the map $\Gamma_M$  satisfies
\begin{align*}
 \vert \vert \vert  \Gamma_M(\varphi_1) - \Gamma_M (\varphi_2)\vert \vert \vert \leq C(M,J,T) \vert \vert \vert \varphi_1 - \varphi_2 \vert \vert \vert ,
\end{align*}
where $C(M,J,T)>0$ is defined as
\[ C(M,J,T):=MT \vert J\vert (6c^2+6c^2{M}^2 +10c^4),\]
and $c>0$ as
\[c:=\max_{i=1,2} \sup_{[0,T]} \Vert \varphi_i (t) \Vert_{\mathcal{D}(\mathcal{N}^{2})}<\infty.\] 
To prove the claim, we need first to establish some useful estimates. To this end, we denote $ \alpha_{M,i}(s):=\langle \varphi_i(s), a_M\varphi_i(s) \rangle$. We have,  for $k\geq0$ and for $i=1,2$,
\begin{align}
& \Vert \mathcal{N}^k  a_M^\sharp\varphi_i(s)\Vert_{\ell^2} \leq M^{k+\frac{1}{2}} \Vert \varphi_i(s)\Vert_{\ell^2}\leq M^{k+\frac{1}{2}} c, \quad \sharp \in\{\ \  ,* \}, \label{amphi}
\\ & \vert \alpha_{M,i} (s)\vert \leq \sqrt{M} \Vert \varphi_i(s)\Vert_{\ell^2}^2\leq \sqrt{M} c^2,\label{alpha}
\\& \vert \alpha_{M,1}(s)-\alpha_{M,2}(s)\vert   \leq 2\sqrt{M} c \Vert \varphi_1(s) -\varphi_2(s)\Vert_{\ell^2},  \label{diffalpha}
\\& \left \vert \vert \alpha_{M,1}(s)\vert^2-\vert\alpha_{M,2}(s)\vert^2 \right \vert \leq 4M c^3 \Vert \varphi_1(s) -\varphi_2(s)\Vert_{\ell^2}. \label{diffsquarealpha}
\end{align}
We have then for all $t\in [0,T]$,
 \begin{align*}
 \Vert \Gamma_M(\varphi_1)(t) - \Gamma_M (\varphi_2)(t)\Vert_{\ell^2} & = \left \Vert \int_0^t e^{-i(t-s)A} \left( F_M(\varphi_1(s))-F_M(\varphi_2(s))\right) \, ds\right \Vert_{\ell^2}
\\& \leq  \vert J\vert \int_0^t   \bigg(  \left \Vert  \overline{\alpha_{M,1}(s)} a_M \varphi_1(s) -  \overline{\alpha_{M,2}(s)}  a_M \varphi_2(s) \right  \Vert_{\ell^2}
 \\& \qquad \quad\quad \ \ +\left \Vert  \alpha_{M,1}(s)   a_M^* \varphi_1(s) - \alpha_{M,2}(s) a_M^* \varphi_2(s) \right  \Vert_{\ell^2} 
\\ & \qquad \quad \quad \ \ + \left \Vert  \vert {\alpha_{M,1}(s)}  \vert^2  \varphi_1(s) -  \vert  {\alpha_{M,2}(s)}  \vert^2 \varphi_2(s) \right \Vert_{\ell^2}  \bigg) ds.  
\end{align*}
We begin by considering the first term. Using \eqref{amphi} with \( k=0 \), along with \eqref{alpha} and \eqref{diffalpha}, we get 
\begin{align*}
 &\left \Vert  \overline{\alpha_{M,1}(s)} a_M \varphi_1(s) -  \overline{\alpha_{M,2}(s)}  a_M \varphi_2(s) \right  \Vert_{\ell^2}
 \\ &\quad\leq  \vert \alpha_{M,1}(s)-\alpha_{M,2}(s)\vert  \Vert a_M \varphi_1(s)\Vert_{\ell^2} +\vert \alpha_{M,2}(s) \vert  \Vert a_M(\varphi_1(s)-\varphi_2(s))\Vert_{\ell^2}
 \\ &\quad \leq  3 {M} c^2 \Vert \varphi_1(s) -\varphi_2(s)\Vert_{\ell^2} .
\end{align*}
Similarly, for the second term, applying \eqref{amphi} with \( k=0 \), along with \eqref{alpha} and \eqref{diffalpha}, we obtain
\begin{align*}
 \left \Vert  {\alpha_{M,1}(s)} a_M^* \varphi_1(s) -  {\alpha_{M,2}(s)}  a_M^* \varphi_2(s) \right  \Vert_{\ell^2}
\leq  3 {M} c^2 \Vert \varphi_1(s) -\varphi_2(s)\Vert_{\ell^2} .
\end{align*}
Finally, for the last term, by using \eqref{amphi} with \( k=0 \), along with \eqref{alpha} and \eqref{diffsquarealpha}, we obtain 
\begin{align*}
 &\left \Vert  \vert {\alpha_{M,1}(s)}  \vert^2  \varphi_1(s) -  \vert  {\alpha_{M,2}(s)}  \vert^2 \varphi_2(s) \right \Vert_{\ell^2}
 \\ &\quad\leq\left \vert \vert {\alpha_{M,1}(s)}  \vert^2-\vert {\alpha_{M,2}(s)}  \vert^2\right \vert  \Vert \varphi_1(s)\Vert_{\ell^2} + \vert {\alpha_{M,2}(s)}  \vert^2 \Vert \varphi_1(s) -\varphi_2(s)\Vert_{\ell^2}
 \\ &\quad \leq  5 {M} c^4 \Vert \varphi_1(s) -\varphi_2(s)\Vert_{\ell^2} .
\end{align*}
To summarize, we obtain
 \begin{align}\label{l2part}
 \Vert  \Gamma_M(\varphi_1)(t) - \Gamma_M (\varphi_2)(t)\Vert_{\ell^2}  \leq MT \vert J\vert (6c^2 +5c^4)\vert \vert  \vert\varphi_1 -\varphi_2\vert \vert  \vert.
\end{align}
More generally, for any $k\geq 0$, we have 
\begin{align}\label{Nkpart}
 \Vert \mathcal{N}^{k} \left(  \Gamma_M(\varphi_1)(t) - \Gamma_M (\varphi_2)(t) \right) \Vert_{\ell^2}  \leq MT \vert J\vert (6{M}^kc^2 +5c^4) \sup_{[0,T]}  \Vert\varphi_1(s) -\varphi_2(s)\Vert_{\mathcal{D}(\mathcal{N}^k)}.
\end{align}
Specifically, for the other component of the norm, we have
\begin{align}\label{Npart}
 \Vert \mathcal{N}^{2} \left(  \Gamma_M(\varphi_1)(t) - \Gamma_M (\varphi_2)(t) \right) \Vert_{\ell^2}  \leq MT \vert J\vert (6{M}^2c^2 +5c^4)\vert \vert  \vert\varphi_1 -\varphi_2\vert \vert  \vert.
\end{align}
By combining the two estimates \eqref{l2part} and \eqref{Npart} above, we obtain
 \begin{align*}
 \vert \vert \vert   \Gamma_M(\varphi_1) - \Gamma_M (\varphi_2)\vert \vert \vert   \leq MT \vert J\vert (6c^2+6c^2{M}^2 +10c^4)\vert \vert  \vert\varphi_1 -\varphi_2\vert \vert  \vert.
\end{align*}
Considering the semilinear equation of the form \eqref{pdeM}, and noting that our nonlinearity is Lipschitz continuous (or can be made a contraction by choosing \( T \) sufficiently small), we can approach the problem in two ways. First, we can apply the local well-posedness results from \cite{cazenave1998introduction}, specifically \cite[Lemma 4.3.2 and Proposition 4.3.3]{cazenave1998introduction}, to obtain a unique local  solution. Then, we can extend this solution globally using \cite[Theorem 4.3.4]{cazenave1998introduction} by employing conservation laws, including  the norm and the moment bounds (i)-(v) in Lemma \ref{conservationlaws}. On the other hand, we can establish global well-posedness by directly applying the Banach fixed point arguments. To this end, we consider the closed ball on the Banach space $X$ defined by 
\[B_X(\tilde \varphi_0,R]:=\{\varphi \in X; \quad \vert \vert \vert \varphi-\tilde \varphi_0 \vert \vert \vert  \leq R\}.\] 
Then, we check that for $T>0$ small enough and for $R>0$ large enough, the map $\Gamma_M$ satisfies the condition of the Banach fixed point theorem, namely
\begin{itemize}
\item $\Gamma_M$ maps $B_X(\tilde \varphi_0,R]$ into itself,
\item $\Gamma_M:(X,\vert \vert \vert  \cdot \vert \vert \vert ) \to (X,\vert \vert \vert  \cdot \vert \vert \vert )$ is a contraction map, 
\end{itemize}
guaranteeing the existence of a fixed point ($\varphi_M=\Gamma_M(\varphi_M) \in X$). The solution can then be extended globally using the same conservation laws employed for the first approach. 

\begin{remark}
In the above theorem, we can also apply the fixed point theorem in the Banach space \( X = C([0,T], \ell^2(\C)) \), which guarantees the existence of a unique local solution \( \varphi_M \in C([0,T], \ell^2(\C)) \). Subsequently, we can globalize the solution (which is equivalent to obtaining an estimate of the norm $\Vert \varphi_M (t) \Vert_{\ell^2} $ on  $[0,T]$) for the set of initial data \( \varphi_0 \in \mathcal{D}(\mathcal{N}^2) \), ensuring that 
\[
\Vert \varphi_M (t)\Vert_{\ell^2} = \Vert \varphi_0\Vert_{\ell^2},
\]
 where the above equality is proved in Lemma \ref{conservationlaws}. The above estimate  implies that
\[
\lim_{t\uparrow T} \Vert \varphi_M (t)\Vert_{\ell^2} = \Vert \varphi_0\Vert_{\ell^2} < \infty,
\]
which  guarantees that the solution does not blow up in finite time and thus \( T = \infty \).
\end{remark}

\end{proof}

\subsection{Convergence}\label{convergenceapproximate}
Assume that $  \varphi_0 \in \mathcal{D}(\mathcal{N}^{k}) $, then by (i) and (v) from Lemma~\ref{conservationlaws}, we have that $(\varphi_M)_{M\in \mathbb{N}}$ and  $(\mathcal{N}^k\varphi_M)_{M\in \mathbb{N}}$ are   bounded sequences in the Hilbert space $\ell^2(\C)$. Then there exist a convergent subsequence still denoted by $(\varphi_{M})_{M\in \mathbb{N}}$ such that 
\begin{itemize}
\item $\varphi_{M}$ converges weakly to $\varphi$ and the limit is unique,
\item $\mathcal{N}^k\varphi_{M}$ converges weakly to $\mathcal{N}^k\varphi$ for all $k\in \mathbb{R}^+$.
\end{itemize}
As a consequence of this convergence, we have 
\begin{align*}
&\Vert \varphi\Vert_{\ell^2} \leq 	\liminf_	{M\rightarrow \infty}  \Vert \varphi_M\Vert_{\ell^2} ,
\\&  \Vert \mathcal{N}^k \varphi\Vert_{\ell^2} \leq 	\liminf_	{M\rightarrow \infty} \Vert \mathcal{N}^k \varphi_M\Vert_{\ell^2}.
\end{align*}
This in fact implies strong convergence.

\begin{lem}\label{strongconv}
Assume that $  \varphi_0 \in \mathcal{D}(\mathcal{N}^{k}) $ and let  $(\varphi_M)_{M}$ be a sequence of  solutions   to  \eqref{pdeM} with $\Vert \varphi_0\Vert_{\ell^2} =1$ and $\varphi$ its associated  weak limit. Then we have for all $k\geq 0$,
\begin{align}
\Vert \mathcal{N}^k( \varphi_{M}(t) -\varphi(t))  \Vert_{\ell^2}  \underset{M\rightarrow \infty}{\longrightarrow} 0.
\end{align}
\end{lem}

\begin{proof}
This follows from the \(\text{weak}^*\) convergence in the Banach space \(\mathcal{L}^1(\ell^2(\C)) = (\mathcal{K}(\ell^2(\C)))^*\), where $\mathcal{L}^1(\ell^2(\C)) $ and  \(\mathcal{K}(\ell^2(\C))\) denote the space of trace-class  and compact operators, respectively. Let \(p_{\varphi_M} = |\varphi_M\rangle \langle \varphi_M|\) be the projection onto the state \(\varphi_M\), so that in particular \(p_{\varphi_M}^2 = p_{\varphi_M}\). We also have  
\[
\text{Tr}(p_{\varphi_M}) = \|\varphi_M\|_{\ell^2}^2 = 1, \qquad \text{Tr}(\mathcal{N}^k p_{\varphi_M}) = \langle \varphi_M, \mathcal{N}^k \varphi_M \rangle <C^\prime,
\]
 for some  $C'>0$, depending on $k$, $t$, $J$, and $\varphi_0$, but independent of $M$, as given in part~(v) of Lemma~\ref{conservationlaws}. This ensures the existence of a subsequence, still denoted by \((p_{\varphi_M})_M\), such that  
\[
p_{\varphi_M} \stackrel{*}{\rightharpoonup} \nu \quad \text{as } M \to \infty \quad \text{ weakly * in } \mathcal{L}^1(\ell^2(\C)),
\]
\[
\mathcal{N}^k p_{\varphi_M} \stackrel{*}{\rightharpoonup} \mathcal{N}^k \nu \quad \text{as } M \to \infty \quad \text{weakly * in } \mathcal{L}^1(\ell^2(\C)).
\]
For any compact operator \(B \in \mathcal{K}(\ell^2(\C))\), this implies
\[
\text{Tr}(p_{\varphi_M} B) \to \text{Tr}(\nu B) \quad \text{and} \quad \text{Tr}(\mathcal{N}^k p_{\varphi_M} B) \to \text{Tr}(\mathcal{N}^k \nu B) \quad \text{as } M \to \infty.
\]
For \(k \geq 0\), this leads to
\begin{align} \label{normconv}
\text{Tr}(\mathcal{N}^k p_{\varphi_M}) = \text{Tr}((\mathcal{N}+1)^{-1} (\mathcal{N}+1) \mathcal{N}^{k} p_{\varphi_M}) \to \text{Tr}((\mathcal{N}+1)^{-1} (\mathcal{N}+1) \mathcal{N}^{k} \nu) = \text{Tr}(\mathcal{N}^k \nu).
\end{align}
Specifically, for $k=0$ we have
\[
\text{Tr}(p_{\varphi_M}) = \|\varphi_M\|_{\ell^2}^2 \to \text{Tr}(\nu),
\]
which implies \(\text{Tr}(\nu) = 1\).
Now, using results from \cite{simon2005trace}, we obtain strong convergence for all \(k \geq 0\), 
\[
\|\mathcal{N}^k p_{\varphi_M} - \mathcal{N}^k \nu\|_{\mathcal{L}^1} ={\rm Tr} \left( \left\vert \mathcal{N}^k p_{\varphi_M} - \mathcal{N}^k \nu\ \right \vert \right)\to 0 \quad \text{as } M \to \infty.
\]
Since both \(p_{\varphi_M}\) and \(\nu\) are bounded in norm, we conclude that \(p_{\varphi_M}^2 = p_{\varphi_M}\) converges strongly to \(\nu = \nu^2\). Therefore, \(\nu\) is a projection, and \(\nu = P_\chi=|\chi \rangle \langle \chi|\). Additionally, we have
\[
\text{Tr}(p_{\varphi_M} \nu) = |\langle \varphi_M, \chi \rangle|^2 \to \text{Tr}(\nu^2) = \text{Tr}(\nu) = 1.
\]
On the other hand, we also know
\[
\langle \varphi_M, \chi \rangle \to \langle \varphi, \chi \rangle,
\]
which implies \(|\langle \varphi, \chi \rangle|^2 = 1\), leading to \(P_\varphi = P_\chi\). Therefore, by exploiting \eqref{normconv},  we have 
\[
\mathcal{N}^k \varphi_M \rightharpoonup \mathcal{N}^k \varphi \quad \text{and} \quad \|\mathcal{N}^k \varphi_M\|_{\ell^2} \to \|\mathcal{N}^k \varphi\|_{\ell^2}.
\]
Since \(\ell^2\) is a Hilbert space, these results imply strong convergence
\[
\|\mathcal{N}^k (\varphi_M - \varphi)\|_{\ell^2} \to 0 \quad \text{as } M \to \infty.
\]
\end{proof}

Next, we show that the limit indeed satisfies the corresponding mean-field equation.
\begin{lem}\label{Duhamelsoln}
Assume that $  \varphi_0 \in \mathcal{D}(\mathcal{N}^{2}) $. Let  $(\varphi_M)_{M}$ be a sequence of  solutions  to  \eqref{pdeM} and assume $\Vert \varphi_0\Vert_{\ell^2} =1$. Then the limit $\varphi$ satisfies the Duhamel version \eqref{DuhamelMF}  of the mean-field dynamics.

\end{lem}

\begin{proof}
Let us start by establishing some useful estimates. Since both $\varphi_{M}$ and $\mathcal{N}^k\varphi_{M}$  converge weakly to $\varphi$ and $\mathcal{N}^k \varphi$, respectively, we have 
\begin{align}
& \Vert \varphi  \Vert_{\ell^2} \leq \liminf_{M \rightarrow \infty} \Vert \varphi_M \Vert_{\ell^2}=\Vert \varphi_0\Vert_{\ell^2} =1,
\\ & \Vert \mathcal{N}^{1/2} \varphi  \Vert_{\ell^2} \leq \liminf_{M \rightarrow \infty} \Vert  \mathcal{N}^{1/2} \varphi_M \Vert_{\ell^2}=\Vert \mathcal{N}^{1/2} \varphi_0\Vert_{\ell^2} ,
\\ & \vert \alpha_{\varphi} \vert =\vert \langle \varphi, a\varphi \rangle \vert \leq \Vert \varphi \Vert_{\ell^2} \Vert a\varphi \Vert_{\ell^2} \leq \Vert \mathcal{N}^{1/2} \varphi_0\Vert_{\ell^2}. \label{eq:alpha bound}
\end{align}
Moreover, we also have  
\begin{align*}
\vert \alpha_M-\alpha_{\varphi}\vert & = \vert \langle \varphi_M, a_M \varphi_M \rangle - \langle \varphi, a \varphi \rangle\vert
\\ & \leq  \vert \langle \varphi_M, (a_M-a) \varphi_M \rangle\vert +\vert \langle \varphi_M, a (\varphi_M-\varphi) \rangle\vert +\vert \langle (\varphi_M-\varphi), a \varphi \rangle\vert
\\ & \leq \Vert \varphi_M \Vert_{\ell^2} \Vert a \mathds{1}_{\mathcal{N}> M} \varphi_M \Vert_{\ell^2} +\Vert a^* \varphi_M \Vert_{\ell^2} \Vert \varphi_M- \varphi \Vert_{\ell^2} +  \Vert a\varphi\Vert_{\ell^2}  \Vert \varphi_M- \varphi \Vert_{\ell^2} 
\\ & \leq 2 \Vert  (\mathcal{N}+1)^{1/2}\varphi_0 \Vert_{\ell^2} \Vert \varphi_M- \varphi \Vert_{\ell^2} +\cE_1(M), 
\end{align*}
where we have introduced $\cE_1(M)$ as 
\begin{equation}\label{firsterror}
\cE_1(M):= \Vert a \mathds{1}_{\mathcal{N}>M} \varphi_M \Vert_{\ell^2} .
\end{equation}
Then we estimate 
\begin{align}
& \bigg\Vert     \varphi(t) -  \tilde  \varphi_0(t)-i\int_0^t e^{-i(t-s)A} F(\varphi(s)) \, ds  \bigg\Vert_{\ell^2} 
\\ &\quad \leq \left \Vert \varphi_M(t) - \varphi(t) \right \Vert_{\ell^2} \label{estappsoln}
\\ &\qquad + \left \Vert \varphi_M(t) -   \tilde  \varphi_0(t)-i\int_0^t e^{-i(t-s)A} F_M(\varphi_M(s)) \, ds \right \Vert_{\ell^2}\label{solnappro} 
\\ &\qquad +  \left \Vert \int_0^t e^{-i(t-s)A} F_M(\varphi_M(s)) \, ds-\int_0^t e^{-i(t-s)A} F(\varphi(s)) \, ds\right \Vert _{\ell^2}\label{estnonlinearterm}
\end{align}
The first term  \eqref{estappsoln} converges to zero by Lemma \ref{strongconv}. The second term \eqref{solnappro} is zero because $\varphi_M$ is a solution to the approximated problem \eqref{Duhamel}. 
It remains to estimate the difference between the nonlinear parts,
\begin{subequations}
\begin{align}
 \eqref{estnonlinearterm} \leq  \vert J\vert \int_0^t  & \Big ( \left \Vert  \overline{\alpha_M (s)}  a_M \varphi_M(s) -  \overline{ \alpha_\varphi (s)} a \varphi(s) \right  \Vert_{\ell^2} \label{annterm} 
\\ & +  \left \Vert  \alpha_M (s) a_M^* \varphi_M(s) -  \alpha_\varphi (s) a^* \varphi(s) \right  \Vert_{\ell^2} \label{crtterm}  
\\ & + \left \Vert  \vert \alpha_M (s)\vert^2  \varphi_M(s) -  \vert  \alpha_\varphi (s)\vert^2  \varphi(s) \right \Vert_{\ell^2}  \Big ) ds. \label{sqrterm} 
\end{align}
\end{subequations}
For \eqref{annterm}, we find
\begin{align*}
\eqref{annterm} &\leq  \vert \alpha_M-\alpha\vert  \Vert a_M \varphi_M \Vert_{\ell^2} + \vert \alpha \vert \Vert (a_M- a) \varphi_M \Vert_{\ell^2} + \vert \alpha \vert \Vert a(\varphi_M-\varphi)\Vert_{\ell^2}
\\ & \leq \Vert \mathcal{N}^{1/2} \varphi_0\Vert_{\ell^2} \left( 2 \Vert  (\mathcal{N}+1)^{1/2}\varphi_0 \Vert_{\ell^2} \Vert \varphi_M- \varphi \Vert_{\ell^2}+\Vert \mathcal{N}^{1/2}(\varphi_M-\varphi)\Vert_{\ell^2} +2\cE_1(M) \right) ,
\end{align*}
with $\cE_1(M)$ from \eqref{firsterror}.  The first and the second term go to zero as $M \rightarrow \infty$. It remains to check $ \cE_1(M) \rightarrow 0$ as $M \rightarrow \infty$. Indeed by \eqref{mombounds},  we have  for some $C>0$ that
\begin{align*}
\cE_1(M)&= \Vert a\mathds{1}_{\mathcal{N}>M} \varphi_M \Vert_{\ell^2}
\\ & \leq \Vert \mathds{1}_{\mathcal{N}+1>M}  (\mathcal{N}+1)^{-1/2}\Vert_{\mathcal{L}} \Vert (\mathcal{N}+1)^{1/2} a \varphi_M \Vert_{\ell^2} \\ &\leq \frac{1}{\sqrt{M}} \underbrace{\left( \sum_{j=0}^{2}  \left(  2C |J|  \Vert \mathcal{N}^{1/2}\varphi_0 \Vert_{\ell^2} \right)^j   \left \langle \varphi_0,(\mathcal{N}+j)^{2-\frac{j}{2} }\varphi_0 \right \rangle \frac{t^j}{j!}\right)^{1/2}}_{<\infty \ \text{since } \varphi_0\in \mathcal{D}(\mathcal{N}^2)}.
\end{align*}
So, the term $\cE_1(M) $ goes to zero as $M \rightarrow \infty$.
Similarly, for \eqref{crtterm} we find that 
\begin{align*}
\eqref{crtterm} &\leq  \vert \alpha_M-\alpha\vert  \Vert a^*_M \varphi_M \Vert_{\ell^2} + \vert \alpha \vert \Vert (a^*_M- a^*) \varphi_M \Vert_{\ell^2} + \vert \alpha \vert \Vert a^*(\varphi_M-\varphi)\Vert_{\ell^2}
\\ & \leq 2  \Vert  (\mathcal{N}+1)^{1/2}\varphi_0 \Vert_{\ell^2}^2  \Vert \varphi_M- \varphi \Vert_{\ell^2}+ \Vert \mathcal{N}^{1/2} \varphi_0\Vert_{\ell^2} \Vert (\mathcal{N}+1)^{1/2}(\varphi_M-\varphi)\Vert_{\ell^2}
\\ &\quad  +\Vert( \mathcal{N}+1)^{1/2} \varphi_0\Vert_{\ell^2}\cE_1(M)+\Vert \mathcal{N}^{1/2} \varphi_0\Vert_{\ell^2}\cE_2(M)    ,
\end{align*}
where we have introduced
\begin{equation}\label{error2}
\cE_2(M):=\Vert \mathds{1}_{\mathcal{N}>M} a^* \varphi_M \Vert_{\ell^2}.
\end{equation}
By the same arguments as for \eqref{annterm}, the term \eqref{crtterm} goes to zero as $M \rightarrow \infty$. It remains to estimate the last term \eqref{sqrterm}. We have 
\begin{align*}
\eqref{sqrterm} &= \left \Vert  \vert \alpha_M \vert^2  \varphi_M -  \vert  \alpha_\varphi \vert^2  \varphi \right \Vert_{\ell^2} 
\\ & \leq \left  \vert  \vert \alpha_M \vert^2 -   \vert  \alpha_\varphi \vert^2 \right \vert  \Vert  \varphi_M  \Vert_{\ell^2}  +   \vert  \alpha_\varphi \vert^2 \Vert \varphi_M -\varphi \Vert_{\ell^2} 
\\ & \leq  \left  \vert  \alpha_M  (\overline{ \alpha_M}-\overline{\alpha_\varphi}) +\overline{\alpha_\varphi} ( \alpha_M-\alpha_\varphi) \right \vert +   \vert  \alpha_\varphi \vert^2 \Vert \varphi_M -\varphi \Vert_{\ell^2}
\\ & \leq    \left( \vert  \alpha_M \vert +\vert  \alpha_\varphi \vert\right)    \vert  \alpha_M-\alpha_\varphi\vert  +   \vert  \alpha_\varphi \vert^2 \Vert \varphi_M -\varphi \Vert_{\ell^2},
\end{align*}
which also converges to zero by the same arguments as for the previous two terms.
\end{proof}

\begin{proof}[Proof of Theorem \ref{gwp}]
By Lemma \ref{Duhamelsoln}, there exists a solution $\varphi(\cdot)$ to the Duhamel formulation of the mean-field dynamics \eqref{DuhamelMF}. Let $T>0$. Then, by parts (i) and (v) of Lemma \ref{conservationlaws}, we have
\begin{align*}
    \varphi(\cdot) \in L^1((0,T), \ell^2(\mathbb{C}))\quad  \text{ and }  \quad \varphi(\cdot) \in L^1((0,T), \mathcal{D}(\mathcal{N}^2)).
\end{align*}
Moreover, by parts (i), (iv) and (ii) of Lemma \ref{conservationlaws}, it follows that
\[ F(\varphi(\cdot)) \in L^1((0,T), \ell^2(\mathbb{C}))\]
since 
\begin{align*}
     \Vert F(\varphi (t))\Vert_2 &= \left \Vert -J \left( \alpha_{\varphi(t)} a^* + \overline{\alpha_{\varphi(t)} } a - \vert \alpha_{\varphi(t)} \vert^2 \right)\varphi(t)  \right \Vert_2 
     \\ & \leq  2|J|  \Vert a \varphi (0) \Vert_2 \Vert a^*\varphi (0) \Vert_2 + |J|  \Vert a \varphi (0) \Vert_2^2.
\end{align*}
Finally, applying \cite[Proposition 4.1.9]{cazenave1998introduction}, we conclude that the solution to the Duhamel equation \eqref{DuhamelMF} yields a solution
\[ \varphi(\cdot) \in L^1((0,T), \mathcal{D}(\mathcal{N}^2) ) \cap W^{1,1}((0,T), \ell^2(\mathbb{C}) ) \]
of the mean-field dynamic \eqref{pde} for almost every $t \in (0,T).$ Finally using the conservation laws stated in Lemma \ref{conservationlaws}, the local solution can be extended globally in time, yielding
\[ \varphi(\cdot) \in L^1(\mathbb{R}^+, \mathcal{D}(\mathcal{N}^2) ) \cap W^{1,1}(\mathbb{R}^+, \ell^2(\mathbb{C}) ).\]
\end{proof}

\section{Preliminaries}\label{sec_preliminaries}

\subsection{Reduced Densities Matrices}
    Given a density matrix $\gamma_d \in \mathcal{L}^1\prth{\mathcal{F}}$ we define a two-lattice-site reduced density matrix
    \begin{align*}
        \gamma_d^{\prth{2}} 
            \coloneq \frac{1}{d \abs{\Lambda}} \sum_{\substack{\braket{x, y}}} \PTr{\Lambda\backslash{\sett{x, y}}}{\gamma_d}
            =\frac{1}{2d \abs{\Lambda}} \sum_{\substack{x, y \in \Lambda \\ x\sim y}} \PTr{\Lambda\backslash{\sett{x, y}}}{\gamma_d}.
    \end{align*}
    Note that this two-lattice-site reduced density matrix is symmetrized over all interacting pairs of sites and not over all pairs of sites. The normalization factor $2d \abs{\Lambda}$ is indeed the number of interacting pairs of sites.

    Note that $\gTwo$ is symmetric under the exchange of its two variables, i.e.,
        \begin{align*}
            &\forall A, B \in \mathcal{L}(\ell^2(\C)), \Tr{\gTwo  A\otimes B} = \Tr{\gTwo  B\otimes A},
        \end{align*}    
    and reduces to $\gOne$, i.e., 
        \begin{align*}
            &\PTr{1}{\gTwo} = \PTr{2}{\gTwo} = \gOne.
        \end{align*}  
    Moreover, if $C \in \mathcal{L}(\ell^2(\C))$ and $D \in \mathcal{L}(\ell^2(\C)^{\otimes 2})$, then it follows directly from its definition that
        \begin{align*}
            &\frac{1}{\abs{\Lambda}}\sum_{x\in\Lambda}\Tr{\gamma_d C_x} = \Tr{\gOne C}, \\
            &\frac{1}{2d\abs{\Lambda}}\sum_{\substack{x, y\in \Lambda \\ x \sim y}}  \Tr{\gamma_d D_{x, y}} = \Tr{\gTwo D}.
        \end{align*}
        
Furthermore, the following standard results hold. If $A \in \mathcal{L}(\ell^2(\C))$ is self adjoint such that $A \ge 0$ or $\gOne A \in \mathcal{L}^1(\ell^2(\C))$, then
        \begin{align}
            \sbra{\sum_{x\in\Lambda} A_x, \gamma_d}^{(1)} = \sbra{A, \gOne }.\label{eq:k=1 partial trace 1} 
        \end{align}
If $B \in \mathcal{L}(\ell^2(\C)^{\otimes 2})$ is self adjoint such that $B \ge 0$ or $\gTwo B \in \mathcal{L}^1(\ell^2(\C)^{\otimes 2})$, then
        \begin{align}
            \dfrac{1}{2d}\sbra{\sum_{\substack{x, y\in \Lambda \\ x \sim y}} B_{x, y}, \gamma_d}^{(1)} = \PTr{1}{\sbra{B, \gTwo}} + \PTr{2}{\sbra{B, \gTwo}}.
             \label{eq:k=1 partial trace 2} 
        \end{align}

\subsection{Energy Bounds}    
    With the definitions of one and two-lattice-site density matrices we can rewrite the energy per lattice site as
    \begin{align}
            \frac{\Tr{\gamma_d H_d}}{\abs{\Lambda}}
                = \Tr{\gOne\prth{\prth{J - \mu}\mathcal{N} + \frac{U}{2}\mathcal{N}\prth{\mathcal{N}-1}}} - J\Tr{\gTwo a^*\otimes a}. \label{eq:RDM energy}
    \end{align}    
    Note that the mean-field energy can be written as
    \begin{align}
             \braket{\varphi, h^\varphi \varphi} = J \prth{\braket{\varphi, \mathcal{N} \varphi} - \abs{\alpha_\varphi}^2} - \mu \braket{\varphi, \mathcal{N} \varphi} + \frac{U}{2}\braket{\varphi, \mathcal{N}(\mathcal{N} - 1) \varphi} \label{eq:E mf}.
    \end{align}

    The following bounds allow us to control the Bose--Hubbard energy and the mean-field energy in terms of moments of the number operator.
\begin{lem}
Let $\gamma_d \in \mathcal{L}^1\prth{\mathcal{F}}$ and $\varphi \in \ell^2(\C)$. Then there exists $C>0$ such that, for $U = 0$,
        \begin{align}
                \left|\braket{\varphi, h^\varphi \varphi}\right| &\le C \braket{\varphi, \mathcal{N} \varphi}, \label{eq:U=0 mf bound} \\
                 \frac{\left|\Tr{\gamma_d H_d}\right|}{\abs{\Lambda}} &\le C \left(1 + \Tr{\gOne \mathcal{N}}\right), \label{eq:U=0 bound}
         \end{align}
        and, for $U \neq 0$,
        \begin{align}
                \left|\braket{\varphi, h^\varphi \varphi} \right| &\le C \left( 1+ \braket{\varphi, \mathcal{N}^2 \varphi}\right), \label{eq:U!=0 mf bound} \\
                \frac{\left|\Tr{\gamma_d H_d}\right|}{\abs{\Lambda}} &\le C \left( 1 + \Tr{\gOne\mathcal{N}^2}\right). \label{eq:U!=0 bound}
         \end{align}
\end{lem}

    \begin{proof}

 Recalling \eqref{eq:alpha bound} and \eqref{eq:E mf} immediately yields \eqref{eq:U=0 mf bound} and \eqref{eq:U!=0 mf bound}. In order to obtain \eqref{eq:U=0 bound} and \eqref{eq:U!=0 bound}, we estimate the two-site term in \eqref{eq:RDM energy} with Cauchy--Schwarz to obtain
        \begin{align*}
            \abs{\Tr{\gTwo a^*\otimes a}}
                \le \Tr{\gOne \mathcal{N}}^{\frac{1}{2}} \Tr{\gOne \prth{\mathcal{N} + 1}}^{\frac{1}{2}} 
                \le \Tr{\gOne \mathcal{N}} + 1.
        \end{align*}
     \end{proof}

\subsection{Conservation Laws} \label{sec:conversation laws}

    For both the Bose--Hubbard model \eqref{eq:HBH} and the mean-field model \eqref{eq:h mf} the total particle number and the total energy are conserved. Furthermore, one can control higher powers of the total particle number. Let us show this first for the mean-field equation. The total particle number is conserved since
    \begin{align}\label{mf_N_conservation}
            i \partial_t \bk{\varphi}{\mathcal{N}\varphi} = \bk{\varphi}{\sbra{\mathcal{N}, h^\varphi}\varphi}
            &= - J \prth{\alpha_\varphi \bk{\varphi}{\sbra{\mathcal{N}, a^*}\varphi} + \overline{\alpha_\varphi} \bk{\varphi}{\sbra{\mathcal{N}, a}\varphi}} \nonumber\\
            &= - J \prth{\alpha_\varphi \bk{\varphi}{a^* \varphi} - \overline{\alpha_\varphi}\bk{\varphi}{a \varphi}} \nonumber\\
            &= -J\prth{\abs{\alpha_\varphi}^2 - \abs{\alpha_\varphi}^2} \nonumber\\
            &= 0.
    \end{align}
    The energy is conserved since 
    \begin{align*}
            i\partial_t \bk{\varphi}{h^\varphi \varphi} 
                = \bk{\varphi}{\partial_t h^\varphi \varphi}
                &= -J \bk{\varphi}{\prth{\partial_t \alpha_\varphi a^* + \partial_t \overline{\alpha_\varphi} a - \overline{\alpha_\varphi} \partial_t \alpha_\varphi  - \alpha_\varphi \partial_t\overline{\alpha_\varphi}} \varphi} \\
                &= - J\prth{ \overline{\alpha_\varphi} \partial_t \alpha_\varphi + \alpha_\varphi \partial_t\overline{\alpha_\varphi} - \overline{\alpha_\varphi} \partial_t \alpha_\varphi  - \alpha_\varphi \partial_t\overline \alpha_\varphi} \\
                &= 0.
    \end{align*}    
    Moreover, we can prove two different bounds for controlling powers of the number operator, which we will use for our two main theorems.
    
    \begin{proposition}\label{prop:moments bound}
        Let $\varphi$ solve \eqref{eq:MFS} with $\varphi(0) \in \ell^2(\C)$.
        Let $k\in \mathbb{N}/2, k\ge 1$ and $t\in\mathbb{R}_+$. Then
                \begin{align}
                    &\Tr{p_\varphi(t)\mathcal{N}^k} \le \prth{\Tr{p_\varphi(0) \mathcal{N}^k} + e^{-1}k^k}e^{2eJk\Tr{p_\varphi(0) \mathcal{N}}^{\frac{1}{2}}t}, \label{eq:mf moment bound 1}\\
                    &\Tr{p_\varphi(t)\mathcal{N}^k} \le \sum_{l=0}^{2(k-1)} \prth{\matrx{2k \\ l}} \prth{J \Tr{p_\varphi(0)\mathcal{N}}^{\frac{1}{2}} t}^l \Tr{p_\varphi(0)\prth{\mathcal{N}+l}^{k - \frac{l}{2}}}. \label{eq:mf moment bound 2}
                \end{align}
    \end{proposition}
    
    \begin{proof} Let $n\in \N$.
    
    \textbf{Proof of an intermediate bound on $\partial_t \bk{\varphi}{\prth{\mathcal{N}+n}^k \varphi}$.}
    
    Recalling the mean-field dynamics \eqref{eq:MFS}, we find
        \begin{align}\label{Gronwall_N_n}
            i\partial_t \bk{\varphi}{\prth{\mathcal{N}+n}^k \varphi}
                & = \bk{\varphi}{\sbra{\prth{\mathcal{N}+n}^k, h^\varphi}\varphi} \nonumber\\
                &= - J  \bk{\varphi}{\sbra{\prth{\mathcal{N}+n}^k, \alpha_\varphi \ad + \overline{\alpha_\varphi} a}\varphi} \nonumber\\
                &= J \alpha_\varphi \bk{\varphi}{\sbra{a^*, \prth{\mathcal{N}+n}^k}\varphi} + J \overline{\alpha_\varphi} \bk{\varphi}{\sbra{a, \prth{\mathcal{N}+n}^k}\varphi} \nonumber\\
                &= 2iJ \text{ Im}\sbra{\overline{\alpha_\varphi} \bk{\varphi}{\sbra{a, \prth{\mathcal{N}+n}^k}\varphi}}.
        \end{align}
        Now, let $\mathcal{A} \in \mathcal{L}^1\prth{\ell^2(\C)}$ be the positive operator defined as
        \begin{align}
            \mathcal{A}^{k-1} \coloneq \prth{\mathcal{N}+n+1}^k - \prth{\mathcal{N} + n}^k \le k\prth{\mathcal{N}+n+1}^{k-1} \label{eq:MB trick}.
        \end{align}
        Since $a \prth{\mathcal{N}+n}^k = \prth{\mathcal{N}+n+1}^k a$ we find
        \begin{align*}
            \bk{\varphi}{\sbra{a, \prth{\mathcal{N}+n}^k}\varphi}
                = \bk{\varphi}{\mathcal{A}^{k-1} a \varphi}
                = \bk{\mathcal{A}^{\frac{k}{2} - \frac{1}{4}}\varphi}{ \mathcal{A}^{\frac{k}{2} - \frac{3}{4}}a \varphi}
        \end{align*}
        so with Cauchy--Schwarz's inequality,
        \begin{align}\label{commutator_N_n_a}
             \abs{\bk{\varphi}{\sbra{a, \prth{\mathcal{N}+n}^k}\varphi}} 
                &\le \bk{\varphi}{\mathcal{A}^{k - \frac{1}{2}} \varphi}^{\frac{1}{2}} \bk{\varphi}{a^* \mathcal{A}^{k - \frac{3}{2}}a\varphi}^{\frac{1}{2}} \nonumber\\
                &\le k \bk{\varphi}{\prth{\mathcal{N}+n+1}^{k-\frac{1}{2}}\varphi}^{\frac{1}{2}} \bk{\varphi}{a^* \prth{\mathcal{N}+n+1}^{k-\frac{3}{2}}a \varphi}^{\frac{1}{2}} \nonumber\\
                &= k \bk{\varphi}{\prth{\mathcal{N}+n+1}^{k-\frac{1}{2}}\varphi}^{\frac{1}{2}} \bk{\varphi}{\prth{\mathcal{N}+n}^{k-\frac{3}{2}}\mathcal{N} \varphi}^{\frac{1}{2}} \nonumber\\
                &\le k \bk{\varphi}{\prth{\mathcal{N}+n+1}^{k-\frac{1}{2}}\varphi}.
        \end{align}
        Combining \eqref{Gronwall_N_n} with \eqref{commutator_N_n_a} and also \eqref{eq:alpha bound} we conclude
        \begin{align}
            \abs{\partial_t \bk{\varphi}{\prth{\mathcal{N}+n}^k \varphi}} 
                \le 2J k \bk{\varphi}{\mathcal{N}\varphi}^{\frac{1}{2}} \bk{\varphi}{\prth{\mathcal{N}+n+1}^{k-\frac{1}{2}}\varphi}. \label{eq:N pre gronwall}
        \end{align}
        
    \noindent\textbf{Proof of \eqref{eq:mf moment bound 2}.}
        By induction on $k$, we prove that, $\bk{\varphi(0)}{\mathcal{N}^{k}\varphi(0)} < \infty$ implies that for all $n \in \mathbb{N}$,
        \begin{align}
            \bk{\varphi(t)}{\prth{\mathcal{N}+n}^k \varphi(t)} \le \sum_{l=0}^{2(k-1)} \prth{\matrx{2k \\ l}} \prth{J \bk{\varphi(0)}{\mathcal{N}\varphi(0)}^{\frac{1}{2}} t}^l  \bk{\varphi(0)}{\prth{\mathcal{N}+n+l}^{k - \frac{l}{2}}\varphi(0)} \label{eq:induction moments}.
        \end{align}
        Then \eqref{eq:mf moment bound 2} follows for $n=0$. The inequality is indeed true for $k=1$ since $\bk{\varphi}{\prth{\mathcal{N} + n} \varphi}$ is conserved, see \eqref{mf_N_conservation}. For the induction step, we assume \eqref{eq:induction moments} holds for some $k$ and that
        \begin{align*}
            \bk{\varphi(0)}{\mathcal{N}^{k+\frac{1}{2}}\varphi(0)} < \infty,
        \end{align*}
        and we now prove \eqref{eq:induction moments} for $k+\frac{1}{2}$ instead of $k$. Using \eqref{eq:N pre gronwall} with $k+\frac{1}{2}$ instead of $k$, and using the conservation of $\bk{\varphi}{\mathcal{N} \varphi}$ we find
        \begin{align*}
            \abs{\partial_t \bk{\varphi}{\prth{\mathcal{N}+n}^{k+\frac{1}{2}} \varphi}}
                \le J (2k+1) \bk{\varphi(0)}{\mathcal{N}\varphi(0)}^{\frac{1}{2}} \bk{\varphi}{\prth{\mathcal{N}+n+1}^{k}\varphi}.
        \end{align*}
        Integrating over time and inserting \eqref{eq:induction moments} we conclude
        \begin{align*}
            &\bk{\varphi(t)}{\prth{\mathcal{N}+n}^{k+\frac{1}{2}} \varphi(t)} \\
                &\quad\le \bk{\varphi(0)}{\prth{\mathcal{N}+n}^{k+\frac{1}{2}}\varphi(0)} + J(2k+1) \bk{\varphi(0)}{\mathcal{N}\varphi(0)}^{\frac{1}{2}} \intr_0^t \bk{\varphi(\tau)}{\prth{\mathcal{N}+n+1}^k\varphi(\tau)} d\tau \\
                &\quad= \bk{\varphi(0)}{\prth{\mathcal{N}+n}^{k+\frac{1}{2}}\varphi(0)} +J(2k+1) \bk{\varphi(0)}{\mathcal{N}\varphi(0)}^{\frac{1}{2}}  \\
                    &\qquad \sum_{l=0}^{2(k-1)} \prth{\matrx{2k \\ l}} \prth{J \bk{\varphi(0)}{\mathcal{N}\varphi(0)}^{\frac{1}{2}}}^l \bk{\varphi(0)}{\prth{\mathcal{N}+n+l + 1}^{k - \frac{l}{2}}\varphi(0)} \intr_0^t \tau^l d\tau \\
                &\quad= \bk{\varphi(0)}{\prth{\mathcal{N}+n}^{k+\frac{1}{2}}\varphi(0)} \\
                    &\qquad+ \sum_{l=0}^{2(k-1)} \prth{\matrx{2k +1 \\ l + 1}} \prth{J \bk{\varphi(0)}{\mathcal{N}\varphi(0)}^{\frac{1}{2}} t}^{l+1} \bk{\varphi(0)}{\prth{\mathcal{N}+n+l + 1}^{k - \frac{l}{2}}\varphi(0)} \\
                &\quad= \bk{\varphi(0)}{\prth{\mathcal{N}+n}^{k+\frac{1}{2}}\varphi(0)} \\
                    &\qquad+ \sum_{l=1}^{2(k-1)+1} \prth{\matrx{2k +1 \\ l}} \prth{J \bk{\varphi(0)}{\mathcal{N}\varphi(0)}^{\frac{1}{2}} \tau}^l \bk{\varphi(0)}{\prth{\mathcal{N}+n+l}^{k +\frac{1}{2} - \frac{l}{2}}\varphi(0)} \\
                &\quad\le \sum_{l=0}^{2\prth{k+\frac{1}{2}-1}} \prth{\matrx{2\prth{k+\frac{1}{2}} \\ l}} \prth{J \bk{\varphi(0)}{\mathcal{N}\varphi(0)}^{\frac{1}{2}} t}^l  \bk{\varphi(0)}{\prth{\mathcal{N}+n+l}^{k + \frac{1}{2} - \frac{l}{2}}\varphi(0)}.
        \end{align*}
        which concludes the induction. 

    \noindent\textbf{Proof of \eqref{eq:mf moment bound 1}.}
        Since
        \begin{align*}
            N \ge 1 \implies \prth{N + 1}^k = N^k e^{k \ln\prth{1 + \frac{1}{N}}} \le N^k e^{\frac{k}{N}}
        \end{align*}
        we notice that
        \begin{align*}
            N \ge k \implies  \prth{N + 1}^k \le e N^k.
        \end{align*}
        Next, we continue from \eqref{eq:N pre gronwall} for $n=0$, and introduce a cutoff, to obtain
        \begin{align} \label{eq:tr pNk}
            \abs{\partial_t \Tr{p_\varphi(t)\mathcal{N}^k }} 
                &\le 2J k \Tr{p_\varphi(0) \mathcal{N}}^{\frac{1}{2}} \Tr{p_\varphi(t) \prth{\mathcal{N}+1}^k} \nonumber\\
                &= 2J k \Tr{p_\varphi(0) \mathcal{N}}^{\frac{1}{2}} \Tr{p_\varphi(t) \prth{\mathcal{N}+1}^k\prth{\mathds{1}_{\mathcal{N} < k} + \mathds{1}_{\mathcal{N} \ge k}}} \nonumber\\
                &\le 2J k \Tr{p_\varphi(0) \mathcal{N}}^{\frac{1}{2}} \prth{\Tr{p_\varphi(t) k^k} + e \Tr{p_\varphi(t) \mathcal{N}^k}} \nonumber\\
                &= 2J k \Tr{p_\varphi(0) \mathcal{N}}^{\frac{1}{2}} \prth{k^k + e \Tr{p_\varphi(t) \mathcal{N}^k}}.
        \end{align}
        With Gronwall's lemma we conclude that
        \begin{align*}
            \Tr{p_\varphi(t) \mathcal{N}^k} 
                &\le \prth{\Tr{p_\varphi(0) \mathcal{N}^k} + e^{-1}k^k}e^{2eJk\Tr{p_\varphi(0) \mathcal{N}}^{\frac{1}{2}}t}.
        \end{align*}
    \end{proof}

    For the Bose--Hubbard model \eqref{eq:HBH} the total energy $\Tr{\gamma_d H_d}$ and the total particle number $\text{Tr}(\gOne \mathcal{N})$ are conserved as well. Moreover, we can prove bounds analogous to the mean-field dynamics for powers of the total number of particles. Note first that we can rewrite the Hamiltonian $H_d$ as
    \begin{align}
            H_d &= \sum_{x\in \Lambda} h^{\varphi}_x - \frac{J}{2d} \sum_{\substack{x, y\in \Lambda \\ x \sim y}}\prth{a^*_x - \overline{\alpha}_\varphi}\prth{a_y - \alpha_\varphi}, \label{eq:H decomposition}
    \end{align}
    with $\alpha_\varphi = \langle \varphi, a \varphi \rangle$. Then, by using \eqref{eq:k=1 partial trace 1} and \eqref{eq:k=1 partial trace 2} we find that the one-lattice-site reduced density matrix satisfies
        \begin{align}
            i \partial_t \gOne 
                = \sbra{H_d, \gamma_d}^{(1)}
                &= \sbra{\sum_{x\in\Lambda} h_x^\varphi, \gamma_d}^{(1)} - \frac{J}{2d} \sbra{\sum_{\substack{x, y\in \Lambda \\ x \sim y}}\prth{a^*_x - \overline{\alpha}}\prth{a_y - \alpha}, \gamma_d}^{(1)} \nonumber\\
                &= \sbra{h^{\varphi}, \gOne } 
                    - J \PTr{2}{\sbra{\prth{a^* - \overline{\alpha}}\otimes\prth{a - \alpha} + \prth{a - \alpha}\otimes\prth{a^* - \overline{\alpha}}, \gTwo}}. \label{eq:gamma1 dynamics}
        \end{align}
    We have the following propagation bounds.
    \begin{proposition}\label{prop:gamma_d conservation laws}
        Let $\gOne$ solve \eqref{eq:gamma1 dynamics}, let $k\in \mathbb{N}/2, k\ge 1$, $t\in\mathbb{R}_+$, and $ \gOne(0)\mathcal{N}^k\in \mathcal{L}^1(\ell^2(\C))$. Then
        \begin{align}
            \Tr{\gOne(t)  \mathcal{N}^k}
                \le \prth{\Tr{\gOne (0) \mathcal{N}^k} + e^{-1}k^k}e^{2eJk t} \label{eq:BH moment bound}.
        \end{align}
    \end{proposition}
    \begin{proof} 
        Similarly to \eqref{eq:MB trick}, let us define
        \begin{align*}
            \mathcal{A}^{k-1} \coloneq \prth{\mathcal{N}+1}^k - \mathcal{N}^k \le k\prth{\mathcal{N}+1}^{k-1}.
        \end{align*}
        Denoting $A_1 \coloneq A\otimes \mathds{1}, A_2 \coloneq \mathds{1}\otimes A$ for an operator $A$ on $\ell^2(\C)$, Cauchy--Schwarz's inequality yields
        \begin{align*}
            \abs{\partial_t \Tr{\gOne \mathcal{N}^k}}
                &\le  2J \abs{\Tr{\gTwo \sbra{a_1, \mathcal{N}_1^k} a_2}} \\
                &= 2J \abs{\Tr{\gTwo a_2 \mathcal{A}_1^{k-1} a_1}} \\
                &\le 2J \Tr{\gTwo  a_2 \mathcal{A}_1^{k-1} a_2^*}^{\frac{1}{2}} \Tr{\gTwo  a_1^* \mathcal{A}_1^{k-1} a_1}^{\frac{1}{2}} \\
                &\leq 2J k \Tr{\gTwo \prth{\mathcal{N}_1+1}^{k-1}\prth{\mathcal{N}_2 + 1}}^{\frac{1}{2}} \Tr{\gOne \mathcal{N}^k}^{\frac{1}{2}}.
        \end{align*}
        Since $\sbra{\prth{\mathcal{N}_1+1}^{k-1}, \prth{\mathcal{N}_2 + 1}} = 0$, by Young's inequality,
        \begin{align*}
            \prth{\mathcal{N}_1+1}^{k-1} \prth{\mathcal{N}_2 + 1}
                &\le \prth{1-\frac{1}{k}}\prth{\mathcal{N}_1+1}^{k} + \frac{1}{k} \prth{\mathcal{N}_2 + 1}^k.
        \end{align*}
        Introducing a cutoff similarly to \eqref{eq:tr pNk}, we conclude that
        \begin{align*}
            \abs{\partial_t \Tr{\gOne \mathcal{N}^k}}
                &\le 2J k \prth{\prth{1-\frac{1}{k}}\Tr{\gTwo \prth{\mathcal{N}_1+1}^{k}} + \frac{1}{k} \Tr{\gTwo \prth{\mathcal{N}_2 + 1}^k}}^{\frac{1}{2}} \Tr{\gOne \mathcal{N}^k}^{\frac{1}{2}} \\
                &= 2Jk \Tr{\gOne \prth{\mathcal{N}+1}^k}^{\frac{1}{2}} \Tr{\gOne \mathcal{N}^k}^{\frac{1}{2}} \\
                &\le  2Jk \Tr{\gOne \mathcal{N}^k}^{\frac{1}{2}} \prth{k^k + e \Tr{\gOne  \mathcal{N}^k}}^{\frac{1}{2}} \\
                &\le 2Jk \prth{k^k + e \Tr{\gOne  \mathcal{N}^k}}.
        \end{align*}

        With Gronwall's lemma we conclude that
        \begin{align*}
            \Tr{\gOne (t) \mathcal{N}^k} 
                \le \prth{\Tr{\gOne (0) \mathcal{N}^k} + e^{-1}k^k}e^{2eJk t}.
        \end{align*}
    
    \end{proof}

\subsection{Gronwall Estimate}

    Both Theorems~\ref{th:moment method} and \ref{thm:excitation method} are proven via a Gronwall estimate for the quantity $\Tr{\gOne q}$. This is directly related to the trace norm difference of reduced density matrices, analogous to the case of the weak coupling limit \cite{knowles2010}, as the following Lemma shows.
    
    \begin{lem}\label{lem_q_dens_mat}
        Let $p$ be a rank one projection and $\gamma$ a positive trace $1$ operator on $\ell^2(\C)$ and $q \coloneq 1 - p$. Then
        \begin{align}
            &2 \Tr{\gamma q} \le \norm{\gamma - p}_{\mathcal{L}^1} \le 2\sqrt{2} \sqrt{\Tr{\gamma q}}. \label{eq:Tr norm estimate}
        \end{align}
    \end{lem}
    \begin{proof}
        In order to get the upper bound in \eqref{eq:Tr norm estimate}, we first notice that since $\gamma \le 1$ and $\Tr{\gamma } = \Tr{p} = 1$,
        \begin{align*}
            \norm{p\gamma p - p}_{\mathcal{L}^1} = \Tr{\prth{1 - \gamma }p} = 1 - \Tr{\gamma p} = \Tr{\gamma (1-p)} =\Tr{\gamma q},
        \end{align*}
        so
        \begin{align*}
            \norm{\gamma - p}_{\mathcal{L}^1} 
                &= \norm{(p+q)\gamma (p+q) - p}_{\mathcal{L}^1}  \\
                &\le 2 \Tr{\gamma q} + 2 \norm{q\gamma p}_{\mathcal{L}^1} \\
                &\le 2 \Tr{\gamma q} + 2\sqrt{\Tr{\gamma q}}\sqrt{\Tr{\gamma p}} \\
                &= 2 \sqrt{\Tr{\gamma q}}\prth{\sqrt{\Tr{\gamma q}} + \sqrt{1 - \Tr{\gamma q}}} \\
                &\le 2\sqrt{2}\sqrt{\Tr{\gamma q}},
        \end{align*}
        where we used $\sqrt{x} + \sqrt{1 - x} \le \sqrt{2}$ for $0 \le x \le 1$. 
        The lower bound follows directly from
        \begin{equation*}
            \Tr{\gamma q} = \Tr{\prth{p - \gamma }p} \leq \norm{\gamma - p}_{\mathcal{L}^1}.
        \end{equation*}
    \end{proof}

    Next, we compute the time derivative of $\Tr{\gOne q}$ and estimate some of the appearing terms. This is analogous to the estimates in the weak coupling limit, see, e.g., \cite[Lemma~3.2]{pickl2011}. The only term that causes technical difficulties is $\Tr{\gOne q\prth{\mathcal{N}+1}q}$, and Sections~\ref{sec_proof_thm_1} and \ref{sec_proof_thm_2} are devoted to controlling this term in different ways, leading to our two main theorems.

    \begin{proposition}\label{prop:gronwall start}
        Let $\gamma_d$ solve \eqref{eq:gammad dynamics} with normalized initial data $\gamma_d(0) \in \mathcal{L}^1\prth{\mathcal{F}}$ and $\varphi$ solve \eqref{eq:MFS} with normalized initial data $\varphi(0) \in \ell^2(\C)$. We define $p \coloneq \ket{\varphi}\bra{\varphi}$ and $q \coloneq 1 - p$. Then
        \begin{align}
            &\abs{\partial_t \Tr{\gOne q}} \nonumber\\
                &\le J\prth{\Tr{p \mathcal{N}} + 1}^{\frac{1}{2}}\prth{8 \Tr{p \mathcal{N}}^{\frac{1}{2}} \Tr{\gOne q}
                    + 4\Tr{\gOne q}^{\frac{1}{2}}\Tr{\gOne q\prth{\mathcal{N}+1}q}^{\frac{1}{2}}
                    + \frac{\Tr{p \mathcal{N}}^{\frac{1}{2}}}{d}}.
        \end{align}
    \end{proposition}
    \begin{proof}
    
    \textbf{Computation of the time derivative.} We introduce the self-adjoint operator
        \begin{align*}
            A \coloneq \prth{a^* - \overline{\alpha_\varphi}}\otimes\prth{a - \alpha_\varphi} + \prth{a - \alpha_\varphi}\otimes\prth{a^* - \overline{\alpha_\varphi}}.
        \end{align*}
        With \eqref{eq:gamma1 dynamics}, we start by computing
        \begin{align}
            i\partial_t \Tr{\gOne q}
            &= \Tr{\sbra{h^{\varphi}, \gOne }q} 
                - J \Tr{\sbra{A, \gTwo}q_1}
                + \Tr{\gOne \sbra{h^{\varphi}, q}} \nonumber\\
            &= J \Tr{\gTwo  \sbra{A, q_1}} \nonumber \\
            &= 2i J \text{Im}\sbra{\Tr{\gTwo  A q_1}}. \label{gamma q computation}
        \end{align}
        Inserting resolution of identities $1=p+q$, we get
        \begin{align*}
            &\Tr{\gTwo  A q_1} \\ &= 
            \Tr{\gTwo  p_1 p_2 A q_1 p_2}
                + \Tr{\gTwo  p_1 p_2 A q_1 q_2}
                + \Tr{\gTwo  p_1 q_2 A q_1 p_2}
                + \Tr{\gTwo  p_1 q_2 A q_1 q_2} \\
                &\quad+ \Tr{\gTwo  q_1 p_2 A q_1 p_2}
                + \Tr{\gTwo  q_1 p_2 A q_1 q_2}
                + \Tr{\gTwo  q_1 q_2 A q_1 p_2}
                + \Tr{\gTwo  q_1 q_2 A q_1 q_2}.
        \end{align*}
        Note that $q_1 p_2 A q_1 p_2$ and $q_1 q_2 A q_1 q_2$ are self adjoint and hence do not contribute to \ref{gamma q computation}. This is also the case for $q_1 p_2 A q_1 q_2$ and $q_1 q_2 A q_1 p_2$ which are each others complex conjugate. Furthermore, $p_1 p_2 A q_1 p_2 = 0$ by definition of $A$. Then, by symmetry, we see that $p_1 q_2 A q_1 p_2$ is also not contributing, since
        \begin{align*}
            \Tr{\gTwo  p_1 q_2 A q_1 p_2} = \Tr{\gTwo  q_1 p_2 A p_1 q_2} = \overline{\Tr{\gTwo  p_1 q_2 A q_1 p_2}}.
        \end{align*}
        Thus, we are left with
        \begin{align}
            i\partial_t \Tr{\gOne q} &=
                2i J \text{Im}\sbra{\Tr{\gTwo  p_1 p_2 A q_1 q_2}}
                + 2i J \text{Im}\sbra{\Tr{\gTwo  p_1 q_2 A q_1 q_2}}. \label{eq:gamma q decomposition}
        \end{align}
        
    \noindent \textbf{Estimation of the $p_1 p_2 A q_1 q_2$ term.}   
        Since $pq = 0$,
        \begin{align*}
            \Tr{\gTwo  p_1 p_2 A q_1 q_2} = \Tr{\gTwo  p_1 p_2 (a^*_1 a_2 + a_1 a^*_2) q_1 q_2},
        \end{align*}
        and by symmetry of $\gTwo $,
        \begin{align*}
            \Tr{\gTwo  p_1 p_2 a^*_1 a_2 q_1 q_2} = \Tr{\gTwo  p_1 p_2 a_1 a^*_2 q_1 q_2}.
        \end{align*}
        We notice that when an operator is taken in between $p_\varphi$ we can factorize the traces since $p_\varphi$ is a projection, for instance
        \begin{align*}
            \Tr{p_x \mathcal{N}_x p_x \sbra{\dots}} = \Tr{p \mathcal{N}} \Tr{p_x  \sbra{\dots}}.
        \end{align*}
        Then, we use Cauchy--Schwarz's inequality to estimate, for any $\epsilon >0$,
        \begin{align*}
            \abs{\Tr{\gTwo  p_1 p_2 A q_1 q_2}}
                &= \dfrac{1}{d \abs{\Lambda}} \abs{\sum_{\substack{x, y\in \Lambda \\ x \sim y}} \Tr{\gamma_d p_x p_y a^*_x a_y q_x q_y}} \\
                &\le \dfrac{1}{d \abs{\Lambda}} \sum_{x \in \Lambda} \abs{\Tr{q_x \gamma_d^{\frac{1}{2}} \cdot \gamma_d^{\frac{1}{2}} \sum_{y \in \Lambda, x \sim y}p_x p_y a^*_x a_y q_y}} \\
                &\le \dfrac{1}{2d \epsilon \abs{\Lambda}} \sum_{x \in \Lambda} \Tr{q_x \gamma_d} 
                    + \dfrac{\epsilon}{2 d \abs{\Lambda}} \sum_{x \in \Lambda}\Tr{\gamma_d \sum_{\substack{y \in \Lambda , x \sim y\\ z \in \Lambda , x \sim z}}  p_x p_y a^*_x a_y q_y q_z a^*_z a_x p_z p_x} \\
                &=\frac{1}{2d \epsilon} \Tr{\gOne q} + \epsilon \dfrac{\Tr{p \mathcal{N}}}{2 d \abs{\Lambda}} \sum_{\substack{x, y\in \Lambda \\ x \sim y}} \sum_{z\in\Lambda,x \sim z}\Tr{\gamma_d  p_x p_y a_y q_y q_z a^*_z p_z} \\
                &= \frac{1}{2d \epsilon} \Tr{\gOne q} 
                    + \epsilon\dfrac{\Tr{p \mathcal{N}}}{2 d \abs{\Lambda}} \sum_{\substack{x, y\in \Lambda \\ x \sim y}}\Tr{\gamma_d  p_x p_y a_y q_y a^*_y p_y} \\
                    &\quad + \epsilon\dfrac{\Tr{p \mathcal{N}}}{2 d \abs{\Lambda}} \sum_{\substack{x, y\in \Lambda \\ x \sim y}} \sum_{\substack{z\in\Lambda,x \sim z \\ z \neq y}}\Tr{q_y \gamma_d q_z p_x p_y a_y a^*_z p_z}.
        \end{align*}
        The last two summands can be estimated as
        \begin{align*}
            \sum_{\substack{x, y\in \Lambda \\ x \sim y}} \Tr{\gamma_d p_x p_y a_y q_y a^*_y p_y}
                &\le \sum_{\substack{x, y\in \Lambda \\ x \sim y}} \Tr{\gamma_d p_x p_y a_y a^*_y p_y}
                = \Tr{p(\mathcal{N} + 1)} \sum_{\substack{x, y\in \Lambda \\ x \sim y}} \Tr{\gamma_d p_x p_y} \\
                &\le 2d \abs{\Lambda} \prth{\Tr{p\mathcal{N}}+1},
        \end{align*}
        and
        \begin{align*}
             \sum_{\substack{x, y\in \Lambda \\ x \sim y}} \sum_{\substack{z\in\Lambda,x \sim z \\ z \neq y}}\Tr{q_y \gamma_d q_z p_x p_y a_y a^*_z p_z} 
                &\le \sum_{\substack{x, y\in \Lambda \\ x \sim y}} \sum_{\substack{z\in\Lambda,x \sim z \\ z \neq y}} \Tr{\gamma_d q_z p_x p_y a_y a^*_y p_y}^{\frac{1}{2}} \Tr{\gamma_d q_y p_z a_z a^*_z p_z}^{\frac{1}{2}} \\
                &= \Tr{p(\mathcal{N} + 1)} \sum_{\substack{x, y\in \Lambda \\ x \sim y}} \sum_{\substack{z\in\Lambda,x \sim z \\ z \neq y}} \Tr{\gamma_d q_z p_x p_y}^{\frac{1}{2}} \Tr{\gamma_d q_y p_z}^{\frac{1}{2}} \\
                &\le \prth{\Tr{p\mathcal{N}} +1} \sum_{\substack{x, y\in \Lambda \\ x \sim y}} \sum_{z\in\Lambda,x \sim z} \Tr{\gamma_d q_z}^{\frac{1}{2}} \Tr{\gamma_d q_y}^{\frac{1}{2}} \\
                &= \prth{\Tr{p\mathcal{N}} +1} \sum_{x\in\Lambda}\prth{\sum_{y\in\Lambda,x \sim y} \Tr{\gamma_d q_y}^{\frac{1}{2}}}^2 \\
                &\le 2d \prth{\Tr{p\mathcal{N}} +1}\sum_{\substack{x, y\in \Lambda \\ x \sim y}} \Tr{\gamma_d q_y} \\
                &= 4d^2\abs{\Lambda} \prth{\Tr{p\mathcal{N}} +1} \Tr{\gOne q}.
        \end{align*}
        Using these two estimates and then choosing $\epsilon^{-1} \coloneq 2d\Tr{p \mathcal{N}}^{\frac{1}{2}}\prth{\Tr{p \mathcal{N}} + 1}^{\frac{1}{2}}$, we obtain 
        \begin{align}
            \abs{\Tr{\gTwo  p_1 p_2 A q_1 q_2}}
                &\le \frac{1}{2d \epsilon} \Tr{\gOne q} 
                    + \epsilon \Tr{p \mathcal{N}}\prth{\Tr{p \mathcal{N}} + 1} \nonumber\\
                    &\quad+ 2d \epsilon \Tr{p \mathcal{N}}\prth{\Tr{p\mathcal{N}} +1} \Tr{\gOne q} \nonumber\\
                &=\Tr{p \mathcal{N}}^{\frac{1}{2}}\prth{\Tr{p \mathcal{N}} + 1}^{\frac{1}{2}}\prth{2 \Tr{\gOne q} + \dfrac{1}{2d}}.\label{eq:ppqq term}
        \end{align}
    
    \noindent \textbf{Estimation of the $p_1 q_2 A q_1 q_2$ term.}   
        Since $pq = 0$,
        \begin{align*}
            \Tr{\gTwo  p_1 q_2 A q_1 q_2} 
                &= \Tr{\gTwo  p_1 q_2 a^*_1 a_2 q_1 q_2}
                    + \Tr{\gTwo  p_1 q_2 a_1 a^*_2 q_1 q_2}
                    - \alpha_\varphi \Tr{\gTwo  p_1 q_2 a^*_1 q_1 q_2}\\
                    &\quad- \overline{\alpha_\varphi} \Tr{\gTwo  p_1 q_2 a_1 q_1 q_2}.
        \end{align*}
        We estimate
        \begin{align*}
            \abs{\Tr{\gTwo  p_1 q_2 a^*_1 a_2 q_1 q_2}}
                &\le \Tr{\gTwo  p_1 q_2 \mathcal{N}_1 p_1}^{\frac{1}{2}} \Tr{\gTwo  q_2q_1 a_2 a^*_2 q_2 }^{\frac{1}{2}} \\
                &= \Tr{p \mathcal{N}}^{\frac{1}{2}} \Tr{\gTwo  p_1 q_2}^{\frac{1}{2}} \Tr{\gTwo  q_1 q_2 \prth{\mathcal{N}_2 + 1}q_2}^{\frac{1}{2}} \\
                &\le \Tr{p \mathcal{N}}^{\frac{1}{2}} \Tr{\gOne q}^{\frac{1}{2}} \Tr{\gOne q\prth{\mathcal{N}+1}q}^{\frac{1}{2}}, \\
            \abs{\Tr{\gTwo  p_1 q_2 a^*_1 q_1 q_2}}
                &\le \Tr{\gTwo  p_1 q_2 \mathcal{N}_1 p_1}^{\frac{1}{2}} \Tr{\gTwo q_1 q_2} ^{\frac{1}{2}} \\
                &= \Tr{p \mathcal{N}}^{\frac{1}{2}}\Tr{\gTwo p_1 q_2} ^{\frac{1}{2}}  \Tr{\gTwo q_1 q_2} ^{\frac{1}{2}}
                \le \Tr{p \mathcal{N}}^{\frac{1}{2}} \Tr{\gOne q},
        \end{align*}
        and similarly
        \begin{align*}
            \abs{\Tr{\gTwo  p_1 q_2 a_1 a^*_2 q_1 q_2}}
                &\le \prth{\Tr{p \mathcal{N}}+1}^{\frac{1}{2}} \Tr{\gOne q}^{\frac{1}{2}} \Tr{\gOne q \mathcal{N}q}^{\frac{1}{2}}, \\
            \abs{\Tr{\gTwo  p_1 q_2 a_1 q_1 q_2}}
                &\le \prth{\Tr{p \mathcal{N}}+1}^{\frac{1}{2}} \Tr{\gOne q}.
        \end{align*}
        Inserting these estimates yields
        \begin{align}
            &\abs{\Tr{\gTwo  p_1 q_2 A q_1 q_2}} \nonumber\\
                &\quad\le 2 \prth{\Tr{p\mathcal{N}} + 1}^{\frac{1}{2}}\prth{\abs{\alpha_\varphi}\Tr{\gOne q} + \Tr{\gOne q}^{\frac{1}{2}}\Tr{\gOne q\prth{\mathcal{N}+1}q}^{\frac{1}{2}}} \nonumber\\
                &\quad\le 2 \prth{\Tr{p\mathcal{N}} + 1}^{\frac{1}{2}}\prth{\Tr{p \mathcal{N}}^{\frac{1}{2}}\Tr{\gOne q} + \Tr{\gOne q}^{\frac{1}{2}}\Tr{\gOne q\prth{\mathcal{N}+1}q}^{\frac{1}{2}}}.\label{eq:pqqq term}
        \end{align}
    
    \noindent \textbf{Conclusion.}    
        Inserting \eqref{eq:ppqq term} and \eqref{eq:pqqq term} into \eqref{eq:gamma q decomposition} we obtain
        \begin{align*}
            &\abs{\partial_t \Tr{\gOne q}} \\
                &\quad\le 2J \Tr{p \mathcal{N}}^{\frac{1}{2}}\prth{\Tr{p \mathcal{N}} + 1}^{\frac{1}{2}}\prth{2 \Tr{\gOne q} + \dfrac{1}{2d}} \\
                    &\qquad+  4J \prth{\Tr{p\mathcal{N}} + 1}^{\frac{1}{2}}\prth{\Tr{p\mathcal{N}}^{\frac{1}{2}}\Tr{\gOne q} + \Tr{\gOne q}^{\frac{1}{2}}\Tr{\gOne q\prth{\mathcal{N}+1}q}^{\frac{1}{2}}} \\
                &\quad= J\prth{\Tr{p \mathcal{N}} + 1}^{\frac{1}{2}}\prth{8 \Tr{p \mathcal{N}}^{\frac{1}{2}} \Tr{\gOne q}
                    + 4\Tr{\gOne q}^{\frac{1}{2}}\Tr{\gOne q\prth{\mathcal{N}+1}q}^{\frac{1}{2}}
                    + \frac{\Tr{p \mathcal{N}}^{\frac{1}{2}}}{d}}.
        \end{align*}
    \end{proof}

\section{Proof of Theorem~\ref{th:moment method}}\label{sec_proof_thm_1}

    In this Section we prove Theorem~\ref{th:moment method} by estimating the term $\Tr{\gOne q\prth{\mathcal{N}+1}q}$ from the Gronwall estimate in Proposition~\ref{prop:gronwall start} using a moment method.

    \begin{lem} \label{lem:iterated CS}
        Let $k\in \N$ and $\gamma, p \in \mathcal{L}^1(\ell^2(\C))$. We assume that $0\le \gamma \le 1$, that $p$ is a rank one projection and $p\mathcal{N}^k, \gamma \mathcal{N}^k \in \mathcal{L}^1(\ell^2(\C))$.
        Then
       \begin{align}
            \Tr{\gamma  q \mathcal{N}^k q}
                \le 2\Tr{\gamma  \mathcal{N}^k}
                    + 2\Tr{p \mathcal{N}^k}. \label{eq:removing q from moment}
        \end{align}
    \end{lem}
    \begin{proof}
        With the Cauchy--Schwarz inequality,
        \begin{align*}
            \Tr{\gamma  q \mathcal{N}^k q}
                &= \Tr{\gamma  \mathcal{N}^k} 
                    - \Tr{\gamma  p \mathcal{N}^k p}
                    - \Tr{\gamma  p \mathcal{N}^k q}
                    - \Tr{\gamma  q \mathcal{N}^k p} \\
                &\le \Tr{\gamma  \mathcal{N}^k}
                    - \Tr{\gamma  p \mathcal{N}^k p}
                    + 2\sqrt{\Tr{\gamma  p \mathcal{N}^k p}}\sqrt{\Tr{\gamma  q \mathcal{N}^k q}} \\
                &\le \Tr{\gamma  \mathcal{N}^k}
                    + \Tr{\gamma  p \mathcal{N}^k p}
                    + \dfrac{1}{2} \Tr{\gamma  q \mathcal{N}^k q},
        \end{align*}
        so
        \begin{align*}
            \Tr{\gamma  q \mathcal{N}^k q}
                \le 2\Tr{\gamma  \mathcal{N}^k}
                    + 2\Tr{\gamma  p \mathcal{N}^k p}
                \le 2\Tr{\gamma  \mathcal{N}^k}
                    + 2\Tr{p \mathcal{N}^k}.
        \end{align*}
    \end{proof}

\subsection{The Moment Method}

We will prove Theorem~\ref{th:moment method} by showing that the probability of having a large lattice site occupation outside the product state structure is small. We use the following basic Calculus estimates.

    \begin{lem}\label{Calculus_lemma}
        Let $(u_n)_{n\in\N} \subset \R_+$. Then
        \begin{align}
            \exists a > 0 ~\text{s.t.}~ \forall n \in \N, u_n \le e^{-\frac{n}{a}} \implies \forall k \in \N, \sum_{n \in \N} n^k u_n \le (1+a) a^{k}k!, \label{eq:decay =>}
        \end{align}
        and conversely, 
        \begin{align}
            \exists b > 0 ~\text{s.t.}~ \forall k \in \N, \sum_{n\in\N} n^k u_n \le b^k k! \implies \forall M \in \N, \sum_{n=M}^\infty (n+1) u_n \le (2 + 4b) e^{-\frac{M}{2b}}.
 \label{eq:decay <=}
        \end{align}
    \end{lem}
    \begin{proof}

    \textbf{Proof of \eqref{eq:decay =>}.}
        The function 
        \begin{align*}
            f_a:\matrx{\R_+ &\to&\R_+, x&\mapsto & x^k e^{-\frac{x}{a}}}
        \end{align*}
        is increasing up to $ak$ and decreasing afterwards. Thus, by series-integral comparison,
        \begin{align*}
            \sum_{n\in\N} f_a(n)  
                &\le \intr_{\R_+} f_a(x) dx + f_a\prth{\floor{ak}} + f_a\prth{\ceil{ak}}
                = a^k\prth{ak! + f_1\prth{a^{-1}\floor{ak}} + f_1\prth{a^{-1}\ceil{ak}}} \\
                &\le a^k\prth{ak! + 2f_1(k)}
                = a^k\prth{ak! + 2\prth{\frac{k}{e}}^k}.
        \end{align*}
        If $k\ge 1$, inserting the Stirling lower approximation
        \begin{align}
            \sqrt{2\pi k}\prth{\frac{k}{e}}^k \le k! \label{eq:Stirling}
        \end{align}
        yields
        \begin{align*}
            \sum_{n\in\N} n^k e^{-\frac{n}{a}} 
                \le a^k k! \prth{a + \sqrt{\frac{2}{\pi k}}}
                \le  \prth{1 + a} a^k k!.
        \end{align*}
        The statement also holds for $k=0$ since
        \begin{align*}
            \sum_{n\in\N} e^{-\frac{n}{a}} \le 1 +\intr_{\R_+}e^{-\frac{x}{a}}dx = 1 + a.
        \end{align*}

    \noindent \textbf{Proof of \eqref{eq:decay <=}.} 
        If $0 < a < \frac{1}{b}$ and $M \in \N$ we find
        \begin{align*}
            \sum_{n=M}^\infty (n+1) u_n e^{aM}
                &\le \sum_{n\in\N} (n+1) u_n e^{an}
                =\sum_{n, k \in \N}(n+1) \dfrac{(an)^k}{k!} u_n
                = \sum_{k\in\N} \dfrac{a^k}{k!} \prth{\sum_{n\in\N}n^{k+1}u_n + \sum_{n\in\N}n^k u_n}\\
                &\le \sum_{k\in\N} \dfrac{a^k}{k!}\prth{b^{k+1}(k+1)! + b^k k!} 
                = \sum_{k\in\N} \prth{b(k+1)(ab)^k + (ab)^k} \\
                &= \dfrac{b}{(1-ab)^2} + \frac{1}{1 - ab}.
        \end{align*}
        Choosing $a = \frac{1}{2b}$ yields
        \begin{align*}
            \sum_{n=M}^\infty (n+1) u_n 
                \le \frac{1-ab+b}{(1-ab)^2}e^{-aM} 
                = (2 + 4b) e^{-\frac{M}{2b}}.
        \end{align*}
    \end{proof}

    With this we can prove our first main theorem.

    \begin{proof}[Proof of Theorem~\ref{th:moment method}.]
    \textbf{Controlling $\Tr{\gOne q \prth{\mathcal{N}+1} \mathds{1}_{\mathcal{N} \ge M}q}$ with moments.}
        Let $k \in \N$. Applying \eqref{eq:decay =>} from Lemma~\ref{Calculus_lemma} first to $u_n \coloneq \Tr{p(0) \mathds{1}_{\mathcal{N} = n}}$ and then to $u_n \coloneq\Tr{\gamma_d^{(1)}(0) \mathds{1}_{\mathcal{N} = n}}$ while using the assumption \eqref{eq:decay hypo} from Theorem~\ref{th:moment method}, we obtain directly
        \begin{align}
            &\Tr{p(0) \mathcal{N}^k} \le c(1+a)a^k k!,\label{p_0_N_k_estimate} \\
            &\Tr{\gamma_d^{(1)}(0) \mathcal{N}^k} \le c(1+a)a^k k!.\label{gamma_0_N_k_estimate}
        \end{align}
        For $k \ge 1$, we use first \eqref{eq:removing q from moment} from Lemma~\ref{lem:iterated CS}, then the moment bounds \eqref{eq:mf moment bound 1} from Proposition~\ref{prop:moments bound} and \eqref{eq:BH moment bound} from Proposition~\ref{prop:gamma_d conservation laws}, then \eqref{p_0_N_k_estimate} and \eqref{gamma_0_N_k_estimate}, and Stirling's approximation \eqref{eq:Stirling}, and find
        \begin{align*}
            \sum_{n\in\N} n^k \Tr{\gOne(t)q\mathds{1}_{\mathcal{N} = n}q}
                &= \Tr{\gOne(t)q(t) \mathcal{N}^k q(t)} \\
                &\le 2\Tr{\gOne(t) \mathcal{N}^k}
                    + 2\Tr{p(t) \mathcal{N}^k} \\
                &\le 2\prth{\Tr{p(0) \mathcal{N}^k} + e^{-1}k^k}e^{2eJk\Tr{p(0) \mathcal{N}}^{\frac{1}{2}}t} \\
                    &\quad+ 2\prth{\Tr{\gOne(0) \mathcal{N}^k} + e^{-1}k^k}e^{2eJk t} \\
                &\le 2\prth{\Tr{p(0) \mathcal{N}^k} + \Tr{\gOne(0) \mathcal{N}^k} + 2e^{-1}k^k}e^{C_1 kt} \\
                &\le 4\prth{c(1+a)a^k k! + e^{-1}k^k}e^{C_1 kt} \\
                &\le 4\prth{c(1+a)a^k + \frac{e^{k-1}}{\sqrt{2\pi k}}} k!e^{C_1 kt} \\
                &\le 4 \prth{c (1+a) + e^{-1}}\prth{a^k + e^k} k!e^{C_1 kt} \\
                &\le 4 \prth{c (1+a) + e^{-1}} \prth{(a+e)e^{C_1 t}}^k k!.
        \end{align*}
        This is also valid for $k=0$, so \eqref{eq:decay <=} from Lemma~\ref{Calculus_lemma} implies
        \begin{align}
            \Tr{\gOne(t) q(t) \prth{\mathcal{N}+1} \mathds{1}_{\mathcal{N} \ge M}q(t)}
                &=\sum_{n=M}^\infty(n+1)\Tr{\gOne(t) q(t)\mathds{1}_{\mathcal{N} = n}q(t)} \nonumber \\
                &\le  4\prth{c (1+a) + e^{-1}}\prth{2+4(a+e)e^{C_1 t}} e^{-\frac{M}{2(a+e)}e^{-C_1 t}} \nonumber\\
                &\le C_2 e^{C_1 t -\frac{M}{2(a+e)}e^{-C_1 t}}. \label{eq:decay lem}
        \end{align}

    \noindent \textbf{Conclusion of the proof.}    
        Let $M \in \N^*$. We use the beginning of the Gronwall estimate from Proposition \ref{prop:gronwall start} while introducing a cutoff on $\mathcal{N}$, and then Proposition~\ref{prop:moments bound} to find, for any $\epsilon>0$,
        \begin{align*}
            &\abs{\partial_t \Tr{\gOne q}} \\
                &\le J C_3\left( 8 \Tr{p \mathcal{N}}^{\frac{1}{2}} \Tr{\gOne q}
                    + 4\Tr{\gOne q}^{\frac{1}{2}}\Tr{\gOne q\prth{\mathcal{N}+1}\prth{\mathds{1}_{\mathcal{N} < M} + \mathds{1}_{\mathcal{N} \ge M}}q}^{\frac{1}{2}}  \right. \\
                    &\qquad\qquad \left.+ \frac{\Tr{p \mathcal{N}}^{\frac{1}{2}}}{d} \right) \\
                &\le J C_3\prth{\prth{8 \Tr{p \mathcal{N}}^{\frac{1}{2}} + 4\sqrt{M}} \Tr{\gOne q} 
                    + 4\Tr{\gOne q}^{\frac{1}{2}}\Tr{\gOne q\prth{\mathcal{N}+1}\mathds{1}_{\mathcal{N} \ge M}q}^{\frac{1}{2}}
                    + \frac{\Tr{p \mathcal{N}}^{\frac{1}{2}}}{d}} \\
                &\le J C_3\prth{\prth{8\Tr{p \mathcal{N}}^{\frac{1}{2}} + 4\sqrt{M} + 4\epsilon^{-1}} \Tr{\gOne q} 
                    + \epsilon\Tr{\gOne q\prth{\mathcal{N}+1}\mathds{1}_{\mathcal{N} \ge M}q}
                    + \frac{\Tr{p \mathcal{N}}^{\frac{1}{2}}}{d}}.
        \end{align*}
        Next, we insert \eqref{eq:decay lem} and use the conservation of the mean-field number of particles (see \eqref{mf_N_conservation}). Then the choice $\epsilon \coloneq d^{-1}\, e^{\frac{M}{2(a+e)}e^{-C_1 t}}$ yields
        \begin{align*}
            &\abs{\partial_t \Tr{\gOne(t) q(t)}}  \\
                &\quad\le J C_3\prth{\prth{8\Tr{p(0) \mathcal{N}}^{\frac{1}{2}} + 4\sqrt{M} + 4\epsilon^{-1}} \Tr{\gOne(t) q(t)} + \epsilon C_2 e^{C_1 t -\frac{M}{2(a+e)}e^{-C_1 t}}
                    + \frac{\Tr{p(0) \mathcal{N}}^{\frac{1}{2}}}{d}} \\
                &\quad\le J C_3\prth{\prth{8\Tr{p(0) \mathcal{N}}^{\frac{1}{2}} + 4\sqrt{M} + 4de^{-\frac{M}{2(a+e)}e^{-C_1 t}}} \Tr{\gOne(t) q(t)} 
                    + \frac{C_2 e^{C_1 t} + \Tr{p(0) \mathcal{N}}^{\frac{1}{2}}}{d}}.
        \end{align*}
        Then we choose\footnote{
        Let us comment on the choice of the cutoff parameter. Optimizing in $M$ requires to solve, for $x\ge 0$,
        \begin{align*}
            \sqrt{x} = de^{-\frac{x}{2(a+e)}e^{-C_1 t}} 
              &\iff e^{\frac{x}{a+e}e^{-C_1 t}} x =d^2 
                     \iff e^{\frac{x}{a+e}e^{-C_1 t}} \frac{x}{a+e}e^{-C_1 t} = \frac{d^2}{a+e}e^{-C_1 t} \\
             &\iff \frac{x}{a+e}e^{-C_1 t} = W_0\prth{\frac{d^2}{a+e}e^{-C_1 t}}
                        \iff x = (a+e)e^{C_1t} W_0\prth{\frac{d^2}{a+e}e^{-C_1 t}},
        \end{align*}
        where $W_0$ is the principal branch of the Lambert $W$ function. Our choice of $M$ comes from the fact that
        \begin{align*}
            W_0(x) \eq_{x\to\infty} \ln\prth{\frac{x}{\ln(x)}} + o(1).
        \end{align*}
        }
        \begin{align*}
            M \coloneq \ceil{2(a+e)e^{C_1 t} \ln\prth{\frac{d}{\sqrt{\ln\prth{d+1}}}}}.
        \end{align*}
        Observing that for $d\ge 1 $ we have
        \begin{align*}
            \ln\prth{\frac{d}{\sqrt{\ln\prth{d+1}}}} \le \ln(d+1),
        \end{align*}
        this choice implies
        \begin{align*}
            \sqrt{M} + de^{-\frac{M}{2(a+e)}e^{-C_1 t}} 
                &\le \prth{2(a+e)e^{C_1 t} \ln\prth{\frac{d}{\sqrt{\ln\prth{d+1}}}} + 1}^{\frac{1}{2}} + \sqrt{\ln\prth{d+1}} \\
                &\le \prth{\sqrt{2(a+e)}e^{\frac{C_1}{2}t} +1} \sqrt{\ln(d+1)} + 1.
        \end{align*}
        Consequently, 
        \begin{align*}
            \abs{\partial_t \Tr{\gOne(t) q(t)}}
                &\le J C_3\Bigg(\prth{2C_4 + 4\prth{\sqrt{2(a+e)}e^{\frac{C_1}{2}t} +1} \sqrt{\ln(d+1)}} \Tr{\gOne(t) q(t)} \\
                    &\qquad\qquad + \frac{C_2 e^{C_1 t} + \Tr{p(0) \mathcal{N}}^{\frac{1}{2}}}{d}\Bigg).
        \end{align*}
        Noticing that the time dependent coefficients in the above expression are non-decreasing in time, we can use Gronwall's lemma to obtain
        \begin{align*}
            \Tr{\gOne(t) q(t)} 
                &\le \prth{\Tr{\gOne(0) q(0)} + \frac{C_2 e^{C_1 t} + \Tr{p(0) \mathcal{N}}^{\frac{1}{2}}}{d\prth{2C_4 + 4\prth{\sqrt{2(a+e)}e^{\frac{C_1}{2}t} +1} \sqrt{\ln(d+1)}}}} \\
                    &\quad e^{JC_3\prth{2C_4 + 4\prth{\sqrt{2(a+e)}e^{\frac{C_1}{2}t} +1} \sqrt{\ln(d+1)}}t}.
        \end{align*}
        Finally, using \eqref{eq:Tr norm estimate} from Lemma~\ref{lem_q_dens_mat} proves Theorem~\ref{th:moment method}.
    \end{proof}

    \section{Proof of Theorem~\ref{thm:excitation method}}\label{sec_proof_thm_2}

In this section we prove Theorem~\ref{thm:excitation method} using an energy estimate. Recall that the Bose--Hubbard Hamiltonian $H_d$ can be written as a sum of two time-dependent quantities,
 \[ H_d= \sum_{x\in \Lambda}h^{\varphi}_x(t) + \tilde H(t),\]
 where $\alpha_\varphi(t):=\langle \varphi(t), a\varphi(t)\rangle$. Here, $h^{\varphi}_x(t)$ is the mean-field operator from \eqref{eq:h mf}, i.e.,
\begin{align}
& h^{\varphi}_x(t):= -J\Big[\alpha_\varphi(t) a_x^* +\overline{ \alpha_\varphi(t)} a_x -\vert \alpha_\varphi(t) \vert^2 \Big]+(J -\mu) \mathcal{N}_x +\frac{U}{2} \mathcal{N}_x(\mathcal{N}_x-1),
 \end{align}
and $\tilde H (t)$, after inserting identities $p_x(t)+q_x(t) = 1 = p_y(t) + q_y(t)$, can be computed as 
\begin{align}
& \begin{aligned} \tilde H(t) :&= -\frac{J}{2d}  \sum_{<x,y>} \Big( p_x(t)p_y(t) K^{(2)}_{x,y} q_x(t)q_y(t) +p_x(t)q_y(t)  K^{(2)}_{x,y} q_x(t)p_y(t)\Big) +h.c. 
\\ & \quad  -\frac{J}{d}  \sum_{<x,y>}p_x(t)q_y(t)  K^{(3)}_{x,y}(t) q_x(t)q_y(t) +h.c.
\\ &  \quad   -\frac{J}{2d} \sum_{<x,y>}q_x(t)q_y(t)  K^{(4)}_{x,y}(t) q_x(t)q_y(t),\end{aligned}\label{expE}
\end{align}
where
\begin{align}
& K_{x,y}^{(2)}:=a^*_x a_y +a^*_ya_x,\label{kxy}
\\& K_{x,y}^{(3)}(t):=K^{(2)}_{x,y} -\alpha_\varphi(t) a^*_x -\overline{\alpha_\varphi(t)} a_x,\label{overlinekxy}
\\&{K}_{x,y}^{(4)}(t):= K_{x,y}^{(3)}(t) -\alpha_\varphi(t) a^*_y -\overline{\alpha_\varphi(t)} a_y +2 \vert \alpha_\varphi(t) \vert^2.\label{tildekxy}
\end{align}
Here, the superscript $i$ in the expression $K^{(i)}_{x,y}$ refers to the number of $q$'s that  accompany it in the expression of $\tilde H$ in \eqref{expE}. Note that $K^{(2)}_{x,y}$ does not depend on $t$ whereas the other terms $K^{(3)}_{x,y}$ and $K^{(4)}_{x,y}$ do through the term $\alpha_\varphi(t)$. 

For our proof we define the quantities
\begin{equation}\label{defoff}
f(t):= \frac{1}{\vert \Lambda \vert}  \left \langle \Psi_d(t) , \left(  H_d + \sum_{x\in \Lambda} \left (q_x(t)h^{\varphi}_x(t)q_x(t)- h^{\varphi}_x(t)+c q_x(t)\right) \right)  \Psi_d(t) \right \rangle
 \end{equation}
with $c>0$, and 
\begin{equation}\label{defg}
g(t):=  \frac{1}{\vert \Lambda \vert}  \sum_{x\in \Lambda} \left \langle \Psi_d(t), \left(  q_x(t)\mathcal{N}_x^2 q_x(t)  +q_x(t)\right)\Psi_d(t) \right \rangle.
\end{equation}
The idea of the proof is the following. In the Gronwall estimate from Proposition~\ref{prop:gronwall start}, the problematic term was $\frac{1}{\vert \Lambda \vert}  \sum_{x\in \Lambda} \left \langle \Psi_d(t), q_x(t)\mathcal{N}_x^2 q_x(t) \Psi_d(t) \right \rangle$. Hence, one might want to attempt to do a joint Gronwall argument for this and the original quantity $\frac{1}{\vert \Lambda \vert}  \sum_{x\in \Lambda} \left \langle \Psi_d(t),q_x(t)\Psi_d(t) \right \rangle$ that we want to estimate, i.e., a Gronwall argument for $g(t)$. However, if one computes the time derivative of $g(t)$, one finds higher and higher powers $q(t)\mathcal{N}^kq(t)$ that need to be controlled, so the Gronwall argument cannot be closed. The trick is to instead do a Gronwall argument for $f(t)$. Except for the $cq(t)$ term, $f(t)$ represents the energy of deviations from the lattice product state structure. The technical advantage for the Gronwall argument is that $\langle \Psi_d(t), H_d \Psi_d(t) \rangle$ is conserved, and 
\begin{equation*}
q(t) h^{\varphi}(t) q(t)  - h^{\varphi}(t)  = -h^{\varphi}(t) p(t)  - p(t) h^{\varphi}(t)  + p(t) h^{\varphi}(t) p(t) ,
\end{equation*}
so the $\mathcal{N}^2$ term from the interaction appears always together with at least one $p(t) $ projection. And all powers of $\mathcal{N}$ can be controlled when traced out against $p(t)$ due to Proposition~\ref{prop:moments bound}. Hence, we can close a Gronwall estimate for $f(t)$. Finally, one can prove that $Cg(t)-d^{-1} \leq f (t) \leq Cg(t)+d^{-1}$. Hence, $g(t)$ can be estimated in terms of its initial data and an error $d^{-1}$, which, together with Lemma~\ref{lem_q_dens_mat}, implies Theorem~\ref{thm:excitation method}. 

In the following, we start by proving the equivalence of $f(t)$ and $g(t)$ up to an error $d^{-1}$ in Section~\ref{sec_f_g_equivalence}. Then, in Section~\ref{sec_Gronwall_for_f}, we prove the Gronwall estimate for $f(t)$. We conclude with the proof of Theorem~\ref{thm:excitation method} in Section~\ref{sec_proof_conclusion_excitation_method}.

{\bf Notation}. In the following estimates, we use the quantities $ C >0$, $ C(J,\mu,U) >0 $, and $ \tilde{C}(t)>0$ with the following definitions:
\begin{itemize}
\item  $C$ is a positive constant that is independent of all parameters of the model.
\item $ C(J,\mu,U)$ is a positive constant that depends on the parameters $J,\mu,U$ only polynomially, and is independent of the initial conditions and time $t$.
\item $ \tilde{C}(t)$ is a positive quantity that may depend on $ C(J,\mu,U)$, the initial data $\langle \varphi(0), \mathcal{N}^j \varphi(0) \rangle$ for $j\leq 4$, and polynomially on time $t$.
\end{itemize}
For convenience, these quantities may change from one line to the next in the subsequent estimates.
Moreover, for simplicity of notation, we will often omit the explicit time dependence of the quantities
$h^\varphi=h^\varphi(t)$, $\tilde H=\tilde H(t)$, $\alpha_\varphi=\alpha_\varphi(t)$
, $q=q(t)$, $p=p(t)$, $\Psi_d=\Psi_d(t)$, $\varphi=\varphi(t)$,
$K_{x,y}^{(3)}=K_{x,y}^{(3)}(t)$, $K_{x,y}^{(4)}=K_{x,y}^{(4)}(t)$, and
$\mathcal{R}=\mathcal{R}(t)$, and indicate it only when relevant.
Whenever a quantity is evaluated at the initial time, we write its time argument explicitly, e.g., $\varphi(0)$, $f(0)$, etc.; otherwise, the corresponding symbol is understood to denote the time-dependent quantity.
Furthermore, an overdot denotes differentiation with respect to time.

\subsection{Equivalence of \texorpdfstring{$f$}{} and \texorpdfstring{$g$}{} }\label{sec_f_g_equivalence}

We start by presenting an estimate for a slightly modified $f$.
 
\begin{proposition}\label{epsestimate}
There exist $C>0$ such that for all $\epsilon>0$ we have 
\begin{equation}\label{esttildeH-N2} 
\begin{aligned} 
& \frac{1}{\vert \Lambda \vert} \left \vert  \left \langle \Psi_d, \left( \tilde H+\sum_{x\in \Lambda}  q_x \left( h^{\varphi}_x-\frac{U}{2} \mathcal{N}_x^2   \right) q_x \right) \Psi_d \right  \rangle  \right \vert 
 \\ &\quad \leq C\left(  1+J^2 + \left(J-\mu-\frac{U}{2}\right)^2\right) \left( 1+\frac{1}{\epsilon}+\langle \varphi(0), \mathcal{N} \varphi(0)\rangle^2 \right) \frac{1}{\vert \Lambda \vert}  \sum_{x\in \Lambda}\langle \Psi_d,  q_x \Psi_d\rangle
 \\ &\qquad +\epsilon \frac{1}{\vert \Lambda \vert}  \sum_{x\in \Lambda}\langle \Psi_d,  q_x\mathcal{N}_x^2 q_x \Psi_d\rangle   +\frac{1}{d}, 
\end{aligned}   
\end{equation}
\end{proposition}

\begin{proof}
Recalling the definition of $ \tilde H$ in \eqref{expE}, we find  
\begin{align}
 \frac{1}{\vert \Lambda \vert} \langle \Psi_d, \tilde H\Psi_d\rangle 
  & =  -\frac{J}{2d} \frac{1}{\vert \Lambda \vert}  \sum_{<x,y>} \langle  \Psi_d , p_xp_y K^{(2)}_{x,y} q_xq_y \Psi_d \rangle +h.c.\label{2qsymest} 
\\ &\quad  -\frac{J}{2d} \frac{1}{\vert \Lambda \vert} \sum_{<x,y>} \langle  \Psi_d ,p_xq_y  K^{(2)}_{x,y} q_xp_y \Psi_d \rangle +h.c.\label{2qest}  
\\ &\quad  -\frac{J}{d}  \frac{1}{\vert \Lambda \vert}   \sum_{<x,y>}  \langle  \Psi_d, p_xq_y  K^{(3)}_{x,y} q_xq_y \Psi_d \rangle   +h.c.\label{3qest}
\\ &\quad   -\frac{J}{2d} \frac{1}{\vert \Lambda \vert} \sum_{<x,y>}  \langle  \Psi_d, q_xq_y  K^{(4)}_{x,y} q_xq_y\Psi_d \rangle , \label{4qest}
\end{align}
where the terms $K^{(2)}_{x,y}$, $ K^{(3)}_{x,y}$ and $ K^{(4)}_{x,y}$ are defined in \eqref{kxy}, \eqref{overlinekxy} and \eqref{tildekxy}.
Let us estimate the above equation term by term. The $pp$-$qq$ term of \eqref{2qsymest} has already been estimated in the proof of Proposition~\ref{prop:gronwall start}. Here, we find it slightly more convenient to choose $\epsilon^{-1} \coloneq 2d\Tr{p \mathcal{N}}\prth{\Tr{p \mathcal{N}} + 1}$, so instead of \eqref{eq:ppqq term} we arrive at

\begin{equation}
\begin{aligned}
 \vert \eqref{2qsymest}\vert \leq \left( 1+ J^2\langle \varphi(0), (\mathcal{N}+1) \varphi(0)\rangle ^{2} \right)  \frac{1}{\vert \Lambda \vert}  \sum_{x\in \Lambda}\langle \Psi_d,  q_x \Psi_d\rangle+\frac{1}{d}.
\end{aligned}
\end{equation}
For the $pq$-$qp$ term of \eqref{2qest} we find, using Cauchy--Schwarz,
\[\begin{aligned}
| \eqref{2qest} | 
 &\leq 2 \left \vert  -\frac{J}{2d} \frac{1}{\vert \Lambda \vert} \sum_{x\in \Lambda} \sum_{\underset{x\sim y}{y\in \Lambda}}   \langle  a_x p_x q_y\Psi_d ,a_y p_y q_x\Psi_d \rangle \right \vert \\ 
 &  \leq \frac{J}{d} \frac{1}{\vert \Lambda \vert} \sum_{x\in \Lambda} \sum_{\underset{x\sim y}{y\in \Lambda}} \langle \varphi(0), \mathcal{N} \varphi(0)\rangle \Vert  q_x\Psi_d \Vert \Vert q_y \Psi_d\Vert \\
& \leq 2{J} \langle \varphi(0), \mathcal{N} \varphi(0)\rangle \frac{1}{\vert \Lambda \vert} \sum_{x\in \Lambda}   \langle \Psi_d,  q_x \Psi_d\rangle.
\end{aligned}
\]

The $pq$-$qq$ terms of \eqref{3qest} have already been estimated in the proof of Proposition~\ref{prop:gronwall start}. We find it convenient to introduce $\epsilon >0$, so continuing from \eqref{eq:pqqq term} and using Cauchy--Schwarz, we find

\begin{equation}
\begin{aligned}
\vert \eqref{3qest} \vert  \leq 2 \left( 3J+\frac{J^2}{\epsilon} +4 J\langle \varphi(0), \mathcal{N} \varphi(0)\rangle  \right) \frac{1}{\vert \Lambda \vert} \sum_{x\in \Lambda}  \langle \Psi_d,  q_x \Psi_d\rangle +2 \epsilon \frac{1}{\vert \Lambda \vert} \sum_{x\in \Lambda}  \langle \Psi_d,  q_x \mathcal{N}^2_x q_x \Psi_d\rangle    .
\end{aligned}
\end{equation}

Finally, the terms involving $a^*a$ in \eqref{4qest} can be estimated directly with Cauchy--Schwarz. We find
 \[\begin{aligned}
&\left \vert  -\frac{J}{2d} \frac{1}{\vert \Lambda \vert} \sum_{x\in \Lambda} \sum_{\underset{x\sim y}{y\in \Lambda}}  \langle  \Psi_d, q_xq_y a^*_x a_y q_xq_y\Psi_d \rangle \right \vert
\\ &\quad \leq  \frac{J}{2d} \frac{1}{\vert \Lambda \vert} \sum_{x\in \Lambda} \sum_{\underset{x\sim y}{y\in \Lambda}}  \sqrt{ \langle \Psi_d, q_y q_x \mathcal{N}_x q_x \Psi_d \rangle} \sqrt{\langle \Psi_d, q_x q_y \mathcal{N}_y q_y \Psi_d \rangle} 
\\ &\quad \leq  \frac{1}{2d} \frac{1}{\vert \Lambda \vert} \sum_{x\in \Lambda} \sum_{\underset{x\sim y}{y\in \Lambda}}  \sqrt{J \Vert q_x \Psi_d\Vert \Vert   \mathcal{N}_x q_x \Psi_d \Vert } \sqrt{J\Vert q_y  \Psi_d \Vert \Vert  \mathcal{N}_y q_y \Psi_d \Vert }
\\ &\quad \leq  \frac{1}{2d} \frac{1}{\vert \Lambda \vert} \sum_{x\in \Lambda} \sum_{\underset{x\sim y}{y\in \Lambda}}  \left( \frac{J^2}{2 \epsilon }\Vert q_x \Psi_d\Vert^2 + \frac{ \epsilon }{2} \Vert   \mathcal{N}_x q_x \Psi_d \Vert^2 \right)^{1/2} \left( \frac{J^2}{2 \epsilon }\Vert q_y  \Psi_d\Vert^2+ \frac{ \epsilon}{2} \Vert  \mathcal{N}_y q_y \Psi_d \Vert^2 \right)^{1/2}
\\ &\quad \leq  \left(  \frac{J^2}{2 \epsilon } \frac{1}{\vert \Lambda \vert} \sum_{x\in \Lambda} \Vert q_x \Psi_d\Vert^2 + \frac{ \epsilon }{2}\frac{1}{\vert \Lambda \vert} \sum_{x\in \Lambda}  \Vert   \mathcal{N}_x q_x \Psi_d \Vert^2 \right)^{1/2} \left(  \frac{J^2}{2 \epsilon } \frac{1}{\vert \Lambda \vert} \sum_{x\in \Lambda} \Vert q_x \Psi_d\Vert^2 + \frac{ \epsilon }{2}\frac{1}{\vert \Lambda \vert} \sum_{x\in \Lambda}  \Vert   \mathcal{N}_x q_x \Psi_d \Vert^2 \right)^{1/2}
\\ &\quad \leq   \frac{J^2}{2 \epsilon } \frac{1}{\vert \Lambda \vert} \sum_{x\in \Lambda} \langle \Psi_d,  q_x \Psi_d\rangle +  \frac{ \epsilon }{2}\frac{1}{\vert \Lambda \vert} \sum_{x\in \Lambda} \langle \Psi_d,q_x\mathcal{N}_x^2 q_x \Psi_d \rangle.
\end{aligned}\]
The terms involving $\alpha_\varphi$ can be estimated in the same way, using additionally that
\begin{equation}\label{alphabound}
\vert  {\alpha_\varphi}(t)\vert =\vert \langle \varphi(t), a \varphi(t) \rangle \vert \leq \Vert a \varphi(t)\Vert=\sqrt{\langle  \varphi(0), \mathcal{ N} \varphi(0)\rangle } .
\end{equation}
Combining these bounds yields
\begin{equation}
\begin{aligned}
\vert \eqref{4qest} \vert  \leq C \left( \frac{1}{\epsilon}+\frac{J^2}{\epsilon}+ J^2 \langle \varphi(0), \mathcal{N} \varphi(0)\rangle \right) \frac{1}{\vert \Lambda \vert} \sum_{x\in \Lambda}   \langle \Psi_d,  q_x \Psi_d\rangle  +C  \epsilon    \frac{1}{\vert \Lambda \vert} \sum_{x\in \Lambda}  \langle \Psi_d,q_x\mathcal{N}_x^2 q_x \Psi_d \rangle  .
\end{aligned}
\end{equation}
Thus, altogether, we get for some $C>0$ that
\begin{equation}\label{estE} 
\begin{aligned}\left \vert \frac{1}{\vert \Lambda \vert} \langle \Psi_d, \tilde H\Psi_d \rangle  \right \vert & \leq C \epsilon \frac{1}{\vert \Lambda \vert}  \sum_{x\in \Lambda}\langle \Psi_d,  q_x\mathcal{N}_x^2 q_x \Psi_d\rangle +\frac{1}{d} 
\\ &\quad +C\left( 1+J^2\right)\left( 1+\frac{1}{\epsilon}+\langle \varphi(0), \mathcal{N} \varphi(0)\rangle+\langle \varphi(0), \mathcal{N} \varphi(0)\rangle^2 \right) \frac{1}{\vert \Lambda \vert}  \sum_{x\in \Lambda}\langle \Psi_d,  q_x \Psi_d\rangle.
\end{aligned}   
\end{equation} 

Similarly, we can use Cauchy--Schwarz and again \eqref{alphabound} to show that for some $C>0$,
\begin{equation}
\begin{aligned}
&\left  \vert \frac{1}{\vert \Lambda \vert}  \sum_{x\in \Lambda}\langle \Psi_d,  q_x \left( h^{\varphi}_x-\frac{U}{2} \mathcal{N}_x^2   \right) q_x \Psi_d\rangle  \right \vert
\\&\quad=
\left \vert \frac{1}{\vert \Lambda \vert}  \sum_{x\in \Lambda}\left \langle \Psi_d,  q_x \left( -J {\alpha_\varphi} a^*_x -J\overline  {\alpha_\varphi} a_x +J\vert  {\alpha_\varphi}\vert^2+ \left(J- \mu  -\frac{U}{2}\right)\mathcal{N}_x    \right) q_x \Psi_d \right \rangle \right \vert 
\\ &\quad \leq C\left(   1+J^2+\left(J-\mu-\frac{U}{2}\right)^2\right) \left( \frac{1}{\epsilon}+\langle \varphi(0), \mathcal{N} \varphi(0)\rangle \right) \frac{1}{\vert \Lambda \vert}  \sum_{x\in \Lambda}\langle \Psi_d,  q_x \Psi_d\rangle 
  \\ &\qquad +\epsilon  \frac{1}{\vert \Lambda \vert}  \sum_{x\in \Lambda}\langle \Psi_d,  q_x\mathcal{N}_x^2 q_x \Psi_d\rangle .
\end{aligned}
\end{equation}
Combining both bounds yields \eqref{esttildeH-N2}.
\end{proof}

This proposition allows us to prove the equivalence of $f$ and $g$ up to an error $d^{-1}$.

\begin{proposition}\label{equiv}
Let $f$ and $g$ be defined as in \eqref{defoff} and \eqref{defg}.
Then, for $U>0$ and for some $C>0$, we have the equivalence
\begin{equation}\label{equivfandgammabeta}
\begin{aligned}
\frac{U}{4}  g -\frac{1}{d}\leq  f
 &\leq C\left( 1+J^2+U+\left(J-\mu-\frac{U}{2}\right)^2\right)\left( 1+\frac{1}{U}+\langle \varphi(0), \mathcal{N} \varphi(0)\rangle^2 \right)
g+\frac{1}{d}.
\end{aligned}
\end{equation}  
\end{proposition}

\begin{proof}
We start by proving the lower bound on $f$ from \eqref{equivfandgammabeta}. 
From Proposition \ref{epsestimate} we know that
\begin{equation}\label{geqestE}
\begin{aligned}
&  \frac{1}{\vert \Lambda \vert} \left \langle \Psi_d,  \left( \tilde H+ \sum_{x\in \Lambda} q_x \left(  h^{\varphi}_x-\frac{U}{2} \mathcal{N}_x^2   \right)  q_x        \right)  \Psi_d \right \rangle  \\
&\quad \geq -C\left(  1+J^2 + \left(J-\mu-\frac{U}{2}\right)^2\right) \left( 1+\frac{1}{\epsilon}+\langle \varphi(0), \mathcal{N} \varphi(0)\rangle^2 \right)  \frac{1}{\vert \Lambda \vert}  \sum_{x\in \Lambda}\langle \Psi_d,  q_x \Psi_d\rangle -\frac{1}{d} \\
&\qquad  -\epsilon \frac{1}{\vert \Lambda \vert}  \sum_{x\in \Lambda}\langle \Psi_d,  q_x\mathcal{N}_x^2 q_x \Psi_d\rangle.
\end{aligned}
\end{equation}
Hence, 
\[\begin{aligned}
f&= \frac{1}{\vert \Lambda \vert} \left \langle \Psi_d,  \left( \tilde H+ \sum_{x\in \Lambda} q_x \left(  h^{\varphi}_x-\frac{U}{2} \mathcal{N}_x^2   \right)  q_x        \right)  \Psi_d \right \rangle \\ 
&\quad +\frac{U}{2} \frac{1}{\vert \Lambda \vert}  \sum_{x\in \Lambda}\langle \Psi_d,  q_x \mathcal{N}_x^2    q_x \Psi_d\rangle 
+\frac{c}{\vert \Lambda \vert}  \sum_{x\in \Lambda}\langle \Psi_d,  q_x \Psi_d\rangle \\ 
& \geq \bigg(\frac{U}{2} -\epsilon \bigg) \frac{1}{\vert \Lambda \vert}  \sum_{x\in \Lambda}\langle \Psi_d,  q_x\mathcal{N}_x^2 q_x \Psi_d\rangle \\ 
&\quad  +\left(c-C\left( 1+J^2 + \left(J-\mu-\frac{U}{2}\right)^2\right)\left( 1+\frac{1}{\epsilon}+ \langle \varphi(0), \mathcal{N} \varphi(0)\rangle^2\right)\right) \frac{1}{\vert \Lambda \vert}  \sum_{x\in \Lambda}\langle \Psi_d,  q_x \Psi_d\rangle -\frac{1}{d}.
\end{aligned}\]
Then the lower bound on $f$ from \eqref{equivfandgammabeta} follows by choosing
\begin{equation}\label{cchoice}
 c=C\left(  1+J^2 + \left(J-\mu-\frac{U}{2}\right)^2\right)\left( 1+\frac{1}{\epsilon}+ \langle \varphi(0), \mathcal{N} \varphi(0)\rangle^2\right) +\frac{U}{4}, \qquad \epsilon=\frac{U}{4}.  
\end{equation}

For the upper bound on $f$ from \eqref{equivfandgammabeta}, note that
\begin{equation}\label{estqhq}
\begin{aligned}
& \left \vert \frac{1}{\vert \Lambda \vert}  \sum_{x\in \Lambda}\langle \Psi_d,  q_xh^{\varphi}_x q_x \Psi_d\rangle \right \vert
\\ &\quad =  \left \vert \frac{1}{\vert \Lambda \vert}  \sum_{x\in \Lambda}\left \langle \Psi_d,  q_x \left( -J {\alpha_\varphi} a^*_x -J\overline  {\alpha_\varphi} a_x +J\vert  {\alpha_\varphi}\vert^2 +(J-\mu) \mathcal{N}_x +\frac{U}{2} \mathcal{N}_x(\mathcal{N}_x-1)   \right) q_x \Psi_d \right \rangle \right \vert
 \\ &\quad \leq C \left(1+\vert J\vert +U+\left \vert J-\mu-\frac{U}{2}\right \vert  \right)  \left(1+ \langle \varphi(0), \mathcal{N} \varphi(0)\rangle\right) \frac{1}{\vert \Lambda \vert}  \sum_{x\in \Lambda}\langle \Psi_d,  ( q_x\mathcal{N}_x^2 q_x+q_x) \Psi_d\rangle .
 \end{aligned}
\end{equation}
Using this and the bound \eqref{estE} from the proof of Proposition~\ref{epsestimate} for $\epsilon=1$, the choice \eqref{cchoice} for the constant $c$ yields

\[\begin{aligned}
\vert f \vert &=
\left \vert \frac{1}{\vert \Lambda \vert} \langle \Psi_d, \tilde H\Psi_d\rangle + \frac{1}{\vert \Lambda \vert}  \sum_{x\in \Lambda}\langle \Psi_d,  q_xh^{\varphi}_x q_x \Psi_d\rangle + \frac{c}{\vert \Lambda \vert}  \sum_{x\in \Lambda}\langle \Psi_d,   q_x \Psi_d\rangle\right \vert  
\\ &\quad \leq C\left( 1+J^2+U+\left(J-\mu-\frac{U}{2}\right)^2\right) \left(1+\frac{1}{U}+ \langle \varphi(0), \mathcal{N} \varphi(0)\rangle^2\right) \frac{1}{\vert \Lambda \vert}  \sum_{x\in \Lambda}\langle \Psi_d,  ( q_x\mathcal{N}_x^2 q_x+q_x) \Psi_d\rangle
\\ &\qquad +\frac{1}{d}.
\end{aligned}\]
\end{proof}

\subsection{Proof of Gronwall Estimate for \texorpdfstring{$f$}{}}\label{sec_Gronwall_for_f}

In the computation of the time derivative of $f$ we need to control in particular $\dot {\tilde H}$, the time derivative of $\tilde{H}$. Its computation is straightforward but a bit lengthy. The key point is to write this time derivative in such a way that it contains the commutator $[\tilde H, q_xh^{\varphi}_xq_x-h^{\varphi}_x]$, which we will later use for cancelations.
\begin{proposition}\label{expdottildeH}
The expectation of  $\dot {\tilde H} $ can be written as 
\begin{equation}\label{derivativeE}
\begin{aligned}
\frac{1}{\vert \Lambda \vert}\langle \Psi_d , \dot {\tilde H}  \Psi_d \rangle&= -\frac{i}{\vert \Lambda \vert}  \sum_{x\in \Lambda} \langle \Psi_d , [\tilde H, q_xh^{\varphi}_xq_x-h^{\varphi}_x] \Psi_d \rangle+\mathcal{R}
\\&=\frac{J}{2d}\frac{i}{\vert \Lambda \vert}  \sum_{<x,y>} \langle \Psi_d , [\tilde H_{x,y}, q_xh^{\varphi}_xq_x-h^{\varphi}_x+q_yh^{\varphi}_yq_y-h^{\varphi}_y] \Psi_d \rangle+\mathcal{R},
\end{aligned}
\end{equation}
with $\tilde H_{x,y} $ refers to the terms in \eqref{expE} such that  $\tilde H =\sum_{<x,y>} \tilde H_{x,y}   $ and  where the term  $\mathcal{R}\equiv \mathcal{R}(t)$ is given by
\begin{align}
\mathcal{R}&:= -\frac{J}{d}\frac{i}{\vert \Lambda \vert}  \sum_{<x,y>} \left \langle \Psi_d , \Big( p_x h^{\varphi}_x p_y  K^{(2)}_{x,y}q_x q_y +q_xh^{\varphi}_xp_x q_y K^{(2)}_{x,y}p_x p_y\Big) \Psi_d \right \rangle +h.c. \label{sym} 
\\ &\quad   - \frac{J}{d}\frac{i}{\vert \Lambda \vert}  \sum_{<x,y>} \left \langle \Psi_d , q_y p_x\Big( h^{\varphi}_x +h^{\varphi}_y p_y\Big)  K^{(2)}_{x,y}q_x p_y  \Psi_d \right \rangle +h.c. \label{est2} 
 \\ 
&\quad 
 \begin{aligned}+ \frac{J}{d}\frac{i}{\vert \Lambda \vert}  \sum_{<x,y>} \Big \langle \Psi_d , q_y\Big(& (h^{\varphi}_xp_x  -p_x h^{\varphi}_x  - p_x h^{\varphi}_y p_y ) K^{(3)}_{x,y}
 \\ & +p_x K^{(3)} _{x,y}(p_x h^{\varphi}_x + p_y h^{\varphi}_y) \Big)q_x q_y  \Psi_d \Big  \rangle +h.c.  \end{aligned}\label{est3}
 \\ 
&\quad 
 + \frac{J}{d}\frac{i}{\vert \Lambda \vert}  \sum_{<x,y>} \langle \Psi_d , q_x q_y  K^{(4)}_{x,y} p_x h^{\varphi}_x q_x q_y \Psi_d \rangle +h.c.\label{est4}
 \\ 
 &\quad  - \frac{J}{2d}\frac{1}{\vert \Lambda \vert}  \sum_{<x,y>} \left \langle \Psi_d , \left( p_x q_y \dot{ K}^{(3)}_{x,y} q_x q_y+\frac{1}{2}q_x q_y\dot{ K}^{(4)}  _{x,y} q_x q_y \right) \Psi_d \right \rangle +h.c. . \label{est5}
\end{align}
\end{proposition}
\begin{proof}
We start by gathering some useful computations,
\begin{align}
& \dot{ \alpha}_\varphi = i \mu  \alpha_\varphi -i U \left \langle \varphi,\mathcal{N} a \varphi \right \rangle,\label{derivalpha}
\\ & \overline \alpha_\varphi \dot \alpha_\varphi  + \alpha_\varphi \dot{ \overline{\alpha_\varphi}}  = 2 U\Im   \left( \langle \varphi, \mathcal{N}a \varphi \rangle \overline {\alpha_\varphi}\right),
\\ & \dot h^{\varphi}_x = -J  \dot \alpha_\varphi a^*_x  -J \dot{ \overline {\alpha_\varphi}} a_x + 2 JU\Im    \left( \langle \varphi,\mathcal{N}a \varphi \rangle \overline {\alpha_\varphi}\right),
\\ &  \dot{K}^{(2)}_{x,y} =0, \\
& \dot{ K}^{(3)}_{x,y}= -\dot  \alpha_\varphi a^*_x - \dot{ \overline {\alpha_\varphi}} a_x, \label{derivKoverlineK}
\\ &  \dot K^{(4)}_{x,y}=-\dot \alpha_\varphi (a^*_x+a^*_y) -\dot{ \overline {\alpha_\varphi}} (a_x+a_y) +4 U\Im    \left( \langle \varphi, \mathcal{N}a\varphi \rangle \overline  {\alpha_\varphi}\right) \label{derivtildeK}.
\end{align}
Starting from the definition \eqref{expE} of $\tilde H$, using these relations and $i\dot  p_x=[h^{\varphi}_x,p_x] $,  $i\dot  q_x=[h^{\varphi}_x,q_x] $, we arrive at

\[\begin{aligned} \dot{ \tilde H} &= \frac{iJ}{2d}  \sum_{<x,y>} 
\bigg( [h^{\varphi}_x,p_x]p_y K^{(2)}_{x,y} q_xq_y + p_x [h^{\varphi}_y,p_y] K^{(2)}_{x,y} q_xq_y + p_xp_y K^{(2)}_{x,y} [h^{\varphi}_x,q_x] q_y + p_xp_y K^{(2)}_{x,y} q_x [h^{\varphi}_y,q_y] 
\\ &\qquad  + [h^{\varphi}_x,p_x]q_y  K^{(2)}_{x,y} q_xp_y +p_x[h^{\varphi}_y,q_y]   K^{(2)}_{x,y} q_xp_y+p_xq_y  K^{(2)}_{x,y} [h^{\varphi}_x,q_x]p_y+p_xq_y  K^{(2)}_{x,y} q_x [h^{\varphi}_y,p_y]
\\ &\qquad  +[h^{\varphi}_x,q_x]q_y  K^{(2)}_{x,y} p_xp_y+q_x[h^{\varphi}_y,q_y]   K^{(2)}_{x,y} p_xp_y+q_xq_y  K^{(2)}_{x,y} [h^{\varphi}_x,p_x]p_y+q_xq_y  K^{(2)}_{x,y} p_x [h^{\varphi}_y,p_y]
\\ &\qquad + [h^{\varphi}_x,q_x]p_y K^{(2)}_{x,y}   p_xq_y +q_x [h^{\varphi}_y,p_y] K^{(2)}_{x,y}   p_xq_y+q_xp_y K^{(2)}_{x,y}   [h^{\varphi}_x,p_x]q_y+q_xp_y K^{(2)}_{x,y}   p_x [h^{\varphi}_y,q_y]  \bigg)
\\ & +\frac{iJ}{d}  \sum_{<x,y>} \bigg( [h^{\varphi}_x,p_x]q_y K^{(3)}_{x,y} q_xq_y+p_x[h^{\varphi}_y,q_y]  K^{(3)}_{x,y} q_xq_y+p_xq_y K^{(3)}_{x,y} [h^{\varphi}_x,q_x]q_y+p_xq_y K^{(3)}_{x,y} q_x[h^{\varphi}_y,q_y] 
\\ &\qquad  +[h^{\varphi}_x,q_x]q_y K^{(3)}_{x,y}  p_xq_y+ q_x[h^{\varphi}_y,q_y]  K^{(3)}_{x,y}  p_xq_y+ q_xq_y K^{(3)}_{x,y} [h^{\varphi}_x,p_x]q_y+ q_xq_y K^{(3)}_{x,y}  p_x [h^{\varphi}_y,q_y]  \bigg)   
\\ &  +\frac{iJ}{2d} \sum_{<x,y>}\bigg([h^{\varphi}_x,q_x]q_y K^{(4)}_{x,y} q_xq_y+ q_x[h^{\varphi}_y,q_y]  K^{(4)}_{x,y} q_xq_y+ q_xq_y K^{(4)}_{x,y} [h^{\varphi}_x,q_x]q_y+ q_xq_y K^{(4)}_{x,y} q_x [h^{\varphi}_y,q_y] \bigg)
\\ & + \frac{J}{d}  \sum_{<x,y>} \bigg[  p_xq_y \big(\dot  \alpha_\varphi a^*_x -\dot{  \overline{ {\alpha_\varphi}}} a_x \big) q_xq_y+  q_xq_y \big(\dot  \alpha_\varphi a^*_x -\dot { \overline {\alpha_\varphi} }a_x \big)  p_xq_y\bigg] 
\\ & +\frac{J}{2d} \sum_{<x,y>} q_xq_y \bigg( \dot  \alpha_\varphi (a^*_x+a^*_y) +\dot{  \overline {\alpha_\varphi}} (a_x+a_y) -4 U\Im    ( \langle \varphi, \mathcal N a\varphi \rangle \overline  {\alpha_\varphi})\bigg)  q_xq_y.
\end{aligned}\] 
To obtain \eqref{derivativeE},  we isolate the first part on the right-hand side of \eqref{derivativeE} and define the rest as the remainder term $\mathcal{R}$.
\end{proof}

Next, we estimate the rest term in Proposition~\ref{expdottildeH}.

\begin{proposition}\label{estrest}
The  term $\mathcal{R}(t)$ in Proposition \ref{expdottildeH}  satisfies the bound
\[\vert\mathcal{R}(t) \vert \leq \tilde C(t) \left( \frac{1}{\vert \Lambda \vert} \sum_{x\in \Lambda} \left \langle \Psi_d(t) , (q_x(t) \mathcal{N}_x q_x(t)+q_x(t) ) \Psi_d(t) \right  \rangle+\frac{1}{d}\right) , \]
where 
\begin{equation}\label{C(t)}
\begin{aligned}
 \tilde C(t)& = C(J,\mu,U)\left(  1+\langle \varphi(0), \mathcal{N} \varphi(0)\rangle ^2 \right) \\  
 &\quad \Big( 1  +
\sum_{j=0}^{6}  \left(  8J {\langle \varphi(0), \mathcal{N} \varphi(0)\rangle }^{1/2}  \right)^j \left \langle \varphi(0),(\mathcal{N}+j)^{4-\frac{j}{2} }\varphi(0) \right \rangle \frac{t^j}{j!}  \Big),
\end{aligned}
\end{equation}
with  $C(J,\mu,U) > 0$ depending polynomially on the model parameters $J$, $\mu$ and  $U$.
\end{proposition}

\begin{proof}
We need to estimate each term in $\mathcal{R}$. We start by explaining in detail how to estimate one of the terms in \eqref{sym}. By Cauchy--Schwarz and Hölder's inequality we have
\[\begin{aligned}
& \bigg \vert \frac{J}{d\vert \Lambda \vert} \sum_{x\in \Lambda} \sum_{\underset{y\sim x}{y\in \Lambda}} \langle \Psi_d ,  p_x h^{\varphi}_x p_y  a^*_x a_y q_x q_y  \Psi_d \rangle \bigg \vert 
\\ & \leq  \frac{J}{d\vert \Lambda \vert} \sum_{x\in \Lambda}  \big\Vert a^*_x q_x \Psi_d \big\Vert \Bigg \Vert h^{\varphi}_x p_x \Bigg( \sum_{\underset{y\sim x}{y\in \Lambda}} q_y  a^*_yp_y \Bigg)  \Psi_d\Bigg \Vert
\\ & \leq \frac{2J}{\vert \Lambda \vert}  \sum_{x\in \Lambda}  \sqrt{\langle \varphi(0), (\mathcal{N}+1) \varphi(0)\rangle}  \sqrt{\langle \varphi, (h^{\varphi})^2  \varphi\rangle}   \big\Vert a^*_x q_x \Psi_d\big\Vert  
\Bigg(\frac{1}{2d} +\frac{1}{(2d)^2} \sum_{\underset{ z\sim x}{z\in \Lambda}}  \sum_{\underset{y\not =z,y\sim x}{y\in \Lambda}}   \Vert q_z \Psi_d \Vert  \Vert q_y \Psi_d \Vert \Bigg)^{\frac{1}{2}}
\\ & \leq\frac{2J}{\vert \Lambda \vert}  \sum_{x\in \Lambda}  \sqrt{\langle \varphi(0), (\mathcal{N}+1) \varphi(0)\rangle}  \sqrt{\langle \varphi, (h^{\varphi})^2  \varphi\rangle}   \big\Vert a^*_x q_x \Psi_d\big\Vert    \Bigg( \frac{1}{2d} +\frac{1}{4d}  \sum_{\underset{y\not =z, y\sim x}{y\in\Lambda}}  \Vert q_y \Psi_d \Vert^2 +  \frac{1}{4d} \sum_{\underset{z\sim x}{z\in \Lambda}}  \Vert q_z \Psi_d \Vert^2 \Bigg)^{\frac{1}{2}}
\\ & \leq CJ \sqrt{ \langle \varphi(0), (\mathcal{N}+1) \varphi(0)\rangle} \sqrt{\langle \varphi, (h^{\varphi})^2  \varphi\rangle}\Bigg( \frac{1}{\vert \Lambda \vert} \sum_x \Vert a^*_x q_x \Psi_d\Vert^2+   \frac{1}{d}
\\ & \qquad \qquad\qquad\qquad\qquad\qquad\qquad\qquad\qquad\quad  +\frac{1}{d}  \frac{1}{\vert \Lambda \vert} \underbrace{ \sum_{x\in \Lambda}\sum_{\underset{y\sim x}{y\in \Lambda}}  \Vert q_y \Psi_d \Vert^2}_{=2d\sum_{x\in \Lambda} \Vert q_x \Psi_d \Vert^2 } +  \frac{1}{d}\frac{1}{\vert \Lambda \vert} \underbrace{\sum_{x\in \Lambda} \sum_{\underset{z\sim x}{z\in \Lambda}}  \Vert q_z \Psi_d \Vert^2}_{=2d\sum_{x\in \Lambda}\Vert q_x \Psi_d \Vert^2} \Bigg)
\\ & \leq C J  \sqrt{\langle \varphi(0), (\mathcal{N}+1) \varphi(0)\rangle}  \sqrt{\langle \varphi, (h^{\varphi})^2  \varphi\rangle} \left(  \frac{1}{\vert \Lambda \vert} \sum_{x\in \Lambda} \langle \Psi_d,  (q_x\mathcal{N}_x q_x+q_x) \Psi_d\rangle+\frac{1}{d}\right) .
\end{aligned}\]
The other terms of \eqref{sym} can be estimated analogously, so we arrive at
\[\begin{aligned}
 \left \vert \eqref{sym} \right \vert 
 \leq CJ\left( 1+ \langle \varphi(0), (\mathcal{N}) \varphi(0)\rangle  \right) \sqrt{ \langle\varphi, (h^{\varphi})^2  \varphi\rangle}  \left(  \frac{1}{\vert \Lambda \vert} \sum_{x\in \Lambda} \langle \Psi_d,  (q_x\mathcal{N}_x q_x+q_x) \Psi_d\rangle+\frac{1}{d}\right).
\end{aligned}
\] 

In order to bound \eqref{est2}, we use Cauchy--Schwarz to find
\[\begin{aligned}
& \left \vert \frac{J}{d}\frac{1}{\vert \Lambda \vert}  \sum_{x\in \Lambda} \sum_{\underset{x\sim y}{y\in \Lambda}} \langle \Psi_d , q_y p_x h^{\varphi}_x   a^*_x a_y q_x p_y  \Psi_d \rangle \right \vert 
\\&\quad =\left  \vert \frac{J}{d}\frac{1}{\vert \Lambda  \big \vert}  \sum_{x\in \Lambda} \sum_{\underset{x\sim y}{y\in \Lambda}} \langle h^{\varphi}_x p_x a^*_y q_y\Psi_d ,    a^*_x  q_x p_y  \Psi_d \rangle \right \vert
\\ &\quad \leq \frac{J}{d}\frac{1}{\vert \Lambda  \big \vert}  \sum_{x\in \Lambda} \sum_{\underset{x\sim y}{y\in \Lambda}}\sqrt{\langle \varphi, (h^{\varphi})^2  \varphi\rangle } \left(\Vert q_x \Psi_d  \Vert^2 +\Vert \mathcal{N}_x^{1/2} q_x \Psi_d \Vert^2\right)^{1/2} \left(\Vert q_y \Psi_d  \Vert^2 +\Vert \mathcal{N}_y^{1/2} q_y \Psi_d \Vert^2\right)^{1/2}
\\ &\quad \leq  C J \sqrt{ \langle \varphi, (h^{\varphi})^2  \varphi\rangle}  \frac{1}{\vert \Lambda \vert} \sum_{x\in \Lambda} \langle \Psi_d, ( q_x\mathcal{N}_x q_x+q_x) \Psi_d\rangle.
 \end{aligned}\]
Estimating the other terms of \eqref{est2} in an analogous way, we get 
 \[\begin{aligned}
 \left \vert \eqref{est2} \right \vert 
 \leq CJ \left( 1+ \sqrt{\langle \varphi(0), \mathcal{N} \varphi(0)\rangle}  \right) \sqrt{\langle \varphi, (h^{\varphi})^2  \varphi\rangle}   \frac{1}{\vert \Lambda \vert} \sum_{x\in \Lambda} \langle \Psi_d,  (q_x\mathcal{N}_x q_x+q_x) \Psi_d\rangle.
\end{aligned}
\] 
To bound \eqref{est3}, we use Cauchy--Schwarz to estimate
\[\begin{aligned}
& \left \vert \frac{J}{d}\frac{1}{\vert \Lambda \vert}  \sum_{x\in \Lambda} \sum_{\underset{x\sim y}{y\in \Lambda}}
\langle \Psi_d ,q_y h^{\varphi}_x p_x  a^*_x a_y q_x q_y    \Psi_d \rangle \right \vert \\
&\quad= \left  \vert \frac{J}{d}\frac{1}{\vert \Lambda \vert}  \sum_{x\in \Lambda} \sum_{\underset{x\sim y}{y\in \Lambda}} 
\langle a^*_y q_y \Psi_d ,h^{\varphi}_x p_x  a^*_x  q_x q_y    \Psi_d \rangle \right  \vert
\\ &\quad \leq  \frac{J}{d}\frac{1}{\vert \Lambda \vert}  \sum_{x\in \Lambda} \sum_{\underset{x\sim y}{y\in \Lambda}}
\Vert  a^*_y q_y \Psi_d \Vert \Big(\langle \Psi_d, q_x a_x \underbrace{p_x ( h^{\varphi}_x)^2 p_x}_{=\langle \varphi ,(h^{\varphi})^2 \varphi \rangle  p_x}  a^*_xq_x  \Psi_d \rangle      \Big)^{\frac{1}{2}} 
\\ &\quad \leq  \frac{J}{d}\frac{1}{\vert \Lambda \vert}  \sum_{x\in \Lambda} \sum_{\underset{x\sim y}{y\in \Lambda}} 
\sqrt{ \langle \varphi, (h^{\varphi})^2  \varphi\rangle  }   \Big(\langle \Psi_d, q_y (\mathcal{N}_y +1)q_y \Psi_d \rangle      \Big)^{\frac{1}{2}}  \Big(\langle \Psi_d, q_x (\mathcal{N}_x +1)q_x \Psi_d \rangle      \Big)^{\frac{1}{2}} 
\\ &\quad \leq C J\sqrt{\langle \varphi, (h^{\varphi})^2  \varphi\rangle}  \frac{1}{\vert \Lambda \vert} \sum_{x\in \Lambda} \langle \Psi_d, ( q_x\mathcal{N}_x q_x+q_x) \Psi_d\rangle,
\end{aligned}\]
and similarly
\[\begin{aligned}
& \left \vert \frac{J}{d}\frac{1}{\vert \Lambda \vert}  \sum_{x\in \Lambda} \sum_{\underset{x\sim y}{y\in \Lambda}} \langle \Psi_d , q_y  p_x h^{\varphi}_y p_y  a^*_x a_y q_x q_y  \Psi_d \rangle \right \vert 
\\ &\quad \leq  \frac{J}{d}\frac{1}{\vert \Lambda \vert}  \sum_{x\in \Lambda} \sum_{\underset{x\sim y}{y\in \Lambda}} \sqrt{\langle \varphi(0), \mathcal{N} \varphi(0)\rangle} \Vert q_y \Psi_d\Vert \left( \left \langle \Psi_d ,q_x q_y a^*_y \underbrace{p_y (h^{\varphi}_y)^2 p_y}_{= \langle \varphi, (h^{\varphi})^2  \varphi\rangle  p_y} a_yq_y \Psi_d \right \rangle\right)^{1/2}
\\ &\quad =\frac{J}{d}\frac{1}{\vert \Lambda \vert} \sum_{x\in \Lambda} \sum_{\underset{x\sim y}{y\in \Lambda}} \sqrt{ \langle \varphi, (h^{\varphi})^2  \varphi\rangle } \sqrt{\langle \varphi(0), \mathcal{N} \varphi(0)\rangle} \Vert q_y \Psi_d\Vert \left( \left  \langle \Psi_d ,q_x q_y a^*_y p_y a_yq_y \Psi_d \right\rangle\right)^{1/2}
\\ &\quad \leq CJ \sqrt{\langle \varphi(0), \mathcal{N} \varphi(0)\rangle} \sqrt{   \langle \varphi, (h^{\varphi})^2  \varphi\rangle}     \frac{1}{\vert \Lambda \vert} \sum_{x\in \Lambda} \langle \Psi_d,  (q_x\mathcal{N}_x q_x+q_x) \Psi_d\rangle.
\end{aligned}\]
For \eqref{est4}, we directly find
\[\begin{aligned}
\left \vert \frac{J}{d}\frac{1}{\vert \Lambda \vert}  \sum_{<x,y>} \langle \Psi_d , q_x q_y a^*_xa_y p_x h^{\varphi}_x q_x q_y \Psi_d \rangle \right \vert 
 &= \left  \vert \frac{J}{d}\frac{1}{\vert \Lambda \vert}  \sum_{<x,y>} \langle h^{\varphi}_x p_x a_x q_x q_y\Psi_d , a_yq_x q_y \Psi_d \rangle \right  \vert 
\\ & \leq  CJ \sqrt{ \langle \varphi, (h^{\varphi})^2  \varphi\rangle }  \frac{1}{\vert \Lambda \vert} \sum_{x\in \Lambda} \langle \Psi_d,  q_x\mathcal{N}_x q_x \Psi_d\rangle.
\end{aligned}\]
Estimating the other terms in an analogous way, we obtain 
\[\begin{aligned}
 \left \vert \eqref{est2}+\eqref{est3}+\eqref{est4} \right \vert 
 \leq C J \left( 1+  \langle \varphi(0), \mathcal{N} \varphi(0)\rangle\right) \sqrt{\langle \varphi, (h^{\varphi})^2  \varphi\rangle } \frac{1}{\vert \Lambda \vert} \sum_{x\in \Lambda} \langle \Psi_d, ( q_x\mathcal{N}_x q_x+q_x) \Psi_d\rangle  .
\end{aligned}
\] 
Using
\[\begin{aligned}
&\vert \dot {\alpha_\varphi}  \vert \leq C \left( \vert \mu\vert \sqrt{\langle \varphi(0), \mathcal{N} \varphi(0)\rangle} +U \langle \varphi, (\mathcal{N}+1)^{3/2} \varphi\rangle\right) 
\\&\vert \dot{\overline{ {\alpha_\varphi}}}  {\alpha_\varphi} + \overline  {\alpha_\varphi} \dot \alpha_\varphi \vert \leq C U  \sqrt{\langle \varphi(0), \mathcal{N} \varphi(0)\rangle} \langle \varphi, (\mathcal{N}+1)^{3/2} \varphi\rangle,
\end{aligned}
\]
we furthermore find
\[\begin{aligned}
\vert \eqref{est5}\vert \leq & CJ \left( \vert \mu+U\vert \sqrt{\langle \varphi(0), \mathcal{N} \varphi(0)\rangle}+  U \langle \varphi, \mathcal{N}^{3/2} \varphi\rangle +U \sqrt{ \langle \varphi(0), \mathcal{N} \varphi(0)\rangle} \sqrt{\langle \varphi,( \mathcal{N}+1)^{3/2} \varphi\rangle} \right)\\ & \frac{1}{\vert \Lambda \vert} \sum_{x\in \Lambda} \langle \Psi_d, \big(  q_x\mathcal{N}_x q_x+q_x\big) \Psi_d\rangle,
\end{aligned}\]
Combining all estimates, we arrive at the bound
\[\begin{aligned}
 \left \vert \mathcal{R} \right \vert 
 \leq  C(J,\mu,U) \left(  1+  \langle \varphi(0), \mathcal{N} \varphi(0)\rangle \right)  \sqrt{\langle \varphi, (h^{\varphi})^2  \varphi\rangle} \left(  \frac{1}{\vert \Lambda \vert} \sum_{x\in \Lambda} \langle \Psi_d, ( q_x\mathcal{N}_x q_x+q_x) \Psi_d\rangle  +\frac{1}{d} \right),
\end{aligned}
\] 
where $C(J,\mu,U)$ is a polynomial in $J$, $\mu$ and $U$.
We also have by Cauchy--Schwarz
\begin{equation}\label{h2est}
\begin{aligned}
 \langle \varphi, (h^{\varphi})^2  \varphi\rangle  &= \left \langle \varphi, \left( -J\big( {\alpha_\varphi} a^* +\overline  {\alpha_\varphi} a -\vert  {\alpha_\varphi} \vert^2 \big) +(J-\mu) \mathcal{N} +\frac{U}{2} \mathcal{N}(\mathcal{N}-1)\right)^2 \varphi \right \rangle 
 \\ &\leq  C(J,\mu,U)\left(  1+\langle \varphi(0), \mathcal{N} \varphi(0)\rangle ^2 \right) \left( 1+  \langle \varphi, \mathcal{N}^4 \varphi\rangle  \right).
\end{aligned}
\end{equation}
The proposition is proven by using the propagation bound \eqref{eq:mf moment bound 2} from Proposition~\eqref{prop:moments bound} for $k=4$, since then
\begin{equation}
\begin{aligned}
 \sqrt{ \langle \varphi, (h^{\varphi})^2  \varphi\rangle}  &
 \leq  C(J,\mu,U)\left(  1+\langle \varphi(0), \mathcal{N} \varphi(0)\rangle  \right) \left( 1+  \sqrt{\langle \varphi, \mathcal{N}^4 \varphi\rangle}  \right) 
 \\ &\leq  C(J,\mu,U)\left(  1+\langle \varphi(0), \mathcal{N} \varphi(0)\rangle  \right)  
\\&\quad \Bigg( 1+ \sum_{j=0}^{6} \frac{ \left(  8J {\langle \varphi(0), \mathcal{N} \varphi(0)\rangle }^{1/2} t \right)^j}{j!}   \left \langle \varphi(0),(\mathcal{N}+j)^{4-\frac{j}{2} }\varphi(0) \right \rangle \Bigg).
\end{aligned}
\end{equation}
\end{proof}

With this proposition we can now prove a Gronwall estimate for $f$.

\begin{proposition}\label{granwalllemmaf}
For $f$ as defined in \eqref{defoff}, we have for all $t\in \mathbb{R}$,
\begin{equation}\label{Granwallf}
f(t)\leq e^{\int_0^t\tilde C(s)ds} f(0)+\frac{1}{d}  \int_0^t \left(1+\tilde C(s) \right) e^{\int_s^t\tilde C(r)dr } ds,
\end{equation}
with
\begin{equation}\label{tildec2}
 \begin{aligned}
 \tilde C(t)  &=   \frac{C(J,\mu,U)}{U} \left(1+\frac{1}{U}+\langle \varphi(0),\mathcal N \varphi(0) \rangle ^2\right) 
 \\ &\qquad \Bigg(  1 +
\sum_{j=0}^{6}  \left(  8J {\langle \varphi(0), \mathcal{N} \varphi(0)\rangle }^{1/2}  \right)^j \left \langle \varphi(0),(\mathcal{N}+j)^{4-\frac{j}{2} }\varphi(0) \right \rangle \frac{t^j}{j!}  \Bigg),
\end{aligned}
\end{equation}
where  $C(J,\mu,U)>0$ is a polynomial in $J$, $\mu$ and $U$.
\end{proposition}

\begin{proof}
Using  $ \dot H_d= 0 = \sum_{x\in \Lambda} \dot h^{\varphi}_x +\dot{\tilde H}$, the time derivative of $f$ can be computed as 
\[\begin{aligned}
\dot f&= \frac{i}{\vert \Lambda \vert} \sum_{x\in \Lambda}  \langle \Psi_d , [H_d, H_d+ q_xh^{\varphi}_xq_x-h^{\varphi}_x +cq_x] \Psi_d \rangle+\frac{1}{\vert \Lambda \vert} \sum_{x\in \Lambda} \langle \Psi_d ,  (q_x \dot h^{\varphi}_xq_x-\dot h^{\varphi}_x) \Psi_d \rangle
\\ &\quad -\frac{i}{\vert \Lambda \vert} \sum_{x\in \Lambda} \langle \Psi_d, [h^{\varphi}_x, q_xh^{\varphi}_xq_x +cq_x] \Psi_d\rangle
\\ &=\frac{i}{\vert \Lambda \vert} \sum_{x\in \Lambda} \langle \Psi_d , [\tilde H, q_xh^{\varphi}_xq_x-h^{\varphi}_x+cq_x] \Psi_d \rangle+\frac{1}{\vert \Lambda \vert} \sum_{x\in \Lambda} \langle \Psi_d ,  (q_x \dot h^{\varphi}_xq_x+\dot{ \tilde H}) \Psi_d \rangle.
\end{aligned}\]

Using Proposition \ref{expdottildeH} we get 
\[\begin{aligned}
\dot f&= \frac{i}{\vert \Lambda \vert} \sum_{x\in \Lambda} \langle \Psi_d, [\tilde H, cq_x] \Psi_d \rangle+\frac{1}{\vert \Lambda \vert}\sum_{x\in \Lambda} \langle \Psi_d ,  q_x \dot h^{\varphi}_xq_x \Psi_d \rangle+\mathcal{R}.
\end{aligned}\] 
For the first two terms of this expression, we find
\[\begin{aligned}
&\left \vert \frac{1}{\vert \Lambda \vert} \sum_{x\in \Lambda}\langle \Psi_d ,  q_x \dot h^{\varphi}_xq_x \Psi_d \rangle \right \vert
\\ &\quad =\left \vert \frac{1}{\vert \Lambda \vert} \left  \langle \Psi_d ,  q_x \left( -J \dot  \alpha_\varphi a^*_x  -J \dot{ \overline  {\alpha_\varphi}} a_x + 2JU\Im    ( \langle \varphi, \mathcal{N}a\varphi \rangle \overline  {\alpha_\varphi}) \right) q_x \Psi_d \right \rangle \right \vert
\\ &\quad  \leq  C(J,\mu,U ) \left( 1+\langle \varphi(0),\mathcal N \varphi(0) \rangle  \right) \left(1+  \sqrt{\langle \varphi,\mathcal{ N}^2 \varphi \rangle }\right)  \frac{1}{\vert \Lambda \vert} \sum_{x\in \Lambda} \left \langle \Psi_d, \big(  q_x\mathcal{N}_x^2 q_x+q_x\big) \Psi_d \right \rangle,
\end{aligned}\]
and
\[\begin{aligned}
&\left \vert \frac{1}{\vert \Lambda \vert}\sum_{x\in \Lambda}  \langle \Psi_d ,  [\tilde H,cq_x]  \Psi_d \rangle \right \vert
\\ &\quad  \leq  C(J,\mu,U) \left(1+\frac{1}{U}+\langle \varphi(0),\mathcal N \varphi(0) \rangle ^2\right)  \frac{1}{\vert \Lambda \vert} \sum_{x\in \Lambda} \left \langle \Psi_d,  (q_x\mathcal{N}_x^2 q_x +q_x)\Psi_d \right \rangle+\frac{1}{d}.
\end{aligned}\]
These two estimates, together with the estimate on $\mathcal{R}$ from Proposition~\ref{estrest} and the equivalence of $f$ and $g$ up to an error $d^{-1}$ from Proposition~\ref{equiv} imply

 \[\left \vert \frac{d}{dt}f(t) \right \vert \leq  \tilde C( t)\left( f(t)+\frac{1}{d} \right) +\frac{1}{d} , \]
where $\tilde C(t)$ depends on the initial data and on the other parameters of our model as defined in \eqref{tildec2}.
With Gronwall's lemma we arrive at \eqref{Granwallf}.

\end{proof}

\subsection{Conclusion of the Proof}\label{sec_proof_conclusion_excitation_method}

We combine the above results to prove our second main result.

\begin{proof}[Proof of Theorem~\ref{thm:excitation method}]

We use the equivalence of $f(t) $ and $g(t) $ up to an error $d^{-1}$ from Proposition~\ref{equiv}, and the Gronwall estimate for $f(t) $ from Proposition~\ref{granwalllemmaf} to find
\begin{equation}
\begin{aligned}
 &\frac{1}{\vert \Lambda \vert} \sum_{x\in \Lambda} \langle \Psi_d(t) , q_x(t) \Psi_d (t) \rangle \\
&\quad\leq \frac{1}{\vert \Lambda \vert} \sum_{x\in \Lambda} \langle \Psi_d (t)  , \big(q_x(t)  \mathcal{N}_x^2 q_x (t) +q_x (t) \big) \Psi_d (t)  \rangle
 \\&\quad \leq  \frac{4}{U} \left(  f (t) +\frac{1}{d} \right) 
  \\&\quad \leq\frac{4}{U}  e^{\int_0^t\tilde C(s)ds} f(0)  +  \frac{1}{ d} \left(   \frac{4}{U}+ \frac{4}{U} \int_0^t \left(1+\tilde C(s) \right) e^{\int_s^t\tilde C(r)dr } ds \right) 
 \\&\quad \leq  \frac{C}{U}\left( 1+J^2+U+\left(J-\mu-\frac{U}{2}\right)^2\right)\left( 1+\frac{1}{U}+\langle \varphi(0), \mathcal{N} \varphi(0)\rangle^2 \right) e^{\int_0^t\tilde C(s)ds} 
 \\&\qquad \frac{1}{\vert \Lambda \vert}  \sum_{x\in \Lambda} \left  \langle \Psi_d(0), (  q_x(0)\mathcal{N}_x^2 q_x(0)  +q_x(0))\Psi_d(0) \right \rangle  
+  \frac{1}{ d} \frac{4}{U}\left(  1+ e^{\int_0^t\tilde C(s)ds}+ \int_0^t (1+\tilde C(s)) e^{\int_s^t\tilde C(r)dr } ds\right),\label{q_estimate_proof_thm_2}
\end{aligned}
\end{equation} 
where $\tilde C (t)$ is defined in \eqref{tildec2}. 

Now note that since ${\rm Tr}(p(0)\mathcal{N}^4)  \leq C$, we get that $\tilde C(t)$ satisfies
\[ \tilde C(t) \leq C(J,\mu,U) \left(1+\sum_{j=1}^6 t^j\right) \]
where $C(J,\mu,U)>0$ depends polynomially on the parameters of our model $J$, $\mu$ and $U$. Thus, \eqref{q_estimate_proof_thm_2} can be estimated as
\begin{equation}\label{detailed}
\begin{aligned}
&\frac{1}{\vert \Lambda \vert} \sum_{x\in \Lambda} \langle \Psi_d(t) , q_x(t) \Psi_d (t) \rangle  \\
&\quad\leq  \frac{1}{ d} \frac{1}{U}  + C(J,\mu,U) e^{C(J,\mu,U) \sum_1^7 |t|^j}\left(1+\frac{1}{U^2}\right)\left(  \frac{1}{\vert \Lambda \vert}  \sum_{x\in \Lambda} \left  \langle \Psi_d(0), \left(  q_x(0)\mathcal{N}_x^2 q_x(0)  +q_x(0)\right)\Psi_d(0) \right \rangle +\frac{1}{d}\right),
 \end{aligned}
\end{equation}
and the Theorem follows from using \eqref{eq:Tr norm estimate} from Lemma~\ref{lem_q_dens_mat}.
\end{proof}